\documentclass[onecolumn,compsoc,conference,11pt]{IEEEtran}

\usepackage[latin1]{inputenc}

\usepackage[english]{babel}

\usepackage{amsthm}

\usepackage[olditem]{paralist}

\usepackage{hyperref}

\usepackage{cite}

\usepackage[final]{microtype}

\usepackage{stringstrings}

\usepackage[font = small]{caption}

\usepackage{cutwin}

\hypersetup
  {
  pdfsubject    = {},
  pdfkeywords   = {},
  pdfproducer   = {},
  pdfcreator    = {},
  pdfpagemode   = {UseNone},
  pdfstartview  = {FitH}
  }

\usepackage{etoolbox}

\AtEndPreamble
  {

  \usepackage{changebar}
  \usepackage[pagewise,mathlines]{lineno}

  \renewcommand{\baselinestretch}{1.00}

  }

\usepackage{xargs}
\usepackage{xspace}
\usepackage{xstring}
\usepackage{boolexpr}

\usepackage[cmex10]{mathtools}
\usepackage{amsxtra}
\usepackage{amssymb}
\usepackage{amsfonts}
\usepackage{mathrsfs}
\usepackage{mathdots}
\usepackage{stmaryrd}
\usepackage{latexsym}
\usepackage{wasysym}

\usepackage{textcomp}

\newcommand{\argemp}[2]
  {\if&#1&\else#2\fi}

\newcommand{\argdef}[2]
  {\if&#1&#2\else#1\fi}

\newcommand{\argint}[3]
  {\if&#2&\else#1#2#3\fi}

\newcommand{\argext}[3]
  {\if&#1&#3\else#1\if&#3&\else#2#3\fi\fi}

\newcommandx{\argsubsup}[3][2=, 3=]
  {\def\argsubscript{{#2}}\def\argsuperscript{{#3}}#1}

\newcommandx{\argind}[9][2=, 3=, 4=, 5=, 6=, 7=, 8=, 9=]
  {%
  \switch[#1=]%
    \case{0}#2%
    \case{1}#3%
    \case{2}#4%
    \case{3}#5%
    \case{4}#6%
    \case{5}#7%
    \case{6}#8%
    \case{7}#9%
    \otherwise\ensuremath{\clubsuit}%
  \endswitch%
  }

\newcommand{\arga}[1]
  {#1}

\newcommand{\argb}[2]
  {\argext{\arga{#1}}{, \allowbreak}{#2}}

\newcommand{\argc}[3]
  {\argext{\argb{#1}{#2}}{, \allowbreak}{#3}}

\newcommand{\argd}[4]
  {\argext{\argc{#1}{#2}{#3}}{, \allowbreak}{#4}}

\newcommand{\arge}[5]
  {\argext{\argd{#1}{#2}{#3}{#4}}{, \allowbreak}{#5}}

\newcommand{\argf}[6]
  {\argext{\arge{#1}{#2}{#3}{#4}{#5}}{, \allowbreak}{#6}}

\newcommand{\argg}[7]
  {\argext{\argf{#1}{#2}{#3}{#4}{#5}{#6}}{, \allowbreak}{#7}}

\newcommand{\argh}[8]
  {\argext{\argg{#1}{#2}{#3}{#4}{#5}{#6}{#7}}{, \allowbreak}{#8}}

\newcommand{\argi}[9]
  {\argext{\argh{#1}{#2}{#3}{#4}{#5}{#6}{#7}{#8}}{, \allowbreak}{#9}}

\newcommand{\argj}[9]
  {%
  \def\valarga{#1}%
  \def\valargb{#2}%
  \def\valargc{#3}%
  \def\valargd{#4}%
  \def\valarge{#5}%
  \def\valargf{#6}%
  \def\valargg{#7}%
  \def\valargh{#8}%
  \def\valargi{#9}%
  \argauxj%
  }
\newcommand{\argauxj}[1]
  {%
  \argext%
    {
    \argi
      {\valarga} {\valargb} {\valargc} {\valargd} {\valarge} {\valargf}
      {\valargg} {\valargh} {\valargi}
    }
    {, \allowbreak}{#1}%
  }

\newcommand{\argk}[9]
  {%
  \def\valarga{#1}%
  \def\valargb{#2}%
  \def\valargc{#3}%
  \def\valargd{#4}%
  \def\valarge{#5}%
  \def\valargf{#6}%
  \def\valargg{#7}%
  \def\valargh{#8}%
  \def\valargi{#9}%
  \argauxk%
  }
\newcommand{\argauxk}[2]
  {\argext{\argauxj{#1}}{, \allowbreak}{#2}}

\newcommand{\argl}[9]
  {%
  \def\valarga{#1}%
  \def\valargb{#2}%
  \def\valargc{#3}%
  \def\valargd{#4}%
  \def\valarge{#5}%
  \def\valargf{#6}%
  \def\valargg{#7}%
  \def\valargh{#8}%
  \def\valargi{#9}%
  \argauxl%
  }
\newcommand{\argauxl}[3]
  {\argext{\argauxk{#1}{#2}}{, \allowbreak}{#3}}

\newcommand{\argm}[9]
  {%
  \def\valarga{#1}%
  \def\valargb{#2}%
  \def\valargc{#3}%
  \def\valargd{#4}%
  \def\valarge{#5}%
  \def\valargf{#6}%
  \def\valargg{#7}%
  \def\valargh{#8}%
  \def\valargi{#9}%
  \argauxm%
  }
\newcommand{\argauxm}[4]
  {\argext{\argauxl{#1}{#2}{#3}}{, \allowbreak}{#4}}

\newcommand{\argn}[9]
  {%
  \def\valarga{#1}%
  \def\valargb{#2}%
  \def\valargc{#3}%
  \def\valargd{#4}%
  \def\valarge{#5}%
  \def\valargf{#6}%
  \def\valargg{#7}%
  \def\valargh{#8}%
  \def\valargi{#9}%
  \argauxn%
  }
\newcommand{\argauxn}[5]
  {\argext{\argauxm{#1}{#2}{#3}{#4}}{, \allowbreak}{#5}}

\newcommand{\argo}[9]
  {%
  \def\valarga{#1}%
  \def\valargb{#2}%
  \def\valargc{#3}%
  \def\valargd{#4}%
  \def\valarge{#5}%
  \def\valargf{#6}%
  \def\valargg{#7}%
  \def\valargh{#8}%
  \def\valargi{#9}%
  \argauxo%
  }
\newcommand{\argauxo}[6]
  {\argext{\argauxn{#1}{#2}{#3}{#4}{#5}}{, \allowbreak}{#6}}

\newcommand{\argp}[9]
  {%
  \def\valarga{#1}%
  \def\valargb{#2}%
  \def\valargc{#3}%
  \def\valargd{#4}%
  \def\valarge{#5}%
  \def\valargf{#6}%
  \def\valargg{#7}%
  \def\valargh{#8}%
  \def\valargi{#9}%
  \argauxp%
  }
\newcommand{\argauxp}[7]
  {\argext{\argauxo{#1}{#2}{#3}{#4}{#5}{#6}}{, \allowbreak}{#7}}

\newcommand{\argq}[9]
  {%
  \def\valarga{#1}%
  \def\valargb{#2}%
  \def\valargc{#3}%
  \def\valargd{#4}%
  \def\valarge{#5}%
  \def\valargf{#6}%
  \def\valargg{#7}%
  \def\valargh{#8}%
  \def\valargi{#9}%
  \argauxq%
  }
\newcommand{\argauxq}[8]
  {\argext{\argauxp{#1}{#2}{#3}{#4}{#5}{#6}{#7}}{, \allowbreak}{#8}}

\newcommand{\argr}[9]
  {%
  \def\valarga{#1}%
  \def\valargb{#2}%
  \def\valargc{#3}%
  \def\valargd{#4}%
  \def\valarge{#5}%
  \def\valargf{#6}%
  \def\valargg{#7}%
  \def\valargh{#8}%
  \def\valargi{#9}%
  \argauxr%
  }
\newcommand{\argauxr}[9]
  {\argext{\argauxq{#1}{#2}{#3}{#4}{#5}{#6}{#7}{#8}}{, \allowbreak}{#9}}

\newcommand{\txtfnt}[2][]
  {{%
  \IfStrEq{#1}{}
    {#2}
    {%
    \StrLeft{#1}{2}[\optbgn]%
    \StrGobbleLeft{#1}{2}[\optend]%
    \IfStrEqCase{\optbgn}
      {%
      {Rm}{\rmfamily\txtfnt[\optend]{#2}}%
      {Sf}{\sffamily\txtfnt[\optend]{#2}}%
      {Tt}{\ttfamily\txtfnt[\optend]{#2}}%
      {Up}{\upshape\txtfnt[\optend]{#2}}%
      {It}{\itshape\txtfnt[\optend]{#2}}%
      {Sl}{\slshape\txtfnt[\optend]{#2}}%
      {Sc}{\scshape\txtfnt[\optend]{#2}}%
      {Md}{\mdseries\txtfnt[\optend]{#2}}%
      {Bf}{\bfseries\txtfnt[\optend]{#2}}%
      {Em}{\emph{\txtfnt[\optend]{#2}}}%
      }
      [\ensuremath{\clubsuit}]%
    }%
  }}

\newcommand{\txtsub}[2][]
  {\argemp{#2}{\ensuremath{_{\text{\txtfnt[#1]{#2}}}}}}

\newcommand{\txtsup}[2][]
  {\argemp{#2}{\ensuremath{^{\text{\txtfnt[#1]{#2}}}}}}

\newcommandx{\txt}[4][1=, 3=, 4=]
  {\text{\txtfnt[#1]{#2}\ensuremath{\txtsub[#1]{#3}\txtsup[#1]{#4}}}\xspace}

\newcommandx{\txtarg}[5][1=, 3=, 4=]
  {{\txt[#1]{#2}[#3][#4]\argint{(}{#5}{)}}\xspace}

\newcommand{\txtstyname}{RmScMd}
\newcommand{\txtname}[1][]
  {\txt[\argdef{#1}{\txtstyname}]}
\newcommand{\txtargname}[1][]
  {\txtarg[\argdef{#1}{\txtstyname}]}

\newcommand{\txtstyabr}{Em}
\newcommand{\txtabr}[1][]
  {\txt[\argdef{#1}{\txtstyabr}]}

\newcommandx{\mthfnt}[3][1=, 2=0]
  {{%
  \IfStrEqCase{#1}
    {%
    {}%
      {#3}%
    {Name}%
      {%
      \IfStrEqCase{#2}
        {%
        {0}{\mathcal{#3}}%
        {1}{\mathscr{#3}}%
        {2}{\mathfrak{#3}}%
        {3}{\mathbf{#3}}%
        }
        [\ensuremath{\clubsuit}]%
      }%
    {Set}%
      {%
      \IfStrEqCase{#2}
        {%
        {0}{\mathrm{#3}}%
        {1}{\mathbb{#3}}%
        {2}{\mathsf{#3}}%
        {3}{\mathtt{#3}}%
        }
        [\ensuremath{\clubsuit}]%
      }%
    {Fun}%
      {%
      \IfStrEqCase{#2}
        {%
        {0}{\mathsf{#3}}%
        {1}{\mathrm{#3}}%
        }
        [\ensuremath{\clubsuit}]%
      }%
    {Rel}%
      {%
      \IfStrEqCase{#2}
        {%
        {0}{\mathit{#3}}%
        {1}{\mathtt{#3}}%
        }
        [\ensuremath{\clubsuit}]%
      }%
    {Sym}%
      {%
      \IfStrEqCase{#2}
        {%
        {0}{\mathtt{#3}}%
        {1}{\mathbf{#3}}%
        }
        [\ensuremath{\clubsuit}]%
      }%
    {Elm}%
      {\mathnormal{#3}}
    }
    [\ensuremath{\clubsuit}]%
  }}

\newcommand{\mthsub}[1]
  {\argemp{#1}{\ensuremath{_{\mathnormal{#1}}}}}

\newcommand{\mthsup}[1]
  {\argemp{#1}{\ensuremath{^{\mathnormal{#1}}}}}

\newcommandx{\mth}[5][1=, 2=0, 4=, 5=]
  {{\ensuremath{\mthfnt[#1][#2]{#3}\mthsub{#4}\mthsup{#5}}}}

\newcommandx{\mtharg}[6][1=, 2=0, 4=, 5=]
  {{\mth[#1][#2]{#3}[#4][#5]\ensuremath{\argint{\!\left(}{#6}{\right)}}}}

\newcommand{\mthempty}
  {\mth[][]}

\newcommand{\mthstyname}{0}
\newcommand{\mthname}[1][]
  {\mth[Name][\argdef{#1}{\mthstyname}]}

\newcommand{\mthstyset}{0}
\newcommand{\mthset}[1][]
  {\mth[Set][\argdef{#1}{\mthstyset}]}
\newcommand{\mthargset}[1][]
  {\mtharg[Set][\argdef{#1}{\mthstyset}]}

\newcommand{\mthstyfun}{0}
\newcommand{\mthfun}[1][]
  {\mth[Fun][\argdef{#1}{\mthstyfun}]}
\newcommand{\mthargfun}[1][]
  {\mtharg[Fun][\argdef{#1}{\mthstyfun}]}

\newcommand{\mthstyrel}{0}
\newcommand{\mthrel}[1][]
  {\mth[Rel][\argdef{#1}{\mthstyrel}]}

\newcommand{\mthstysym}{0}
\newcommand{\mthsym}[1][]
  {\mth[Sym][\argdef{#1}{\mthstysym}]}

\newcommand{\mthstyelm}{0}
\newcommand{\mthelm}[1][]
  {\mth[Elm][\argdef{#1}{\mthstyelm}]}

\newcommandx{\AName}[4][1=, 2=, 3=, 4=]{\mthname[#4]{A#3}[#1][#2]}
\newcommandx{\BName}[4][1=, 2=, 3=, 4=]{\mthname[#4]{B#3}[#1][#2]}
\newcommandx{\CName}[4][1=, 2=, 3=, 4=]{\mthname[#4]{C#3}[#1][#2]}
\newcommandx{\DName}[4][1=, 2=, 3=, 4=]{\mthname[#4]{D#3}[#1][#2]}
\newcommandx{\EName}[4][1=, 2=, 3=, 4=]{\mthname[#4]{E#3}[#1][#2]}
\newcommandx{\FName}[4][1=, 2=, 3=, 4=]{\mthname[#4]{F#3}[#1][#2]}
\newcommandx{\GName}[4][1=, 2=, 3=, 4=]{\mthname[#4]{G#3}[#1][#2]}
\newcommandx{\HName}[4][1=, 2=, 3=, 4=]{\mthname[#4]{H#3}[#1][#2]}
\newcommandx{\IName}[4][1=, 2=, 3=, 4=]{\mthname[#4]{I#3}[#1][#2]}
\newcommandx{\JName}[4][1=, 2=, 3=, 4=]{\mthname[#4]{J#3}[#1][#2]}
\newcommandx{\KName}[4][1=, 2=, 3=, 4=]{\mthname[#4]{K#3}[#1][#2]}
\newcommandx{\LName}[4][1=, 2=, 3=, 4=]{\mthname[#4]{L#3}[#1][#2]}
\newcommandx{\MName}[4][1=, 2=, 3=, 4=]{\mthname[#4]{M#3}[#1][#2]}
\newcommandx{\NName}[4][1=, 2=, 3=, 4=]{\mthname[#4]{N#3}[#1][#2]}
\newcommandx{\OName}[4][1=, 2=, 3=, 4=]{\mthname[#4]{O#3}[#1][#2]}
\newcommandx{\PName}[4][1=, 2=, 3=, 4=]{\mthname[#4]{P#3}[#1][#2]}
\newcommandx{\QName}[4][1=, 2=, 3=, 4=]{\mthname[#4]{Q#3}[#1][#2]}
\newcommandx{\RName}[4][1=, 2=, 3=, 4=]{\mthname[#4]{R#3}[#1][#2]}
\newcommandx{\SName}[4][1=, 2=, 3=, 4=]{\mthname[#4]{S#3}[#1][#2]}
\newcommandx{\TName}[4][1=, 2=, 3=, 4=]{\mthname[#4]{T#3}[#1][#2]}
\newcommandx{\UName}[4][1=, 2=, 3=, 4=]{\mthname[#4]{U#3}[#1][#2]}
\newcommandx{\VName}[4][1=, 2=, 3=, 4=]{\mthname[#4]{V#3}[#1][#2]}
\newcommandx{\WName}[4][1=, 2=, 3=, 4=]{\mthname[#4]{W#3}[#1][#2]}
\newcommandx{\XName}[4][1=, 2=, 3=, 4=]{\mthname[#4]{X#3}[#1][#2]}
\newcommandx{\YName}[4][1=, 2=, 3=, 4=]{\mthname[#4]{Y#3}[#1][#2]}
\newcommandx{\ZName}[4][1=, 2=, 3=, 4=]{\mthname[#4]{Z#3}[#1][#2]}

\newcommandx{\ASet}[4][1=, 2=, 3=, 4=]{\mthset[#4]{A#3}[#1][#2]}
\newcommandx{\BSet}[4][1=, 2=, 3=, 4=]{\mthset[#4]{B#3}[#1][#2]}
\newcommandx{\CSet}[4][1=, 2=, 3=, 4=]{\mthset[#4]{C#3}[#1][#2]}
\newcommandx{\DSet}[4][1=, 2=, 3=, 4=]{\mthset[#4]{D#3}[#1][#2]}
\newcommandx{\ESet}[4][1=, 2=, 3=, 4=]{\mthset[#4]{E#3}[#1][#2]}
\newcommandx{\FSet}[4][1=, 2=, 3=, 4=]{\mthset[#4]{F#3}[#1][#2]}
\newcommandx{\GSet}[4][1=, 2=, 3=, 4=]{\mthset[#4]{G#3}[#1][#2]}
\newcommandx{\HSet}[4][1=, 2=, 3=, 4=]{\mthset[#4]{H#3}[#1][#2]}
\newcommandx{\ISet}[4][1=, 2=, 3=, 4=]{\mthset[#4]{I#3}[#1][#2]}
\newcommandx{\JSet}[4][1=, 2=, 3=, 4=]{\mthset[#4]{J#3}[#1][#2]}
\newcommandx{\KSet}[4][1=, 2=, 3=, 4=]{\mthset[#4]{K#3}[#1][#2]}
\newcommandx{\LSet}[4][1=, 2=, 3=, 4=]{\mthset[#4]{L#3}[#1][#2]}
\newcommandx{\MSet}[4][1=, 2=, 3=, 4=]{\mthset[#4]{M#3}[#1][#2]}
\newcommandx{\NSet}[4][1=, 2=, 3=, 4=]{\mthset[#4]{N#3}[#1][#2]}
\newcommandx{\OSet}[4][1=, 2=, 3=, 4=]{\mthset[#4]{O#3}[#1][#2]}
\newcommandx{\PSet}[4][1=, 2=, 3=, 4=]{\mthset[#4]{P#3}[#1][#2]}
\newcommandx{\QSet}[4][1=, 2=, 3=, 4=]{\mthset[#4]{Q#3}[#1][#2]}
\newcommandx{\RSet}[4][1=, 2=, 3=, 4=]{\mthset[#4]{R#3}[#1][#2]}
\newcommandx{\SSet}[4][1=, 2=, 3=, 4=]{\mthset[#4]{S#3}[#1][#2]}
\newcommandx{\TSet}[4][1=, 2=, 3=, 4=]{\mthset[#4]{T#3}[#1][#2]}
\newcommandx{\USet}[4][1=, 2=, 3=, 4=]{\mthset[#4]{U#3}[#1][#2]}
\newcommandx{\VSet}[4][1=, 2=, 3=, 4=]{\mthset[#4]{V#3}[#1][#2]}
\newcommandx{\WSet}[4][1=, 2=, 3=, 4=]{\mthset[#4]{W#3}[#1][#2]}
\newcommandx{\XSet}[4][1=, 2=, 3=, 4=]{\mthset[#4]{X#3}[#1][#2]}
\newcommandx{\YSet}[4][1=, 2=, 3=, 4=]{\mthset[#4]{Y#3}[#1][#2]}
\newcommandx{\ZSet}[4][1=, 2=, 3=, 4=]{\mthset[#4]{Z#3}[#1][#2]}

\newcommandx{\aSet}[4][1=, 2=, 3=, 4=]{\mthset[#4]{a#3}[#1][#2]}
\newcommandx{\bSet}[4][1=, 2=, 3=, 4=]{\mthset[#4]{b#3}[#1][#2]}
\newcommandx{\cSet}[4][1=, 2=, 3=, 4=]{\mthset[#4]{c#3}[#1][#2]}
\newcommandx{\dSet}[4][1=, 2=, 3=, 4=]{\mthset[#4]{d#3}[#1][#2]}
\newcommandx{\eSet}[4][1=, 2=, 3=, 4=]{\mthset[#4]{e#3}[#1][#2]}
\newcommandx{\fSet}[4][1=, 2=, 3=, 4=]{\mthset[#4]{f#3}[#1][#2]}
\newcommandx{\gSet}[4][1=, 2=, 3=, 4=]{\mthset[#4]{g#3}[#1][#2]}
\newcommandx{\hSet}[4][1=, 2=, 3=, 4=]{\mthset[#4]{h#3}[#1][#2]}
\newcommandx{\iSet}[4][1=, 2=, 3=, 4=]{\mthset[#4]{i#3}[#1][#2]}
\newcommandx{\jSet}[4][1=, 2=, 3=, 4=]{\mthset[#4]{j#3}[#1][#2]}
\newcommandx{\kSet}[4][1=, 2=, 3=, 4=]{\mthset[#4]{k#3}[#1][#2]}
\newcommandx{\lSet}[4][1=, 2=, 3=, 4=]{\mthset[#4]{l#3}[#1][#2]}
\newcommandx{\mSet}[4][1=, 2=, 3=, 4=]{\mthset[#4]{m#3}[#1][#2]}
\newcommandx{\nSet}[4][1=, 2=, 3=, 4=]{\mthset[#4]{n#3}[#1][#2]}
\newcommandx{\oSet}[4][1=, 2=, 3=, 4=]{\mthset[#4]{o#3}[#1][#2]}
\newcommandx{\pSet}[4][1=, 2=, 3=, 4=]{\mthset[#4]{p#3}[#1][#2]}
\newcommandx{\qSet}[4][1=, 2=, 3=, 4=]{\mthset[#4]{q#3}[#1][#2]}
\newcommandx{\rSet}[4][1=, 2=, 3=, 4=]{\mthset[#4]{r#3}[#1][#2]}
\newcommandx{\sSet}[4][1=, 2=, 3=, 4=]{\mthset[#4]{s#3}[#1][#2]}
\newcommandx{\tSet}[4][1=, 2=, 3=, 4=]{\mthset[#4]{t#3}[#1][#2]}
\newcommandx{\uSet}[4][1=, 2=, 3=, 4=]{\mthset[#4]{u#3}[#1][#2]}
\newcommandx{\vSet}[4][1=, 2=, 3=, 4=]{\mthset[#4]{v#3}[#1][#2]}
\newcommandx{\wSet}[4][1=, 2=, 3=, 4=]{\mthset[#4]{w#3}[#1][#2]}
\newcommandx{\xSet}[4][1=, 2=, 3=, 4=]{\mthset[#4]{x#3}[#1][#2]}
\newcommandx{\ySet}[4][1=, 2=, 3=, 4=]{\mthset[#4]{y#3}[#1][#2]}
\newcommandx{\zSet}[4][1=, 2=, 3=, 4=]{\mthset[#4]{z#3}[#1][#2]}

\newcommandx{\AFun}[4][1=, 2=, 3=, 4=]{\mthfun[#4]{A#3}[#1][#2]}
\newcommandx{\BFun}[4][1=, 2=, 3=, 4=]{\mthfun[#4]{B#3}[#1][#2]}
\newcommandx{\CFun}[4][1=, 2=, 3=, 4=]{\mthfun[#4]{C#3}[#1][#2]}
\newcommandx{\DFun}[4][1=, 2=, 3=, 4=]{\mthfun[#4]{D#3}[#1][#2]}
\newcommandx{\EFun}[4][1=, 2=, 3=, 4=]{\mthfun[#4]{E#3}[#1][#2]}
\newcommandx{\FFun}[4][1=, 2=, 3=, 4=]{\mthfun[#4]{F#3}[#1][#2]}
\newcommandx{\GFun}[4][1=, 2=, 3=, 4=]{\mthfun[#4]{G#3}[#1][#2]}
\newcommandx{\HFun}[4][1=, 2=, 3=, 4=]{\mthfun[#4]{H#3}[#1][#2]}
\newcommandx{\IFun}[4][1=, 2=, 3=, 4=]{\mthfun[#4]{I#3}[#1][#2]}
\newcommandx{\JFun}[4][1=, 2=, 3=, 4=]{\mthfun[#4]{J#3}[#1][#2]}
\newcommandx{\KFun}[4][1=, 2=, 3=, 4=]{\mthfun[#4]{K#3}[#1][#2]}
\newcommandx{\LFun}[4][1=, 2=, 3=, 4=]{\mthfun[#4]{L#3}[#1][#2]}
\newcommandx{\MFun}[4][1=, 2=, 3=, 4=]{\mthfun[#4]{M#3}[#1][#2]}
\newcommandx{\NFun}[4][1=, 2=, 3=, 4=]{\mthfun[#4]{N#3}[#1][#2]}
\newcommandx{\OFun}[4][1=, 2=, 3=, 4=]{\mthfun[#4]{O#3}[#1][#2]}
\newcommandx{\PFun}[4][1=, 2=, 3=, 4=]{\mthfun[#4]{P#3}[#1][#2]}
\newcommandx{\QFun}[4][1=, 2=, 3=, 4=]{\mthfun[#4]{Q#3}[#1][#2]}
\newcommandx{\RFun}[4][1=, 2=, 3=, 4=]{\mthfun[#4]{R#3}[#1][#2]}
\newcommandx{\SFun}[4][1=, 2=, 3=, 4=]{\mthfun[#4]{S#3}[#1][#2]}
\newcommandx{\TFun}[4][1=, 2=, 3=, 4=]{\mthfun[#4]{T#3}[#1][#2]}
\newcommandx{\UFun}[4][1=, 2=, 3=, 4=]{\mthfun[#4]{U#3}[#1][#2]}
\newcommandx{\VFun}[4][1=, 2=, 3=, 4=]{\mthfun[#4]{V#3}[#1][#2]}
\newcommandx{\WFun}[4][1=, 2=, 3=, 4=]{\mthfun[#4]{W#3}[#1][#2]}
\newcommandx{\XFun}[4][1=, 2=, 3=, 4=]{\mthfun[#4]{X#3}[#1][#2]}
\newcommandx{\YFun}[4][1=, 2=, 3=, 4=]{\mthfun[#4]{Y#3}[#1][#2]}
\newcommandx{\ZFun}[4][1=, 2=, 3=, 4=]{\mthfun[#4]{Z#3}[#1][#2]}

\newcommandx{\aFun}[4][1=, 2=, 3=, 4=]{\mthfun[#4]{a#3}[#1][#2]}
\newcommandx{\bFun}[4][1=, 2=, 3=, 4=]{\mthfun[#4]{b#3}[#1][#2]}
\newcommandx{\cFun}[4][1=, 2=, 3=, 4=]{\mthfun[#4]{c#3}[#1][#2]}
\newcommandx{\dFun}[4][1=, 2=, 3=, 4=]{\mthfun[#4]{d#3}[#1][#2]}
\newcommandx{\eFun}[4][1=, 2=, 3=, 4=]{\mthfun[#4]{e#3}[#1][#2]}
\newcommandx{\fFun}[4][1=, 2=, 3=, 4=]{\mthfun[#4]{f#3}[#1][#2]}
\newcommandx{\gFun}[4][1=, 2=, 3=, 4=]{\mthfun[#4]{g#3}[#1][#2]}
\newcommandx{\hFun}[4][1=, 2=, 3=, 4=]{\mthfun[#4]{h#3}[#1][#2]}
\newcommandx{\iFun}[4][1=, 2=, 3=, 4=]{\mthfun[#4]{i#3}[#1][#2]}
\newcommandx{\jFun}[4][1=, 2=, 3=, 4=]{\mthfun[#4]{j#3}[#1][#2]}
\newcommandx{\kFun}[4][1=, 2=, 3=, 4=]{\mthfun[#4]{k#3}[#1][#2]}
\newcommandx{\lFun}[4][1=, 2=, 3=, 4=]{\mthfun[#4]{l#3}[#1][#2]}
\newcommandx{\mFun}[4][1=, 2=, 3=, 4=]{\mthfun[#4]{m#3}[#1][#2]}
\newcommandx{\nFun}[4][1=, 2=, 3=, 4=]{\mthfun[#4]{n#3}[#1][#2]}
\newcommandx{\oFun}[4][1=, 2=, 3=, 4=]{\mthfun[#4]{o#3}[#1][#2]}
\newcommandx{\pFun}[4][1=, 2=, 3=, 4=]{\mthfun[#4]{p#3}[#1][#2]}
\newcommandx{\qFun}[4][1=, 2=, 3=, 4=]{\mthfun[#4]{q#3}[#1][#2]}
\newcommandx{\rFun}[4][1=, 2=, 3=, 4=]{\mthfun[#4]{r#3}[#1][#2]}
\newcommandx{\sFun}[4][1=, 2=, 3=, 4=]{\mthfun[#4]{s#3}[#1][#2]}
\newcommandx{\tFun}[4][1=, 2=, 3=, 4=]{\mthfun[#4]{t#3}[#1][#2]}
\newcommandx{\uFun}[4][1=, 2=, 3=, 4=]{\mthfun[#4]{u#3}[#1][#2]}
\newcommandx{\vFun}[4][1=, 2=, 3=, 4=]{\mthfun[#4]{v#3}[#1][#2]}
\newcommandx{\wFun}[4][1=, 2=, 3=, 4=]{\mthfun[#4]{w#3}[#1][#2]}
\newcommandx{\xFun}[4][1=, 2=, 3=, 4=]{\mthfun[#4]{x#3}[#1][#2]}
\newcommandx{\yFun}[4][1=, 2=, 3=, 4=]{\mthfun[#4]{y#3}[#1][#2]}
\newcommandx{\zFun}[4][1=, 2=, 3=, 4=]{\mthfun[#4]{z#3}[#1][#2]}

\newcommandx{\ARel}[4][1=, 2=, 3=, 4=]{\mthrel[#4]{A#3}[#1][#2]}
\newcommandx{\BRel}[4][1=, 2=, 3=, 4=]{\mthrel[#4]{B#3}[#1][#2]}
\newcommandx{\CRel}[4][1=, 2=, 3=, 4=]{\mthrel[#4]{C#3}[#1][#2]}
\newcommandx{\DRel}[4][1=, 2=, 3=, 4=]{\mthrel[#4]{D#3}[#1][#2]}
\newcommandx{\ERel}[4][1=, 2=, 3=, 4=]{\mthrel[#4]{E#3}[#1][#2]}
\newcommandx{\FRel}[4][1=, 2=, 3=, 4=]{\mthrel[#4]{F#3}[#1][#2]}
\newcommandx{\GRel}[4][1=, 2=, 3=, 4=]{\mthrel[#4]{G#3}[#1][#2]}
\newcommandx{\HRel}[4][1=, 2=, 3=, 4=]{\mthrel[#4]{H#3}[#1][#2]}
\newcommandx{\IRel}[4][1=, 2=, 3=, 4=]{\mthrel[#4]{I#3}[#1][#2]}
\newcommandx{\JRel}[4][1=, 2=, 3=, 4=]{\mthrel[#4]{J#3}[#1][#2]}
\newcommandx{\KRel}[4][1=, 2=, 3=, 4=]{\mthrel[#4]{K#3}[#1][#2]}
\newcommandx{\LRel}[4][1=, 2=, 3=, 4=]{\mthrel[#4]{L#3}[#1][#2]}
\newcommandx{\MRel}[4][1=, 2=, 3=, 4=]{\mthrel[#4]{M#3}[#1][#2]}
\newcommandx{\NRel}[4][1=, 2=, 3=, 4=]{\mthrel[#4]{N#3}[#1][#2]}
\newcommandx{\ORel}[4][1=, 2=, 3=, 4=]{\mthrel[#4]{O#3}[#1][#2]}
\newcommandx{\PRel}[4][1=, 2=, 3=, 4=]{\mthrel[#4]{P#3}[#1][#2]}
\newcommandx{\QRel}[4][1=, 2=, 3=, 4=]{\mthrel[#4]{Q#3}[#1][#2]}
\newcommandx{\RRel}[4][1=, 2=, 3=, 4=]{\mthrel[#4]{R#3}[#1][#2]}
\newcommandx{\SRel}[4][1=, 2=, 3=, 4=]{\mthrel[#4]{S#3}[#1][#2]}
\newcommandx{\TRel}[4][1=, 2=, 3=, 4=]{\mthrel[#4]{T#3}[#1][#2]}
\newcommandx{\URel}[4][1=, 2=, 3=, 4=]{\mthrel[#4]{U#3}[#1][#2]}
\newcommandx{\VRel}[4][1=, 2=, 3=, 4=]{\mthrel[#4]{V#3}[#1][#2]}
\newcommandx{\WRel}[4][1=, 2=, 3=, 4=]{\mthrel[#4]{W#3}[#1][#2]}
\newcommandx{\XRel}[4][1=, 2=, 3=, 4=]{\mthrel[#4]{X#3}[#1][#2]}
\newcommandx{\YRel}[4][1=, 2=, 3=, 4=]{\mthrel[#4]{Y#3}[#1][#2]}
\newcommandx{\ZRel}[4][1=, 2=, 3=, 4=]{\mthrel[#4]{Z#3}[#1][#2]}

\newcommandx{\aRel}[4][1=, 2=, 3=, 4=]{\mthrel[#4]{a#3}[#1][#2]}
\newcommandx{\bRel}[4][1=, 2=, 3=, 4=]{\mthrel[#4]{b#3}[#1][#2]}
\newcommandx{\cRel}[4][1=, 2=, 3=, 4=]{\mthrel[#4]{c#3}[#1][#2]}
\newcommandx{\dRel}[4][1=, 2=, 3=, 4=]{\mthrel[#4]{d#3}[#1][#2]}
\newcommandx{\eRel}[4][1=, 2=, 3=, 4=]{\mthrel[#4]{e#3}[#1][#2]}
\newcommandx{\fRel}[4][1=, 2=, 3=, 4=]{\mthrel[#4]{f#3}[#1][#2]}
\newcommandx{\gRel}[4][1=, 2=, 3=, 4=]{\mthrel[#4]{g#3}[#1][#2]}
\newcommandx{\hRel}[4][1=, 2=, 3=, 4=]{\mthrel[#4]{h#3}[#1][#2]}
\newcommandx{\iRel}[4][1=, 2=, 3=, 4=]{\mthrel[#4]{i#3}[#1][#2]}
\newcommandx{\jRel}[4][1=, 2=, 3=, 4=]{\mthrel[#4]{j#3}[#1][#2]}
\newcommandx{\kRel}[4][1=, 2=, 3=, 4=]{\mthrel[#4]{k#3}[#1][#2]}
\newcommandx{\lRel}[4][1=, 2=, 3=, 4=]{\mthrel[#4]{l#3}[#1][#2]}
\newcommandx{\mRel}[4][1=, 2=, 3=, 4=]{\mthrel[#4]{m#3}[#1][#2]}
\newcommandx{\nRel}[4][1=, 2=, 3=, 4=]{\mthrel[#4]{n#3}[#1][#2]}
\newcommandx{\oRel}[4][1=, 2=, 3=, 4=]{\mthrel[#4]{o#3}[#1][#2]}
\newcommandx{\pRel}[4][1=, 2=, 3=, 4=]{\mthrel[#4]{p#3}[#1][#2]}
\newcommandx{\qRel}[4][1=, 2=, 3=, 4=]{\mthrel[#4]{q#3}[#1][#2]}
\newcommandx{\rRel}[4][1=, 2=, 3=, 4=]{\mthrel[#4]{r#3}[#1][#2]}
\newcommandx{\sRel}[4][1=, 2=, 3=, 4=]{\mthrel[#4]{s#3}[#1][#2]}
\newcommandx{\tRel}[4][1=, 2=, 3=, 4=]{\mthrel[#4]{t#3}[#1][#2]}
\newcommandx{\uRel}[4][1=, 2=, 3=, 4=]{\mthrel[#4]{u#3}[#1][#2]}
\newcommandx{\vRel}[4][1=, 2=, 3=, 4=]{\mthrel[#4]{v#3}[#1][#2]}
\newcommandx{\wRel}[4][1=, 2=, 3=, 4=]{\mthrel[#4]{w#3}[#1][#2]}
\newcommandx{\xRel}[4][1=, 2=, 3=, 4=]{\mthrel[#4]{x#3}[#1][#2]}
\newcommandx{\yRel}[4][1=, 2=, 3=, 4=]{\mthrel[#4]{y#3}[#1][#2]}
\newcommandx{\zRel}[4][1=, 2=, 3=, 4=]{\mthrel[#4]{z#3}[#1][#2]}

\newcommandx{\ASym}[4][1=, 2=, 3=, 4=]{\mthsym[#4]{A#3}[#1][#2]}
\newcommandx{\BSym}[4][1=, 2=, 3=, 4=]{\mthsym[#4]{B#3}[#1][#2]}
\newcommandx{\CSym}[4][1=, 2=, 3=, 4=]{\mthsym[#4]{C#3}[#1][#2]}
\newcommandx{\DSym}[4][1=, 2=, 3=, 4=]{\mthsym[#4]{D#3}[#1][#2]}
\newcommandx{\ESym}[4][1=, 2=, 3=, 4=]{\mthsym[#4]{E#3}[#1][#2]}
\newcommandx{\FSym}[4][1=, 2=, 3=, 4=]{\mthsym[#4]{F#3}[#1][#2]}
\newcommandx{\GSym}[4][1=, 2=, 3=, 4=]{\mthsym[#4]{G#3}[#1][#2]}
\newcommandx{\HSym}[4][1=, 2=, 3=, 4=]{\mthsym[#4]{H#3}[#1][#2]}
\newcommandx{\ISym}[4][1=, 2=, 3=, 4=]{\mthsym[#4]{I#3}[#1][#2]}
\newcommandx{\JSym}[4][1=, 2=, 3=, 4=]{\mthsym[#4]{J#3}[#1][#2]}
\newcommandx{\KSym}[4][1=, 2=, 3=, 4=]{\mthsym[#4]{K#3}[#1][#2]}
\newcommandx{\LSym}[4][1=, 2=, 3=, 4=]{\mthsym[#4]{L#3}[#1][#2]}
\newcommandx{\MSym}[4][1=, 2=, 3=, 4=]{\mthsym[#4]{M#3}[#1][#2]}
\newcommandx{\NSym}[4][1=, 2=, 3=, 4=]{\mthsym[#4]{N#3}[#1][#2]}
\newcommandx{\OSym}[4][1=, 2=, 3=, 4=]{\mthsym[#4]{O#3}[#1][#2]}
\newcommandx{\PSym}[4][1=, 2=, 3=, 4=]{\mthsym[#4]{P#3}[#1][#2]}
\newcommandx{\QSym}[4][1=, 2=, 3=, 4=]{\mthsym[#4]{Q#3}[#1][#2]}
\newcommandx{\RSym}[4][1=, 2=, 3=, 4=]{\mthsym[#4]{R#3}[#1][#2]}
\newcommandx{\SSym}[4][1=, 2=, 3=, 4=]{\mthsym[#4]{S#3}[#1][#2]}
\newcommandx{\TSym}[4][1=, 2=, 3=, 4=]{\mthsym[#4]{T#3}[#1][#2]}
\newcommandx{\USym}[4][1=, 2=, 3=, 4=]{\mthsym[#4]{U#3}[#1][#2]}
\newcommandx{\VSym}[4][1=, 2=, 3=, 4=]{\mthsym[#4]{V#3}[#1][#2]}
\newcommandx{\WSym}[4][1=, 2=, 3=, 4=]{\mthsym[#4]{W#3}[#1][#2]}
\newcommandx{\XSym}[4][1=, 2=, 3=, 4=]{\mthsym[#4]{X#3}[#1][#2]}
\newcommandx{\YSym}[4][1=, 2=, 3=, 4=]{\mthsym[#4]{Y#3}[#1][#2]}
\newcommandx{\ZSym}[4][1=, 2=, 3=, 4=]{\mthsym[#4]{Z#3}[#1][#2]}

\newcommandx{\aSym}[4][1=, 2=, 3=, 4=]{\mthsym[#4]{a#3}[#1][#2]}
\newcommandx{\bSym}[4][1=, 2=, 3=, 4=]{\mthsym[#4]{b#3}[#1][#2]}
\newcommandx{\cSym}[4][1=, 2=, 3=, 4=]{\mthsym[#4]{c#3}[#1][#2]}
\newcommandx{\dSym}[4][1=, 2=, 3=, 4=]{\mthsym[#4]{d#3}[#1][#2]}
\newcommandx{\eSym}[4][1=, 2=, 3=, 4=]{\mthsym[#4]{e#3}[#1][#2]}
\newcommandx{\fSym}[4][1=, 2=, 3=, 4=]{\mthsym[#4]{f#3}[#1][#2]}
\newcommandx{\gSym}[4][1=, 2=, 3=, 4=]{\mthsym[#4]{g#3}[#1][#2]}
\newcommandx{\hSym}[4][1=, 2=, 3=, 4=]{\mthsym[#4]{h#3}[#1][#2]}
\newcommandx{\iSym}[4][1=, 2=, 3=, 4=]{\mthsym[#4]{i#3}[#1][#2]}
\newcommandx{\jSym}[4][1=, 2=, 3=, 4=]{\mthsym[#4]{j#3}[#1][#2]}
\newcommandx{\kSym}[4][1=, 2=, 3=, 4=]{\mthsym[#4]{k#3}[#1][#2]}
\newcommandx{\lSym}[4][1=, 2=, 3=, 4=]{\mthsym[#4]{l#3}[#1][#2]}
\newcommandx{\mSym}[4][1=, 2=, 3=, 4=]{\mthsym[#4]{m#3}[#1][#2]}
\newcommandx{\nSym}[4][1=, 2=, 3=, 4=]{\mthsym[#4]{n#3}[#1][#2]}
\newcommandx{\oSym}[4][1=, 2=, 3=, 4=]{\mthsym[#4]{o#3}[#1][#2]}
\newcommandx{\pSym}[4][1=, 2=, 3=, 4=]{\mthsym[#4]{p#3}[#1][#2]}
\newcommandx{\qSym}[4][1=, 2=, 3=, 4=]{\mthsym[#4]{q#3}[#1][#2]}
\newcommandx{\rSym}[4][1=, 2=, 3=, 4=]{\mthsym[#4]{r#3}[#1][#2]}
\newcommandx{\sSym}[4][1=, 2=, 3=, 4=]{\mthsym[#4]{s#3}[#1][#2]}
\newcommandx{\tSym}[4][1=, 2=, 3=, 4=]{\mthsym[#4]{t#3}[#1][#2]}
\newcommandx{\uSym}[4][1=, 2=, 3=, 4=]{\mthsym[#4]{u#3}[#1][#2]}
\newcommandx{\vSym}[4][1=, 2=, 3=, 4=]{\mthsym[#4]{v#3}[#1][#2]}
\newcommandx{\wSym}[4][1=, 2=, 3=, 4=]{\mthsym[#4]{w#3}[#1][#2]}
\newcommandx{\xSym}[4][1=, 2=, 3=, 4=]{\mthsym[#4]{x#3}[#1][#2]}
\newcommandx{\ySym}[4][1=, 2=, 3=, 4=]{\mthsym[#4]{y#3}[#1][#2]}
\newcommandx{\zSym}[4][1=, 2=, 3=, 4=]{\mthsym[#4]{z#3}[#1][#2]}

\newcommandx{\AElm}[4][1=, 2=, 3=, 4=]{\mthelm[#4]{A#3}[#1][#2]}
\newcommandx{\BElm}[4][1=, 2=, 3=, 4=]{\mthelm[#4]{B#3}[#1][#2]}
\newcommandx{\CElm}[4][1=, 2=, 3=, 4=]{\mthelm[#4]{C#3}[#1][#2]}
\newcommandx{\DElm}[4][1=, 2=, 3=, 4=]{\mthelm[#4]{D#3}[#1][#2]}
\newcommandx{\EElm}[4][1=, 2=, 3=, 4=]{\mthelm[#4]{E#3}[#1][#2]}
\newcommandx{\FElm}[4][1=, 2=, 3=, 4=]{\mthelm[#4]{F#3}[#1][#2]}
\newcommandx{\GElm}[4][1=, 2=, 3=, 4=]{\mthelm[#4]{G#3}[#1][#2]}
\newcommandx{\HElm}[4][1=, 2=, 3=, 4=]{\mthelm[#4]{H#3}[#1][#2]}
\newcommandx{\IElm}[4][1=, 2=, 3=, 4=]{\mthelm[#4]{I#3}[#1][#2]}
\newcommandx{\JElm}[4][1=, 2=, 3=, 4=]{\mthelm[#4]{J#3}[#1][#2]}
\newcommandx{\KElm}[4][1=, 2=, 3=, 4=]{\mthelm[#4]{K#3}[#1][#2]}
\newcommandx{\LElm}[4][1=, 2=, 3=, 4=]{\mthelm[#4]{L#3}[#1][#2]}
\newcommandx{\MElm}[4][1=, 2=, 3=, 4=]{\mthelm[#4]{M#3}[#1][#2]}
\newcommandx{\NElm}[4][1=, 2=, 3=, 4=]{\mthelm[#4]{N#3}[#1][#2]}
\newcommandx{\OElm}[4][1=, 2=, 3=, 4=]{\mthelm[#4]{O#3}[#1][#2]}
\newcommandx{\PElm}[4][1=, 2=, 3=, 4=]{\mthelm[#4]{P#3}[#1][#2]}
\newcommandx{\QElm}[4][1=, 2=, 3=, 4=]{\mthelm[#4]{Q#3}[#1][#2]}
\newcommandx{\RElm}[4][1=, 2=, 3=, 4=]{\mthelm[#4]{R#3}[#1][#2]}
\newcommandx{\SElm}[4][1=, 2=, 3=, 4=]{\mthelm[#4]{S#3}[#1][#2]}
\newcommandx{\TElm}[4][1=, 2=, 3=, 4=]{\mthelm[#4]{T#3}[#1][#2]}
\newcommandx{\UElm}[4][1=, 2=, 3=, 4=]{\mthelm[#4]{U#3}[#1][#2]}
\newcommandx{\VElm}[4][1=, 2=, 3=, 4=]{\mthelm[#4]{V#3}[#1][#2]}
\newcommandx{\WElm}[4][1=, 2=, 3=, 4=]{\mthelm[#4]{W#3}[#1][#2]}
\newcommandx{\XElm}[4][1=, 2=, 3=, 4=]{\mthelm[#4]{X#3}[#1][#2]}
\newcommandx{\YElm}[4][1=, 2=, 3=, 4=]{\mthelm[#4]{Y#3}[#1][#2]}
\newcommandx{\ZElm}[4][1=, 2=, 3=, 4=]{\mthelm[#4]{Z#3}[#1][#2]}

\newcommandx{\aElm}[4][1=, 2=, 3=, 4=]{\mthelm[#4]{a#3}[#1][#2]}
\newcommandx{\bElm}[4][1=, 2=, 3=, 4=]{\mthelm[#4]{b#3}[#1][#2]}
\newcommandx{\cElm}[4][1=, 2=, 3=, 4=]{\mthelm[#4]{c#3}[#1][#2]}
\newcommandx{\dElm}[4][1=, 2=, 3=, 4=]{\mthelm[#4]{d#3}[#1][#2]}
\newcommandx{\eElm}[4][1=, 2=, 3=, 4=]{\mthelm[#4]{e#3}[#1][#2]}
\newcommandx{\fElm}[4][1=, 2=, 3=, 4=]{\mthelm[#4]{f#3}[#1][#2]}
\newcommandx{\gElm}[4][1=, 2=, 3=, 4=]{\mthelm[#4]{g#3}[#1][#2]}
\newcommandx{\hElm}[4][1=, 2=, 3=, 4=]{\mthelm[#4]{h#3}[#1][#2]}
\newcommandx{\iElm}[4][1=, 2=, 3=, 4=]{\mthelm[#4]{i#3}[#1][#2]}
\newcommandx{\jElm}[4][1=, 2=, 3=, 4=]{\mthelm[#4]{j#3}[#1][#2]}
\newcommandx{\kElm}[4][1=, 2=, 3=, 4=]{\mthelm[#4]{k#3}[#1][#2]}
\newcommandx{\lElm}[4][1=, 2=, 3=, 4=]{\mthelm[#4]{l#3}[#1][#2]}
\newcommandx{\mElm}[4][1=, 2=, 3=, 4=]{\mthelm[#4]{m#3}[#1][#2]}
\newcommandx{\nElm}[4][1=, 2=, 3=, 4=]{\mthelm[#4]{n#3}[#1][#2]}
\newcommandx{\oElm}[4][1=, 2=, 3=, 4=]{\mthelm[#4]{o#3}[#1][#2]}
\newcommandx{\pElm}[4][1=, 2=, 3=, 4=]{\mthelm[#4]{p#3}[#1][#2]}
\newcommandx{\qElm}[4][1=, 2=, 3=, 4=]{\mthelm[#4]{q#3}[#1][#2]}
\newcommandx{\rElm}[4][1=, 2=, 3=, 4=]{\mthelm[#4]{r#3}[#1][#2]}
\newcommandx{\sElm}[4][1=, 2=, 3=, 4=]{\mthelm[#4]{s#3}[#1][#2]}
\newcommandx{\tElm}[4][1=, 2=, 3=, 4=]{\mthelm[#4]{t#3}[#1][#2]}
\newcommandx{\uElm}[4][1=, 2=, 3=, 4=]{\mthelm[#4]{u#3}[#1][#2]}
\newcommandx{\vElm}[4][1=, 2=, 3=, 4=]{\mthelm[#4]{v#3}[#1][#2]}
\newcommandx{\wElm}[4][1=, 2=, 3=, 4=]{\mthelm[#4]{w#3}[#1][#2]}
\newcommandx{\xElm}[4][1=, 2=, 3=, 4=]{\mthelm[#4]{x#3}[#1][#2]}
\newcommandx{\yElm}[4][1=, 2=, 3=, 4=]{\mthelm[#4]{y#3}[#1][#2]}
\newcommandx{\zElm}[4][1=, 2=, 3=, 4=]{\mthelm[#4]{z#3}[#1][#2]}

\newcommand{\eg}
  {\txtabr{e.g.}}

\newcommand{\etal}
  {\txtabr{et al.}}

\newcommand{\ie}
  {\txtabr{i.e.}}

\renewcommand{\iff}
  {\txtabr{iff}}

\newcommand{\resp}
  {\txtabr{resp.}}

\newcommand{\wrt}
  {\txtabr{w.r.t.}}

\newcommand{\defeq}
  {\ensuremath{\triangleq}}

\newcommand{\fst}
  {\mthargfun{fst}}
\newcommand{\lst}
  {\mthargfun{lst}}

\newcommand{\dual}[1]
  {\mthempty{\overline{#1}}}

\newcommand{\der}[1]
  {\mthempty{\widehat{#1}}}

\newcommand{\tuple}[1]
  {\ensuremath{\!\argint{\langle}{#1}{\rangle}}}

\newcommand{\tupleb}[2]
  {\tuple{\argb{#1}{#2}}}
\newcommand{\tuplec}[3]
  {\tuple{\argc{#1}{#2}{#3}}}
\newcommand{\tupled}[4]
  {\tuple{\argd{#1}{#2}{#3}{#4}}}
\newcommand{\tuplee}[5]
  {\tuple{\arge{#1}{#2}{#3}{#4}{#5}}}
\newcommand{\tuplef}[6]
  {\tuple{\argf{#1}{#2}{#3}{#4}{#5}{#6}}}
\newcommand{\tupleg}[7]
  {\tuple{\argg{#1}{#2}{#3}{#4}{#5}{#6}{#7}}}
\newcommand{\tupleh}[8]
  {\tuple{\argh{#1}{#2}{#3}{#4}{#5}{#6}{#7}{#8}}}
\newcommand{\tuplei}[9]
  {\tuple{\argi{#1}{#2}{#3}{#4}{#5}{#6}{#7}{#8}{#9}}}

\newcommand{\tuplej}[9]
  {%
  \def\defarga{#1}%
  \def\defargb{#2}%
  \def\defargc{#3}%
  \def\defargd{#4}%
  \def\defarge{#5}%
  \def\defargf{#6}%
  \def\defargg{#7}%
  \def\defargh{#8}%
  \def\defargi{#9}%
  \tupleauxj%
  }
\newcommand{\tuplek}[9]
  {%
  \def\defarga{#1}%
  \def\defargb{#2}%
  \def\defargc{#3}%
  \def\defargd{#4}%
  \def\defarge{#5}%
  \def\defargf{#6}%
  \def\defargg{#7}%
  \def\defargh{#8}%
  \def\defargi{#9}%
  \tupleauxk%
  }
\newcommand{\tuplel}[9]
  {%
  \def\defarga{#1}%
  \def\defargb{#2}%
  \def\defargc{#3}%
  \def\defargd{#4}%
  \def\defarge{#5}%
  \def\defargf{#6}%
  \def\defargg{#7}%
  \def\defargh{#8}%
  \def\defargi{#9}%
  \tupleauxl%
  }
\newcommand{\tuplem}[9]
  {%
  \def\defarga{#1}%
  \def\defargb{#2}%
  \def\defargc{#3}%
  \def\defargd{#4}%
  \def\defarge{#5}%
  \def\defargf{#6}%
  \def\defargg{#7}%
  \def\defargh{#8}%
  \def\defargi{#9}%
  \tupleauxm%
  }
\newcommand{\tuplen}[9]
  {%
  \def\defarga{#1}%
  \def\defargb{#2}%
  \def\defargc{#3}%
  \def\defargd{#4}%
  \def\defarge{#5}%
  \def\defargf{#6}%
  \def\defargg{#7}%
  \def\defargh{#8}%
  \def\defargi{#9}%
  \tupleauxn%
  }
\newcommand{\tupleo}[9]
  {%
  \def\defarga{#1}%
  \def\defargb{#2}%
  \def\defargc{#3}%
  \def\defargd{#4}%
  \def\defarge{#5}%
  \def\defargf{#6}%
  \def\defargg{#7}%
  \def\defargh{#8}%
  \def\defargi{#9}%
  \tupleauxo%
  }
\newcommand{\tuplep}[9]
  {%
  \def\defarga{#1}%
  \def\defargb{#2}%
  \def\defargc{#3}%
  \def\defargd{#4}%
  \def\defarge{#5}%
  \def\defargf{#6}%
  \def\defargg{#7}%
  \def\defargh{#8}%
  \def\defargi{#9}%
  \tupleauxp%
  }
\newcommand{\tupleq}[9]
  {%
  \def\defarga{#1}%
  \def\defargb{#2}%
  \def\defargc{#3}%
  \def\defargd{#4}%
  \def\defarge{#5}%
  \def\defargf{#6}%
  \def\defargg{#7}%
  \def\defargh{#8}%
  \def\defargi{#9}%
  \tupleauxq%
  }
\newcommand{\tupler}[9]
  {%
  \def\defarga{#1}%
  \def\defargb{#2}%
  \def\defargc{#3}%
  \def\defargd{#4}%
  \def\defarge{#5}%
  \def\defargf{#6}%
  \def\defargg{#7}%
  \def\defargh{#8}%
  \def\defargi{#9}%
  \tupleauxr%
  }

\newcommand{\tupleauxj}[1]
  {%
  \tuple{\argj{\defarga}{\defargb}{\defargc}{\defargd}{\defarge}{\defargf}%
    {\defargg}{\defargh}{\defargi}{#1}}%
  }
\newcommand{\tupleauxk}[2]
  {%
  \tuple{\argk{\defarga}{\defargb}{\defargc}{\defargd}{\defarge}{\defargf}%
    {\defargg}{\defargh}{\defargi}{#1}{#2}}%
  }
\newcommand{\tupleauxl}[3]
  {%
  \tuple{\argl{\defarga}{\defargb}{\defargc}{\defargd}{\defarge}{\defargf}%
    {\defargg}{\defargh}{\defargi}{#1}{#2}{#3}}%
  }
\newcommand{\tupleauxm}[4]
  {%
  \tuple{\argm{\defarga}{\defargb}{\defargc}{\defargd}{\defarge}{\defargf}%
    {\defargg}{\defargh}{\defargi}{#1}{#2}{#3}{#4}}%
  }
\newcommand{\tupleauxn}[5]
  {%
  \tuple{\argn{\defarga}{\defargb}{\defargc}{\defargd}{\defarge}{\defargf}%
    {\defargg}{\defargh}{\defargi}{#1}{#2}{#3}{#4}{#5}}%
  }
\newcommand{\tupleauxo}[6]
  {%
  \tuple{\argo{\defarga}{\defargb}{\defargc}{\defargd}{\defarge}{\defargf}%
    {\defargg}{\defargh}{\defargi}{#1}{#2}{#3}{#4}{#5}{#6}}%
  }
\newcommand{\tupleauxp}[7]
  {%
  \tuple{\argp{\defarga}{\defargb}{\defargc}{\defargd}{\defarge}{\defargf}%
    {\defargg}{\defargh}{\defargi}{#1}{#2}{#3}{#4}{#5}{#6}{#7}}%
  }
\newcommand{\tupleauxq}[8]
  {%
  \tuple{\argq{\defarga}{\defargb}{\defargc}{\defargd}{\defarge}{\defargf}%
    {\defargg}{\defargh}{\defargi}{#1}{#2}{#3}{#4}{#5}{#6}{#7}{#8}}%
  }
\newcommand{\tupleauxr}[9]
  {%
  \tuple{\argr{\defarga}{\defargb}{\defargc}{\defargd}{\defarge}{\defargf}%
    {\defargg}{\defargh}{\defargi}{#1}{#2}{#3}{#4}{#5}{#6}{#7}{#8}{#9}}%
  }

\newcommandx{\tupleauxbx}[2][1=, 2=]
  {%
  \tupleb
    {\argdef{#1}{\defarga[\argsubscript][\argsuperscript]}}
    {\argdef{#2}{\defargb[\argsubscript][\argsuperscript]}}%
  }
\newcommandx{\tupleauxcx}[3][1=, 2=, 3=]
  {%
  \tuplec
    {\argdef{#1}{\defarga[\argsubscript][\argsuperscript]}}
    {\argdef{#2}{\defargb[\argsubscript][\argsuperscript]}}
    {\argdef{#3}{\defargc[\argsubscript][\argsuperscript]}}%
  }
\newcommandx{\tupleauxdx}[4][1=, 2=, 3=, 4=]
  {%
  \tupled
    {\argdef{#1}{\defarga[\argsubscript][\argsuperscript]}}
    {\argdef{#2}{\defargb[\argsubscript][\argsuperscript]}}
    {\argdef{#3}{\defargc[\argsubscript][\argsuperscript]}}
    {\argdef{#4}{\defargd[\argsubscript][\argsuperscript]}}%
  }
\newcommandx{\tupleauxex}[5][1=, 2=, 3=, 4=, 5=]
  {%
  \tuplee
    {\argdef{#1}{\defarga[\argsubscript][\argsuperscript]}}
    {\argdef{#2}{\defargb[\argsubscript][\argsuperscript]}}
    {\argdef{#3}{\defargc[\argsubscript][\argsuperscript]}}
    {\argdef{#4}{\defargd[\argsubscript][\argsuperscript]}}
    {\argdef{#5}{\defarge[\argsubscript][\argsuperscript]}}%
  }
\newcommandx{\tupleauxfx}[6][1=, 2=, 3=, 4=, 5=, 6=]
  {%
  \tuplef
    {\argdef{#1}{\defarga[\argsubscript][\argsuperscript]}}
    {\argdef{#2}{\defargb[\argsubscript][\argsuperscript]}}
    {\argdef{#3}{\defargc[\argsubscript][\argsuperscript]}}
    {\argdef{#4}{\defargd[\argsubscript][\argsuperscript]}}
    {\argdef{#5}{\defarge[\argsubscript][\argsuperscript]}}
    {\argdef{#6}{\defargf[\argsubscript][\argsuperscript]}}%
  }
\newcommandx{\tupleauxgx}[7][1=, 2=, 3=, 4=, 5=, 6=, 7=]
  {%
  \tupleg
    {\argdef{#1}{\defarga[\argsubscript][\argsuperscript]}}
    {\argdef{#2}{\defargb[\argsubscript][\argsuperscript]}}
    {\argdef{#3}{\defargc[\argsubscript][\argsuperscript]}}
    {\argdef{#4}{\defargd[\argsubscript][\argsuperscript]}}
    {\argdef{#5}{\defarge[\argsubscript][\argsuperscript]}}
    {\argdef{#6}{\defargf[\argsubscript][\argsuperscript]}}
    {\argdef{#7}{\defargg[\argsubscript][\argsuperscript]}}%
  }
\newcommandx{\tupleauxhx}[8][1=, 2=, 3=, 4=, 5=, 6=, 7=, 8=]
  {%
  \tupleh
    {\argdef{#1}{\defarga[\argsubscript][\argsuperscript]}}
    {\argdef{#2}{\defargb[\argsubscript][\argsuperscript]}}
    {\argdef{#3}{\defargc[\argsubscript][\argsuperscript]}}
    {\argdef{#4}{\defargd[\argsubscript][\argsuperscript]}}
    {\argdef{#5}{\defarge[\argsubscript][\argsuperscript]}}
    {\argdef{#6}{\defargf[\argsubscript][\argsuperscript]}}
    {\argdef{#7}{\defargg[\argsubscript][\argsuperscript]}}
    {\argdef{#8}{\defargh[\argsubscript][\argsuperscript]}}%
  }
\newcommandx{\tupleauxix}[9][1=, 2=, 3=, 4=, 5=, 6=, 7=, 8=, 9=]
  {%
  \tuplei
    {\argdef{#1}{\defarga[\argsubscript][\argsuperscript]}}
    {\argdef{#2}{\defargb[\argsubscript][\argsuperscript]}}
    {\argdef{#3}{\defargc[\argsubscript][\argsuperscript]}}
    {\argdef{#4}{\defargd[\argsubscript][\argsuperscript]}}
    {\argdef{#5}{\defarge[\argsubscript][\argsuperscript]}}
    {\argdef{#6}{\defargf[\argsubscript][\argsuperscript]}}
    {\argdef{#7}{\defargg[\argsubscript][\argsuperscript]}}
    {\argdef{#8}{\defargh[\argsubscript][\argsuperscript]}}
    {\argdef{#9}{\defargi[\argsubscript][\argsuperscript]}}%
  }

\newcommandx{\tupleauxxjx}[9][1=, 2=, 3=, 4=, 5=, 6=, 7=, 8=, 9=]
  {%
  \def\optarga{#1}%
  \def\optargb{#2}%
  \def\optargc{#3}%
  \def\optargd{#4}%
  \def\optarge{#5}%
  \def\optargf{#6}%
  \def\optargg{#7}%
  \def\optargh{#8}%
  \def\optargi{#9}%
  \tupleauxxxjx%
  }
\newcommandx{\tupleauxxkx}[9][1=, 2=, 3=, 4=, 5=, 6=, 7=, 8=, 9=]
  {%
  \def\optarga{#1}%
  \def\optargb{#2}%
  \def\optargc{#3}%
  \def\optargd{#4}%
  \def\optarge{#5}%
  \def\optargf{#6}%
  \def\optargg{#7}%
  \def\optargh{#8}%
  \def\optargi{#9}%
  \tupleauxxxkx%
  }
\newcommandx{\tupleauxxlx}[9][1=, 2=, 3=, 4=, 5=, 6=, 7=, 8=, 9=]
  {%
  \def\optarga{#1}%
  \def\optargb{#2}%
  \def\optargc{#3}%
  \def\optargd{#4}%
  \def\optarge{#5}%
  \def\optargf{#6}%
  \def\optargg{#7}%
  \def\optargh{#8}%
  \def\optargi{#9}%
  \tupleauxxxlx%
  }
\newcommandx{\tupleauxxmx}[9][1=, 2=, 3=, 4=, 5=, 6=, 7=, 8=, 9=]
  {%
  \def\optarga{#1}%
  \def\optargb{#2}%
  \def\optargc{#3}%
  \def\optargd{#4}%
  \def\optarge{#5}%
  \def\optargf{#6}%
  \def\optargg{#7}%
  \def\optargh{#8}%
  \def\optargi{#9}%
  \tupleauxxxmx%
  }
\newcommandx{\tupleauxxnx}[9][1=, 2=, 3=, 4=, 5=, 6=, 7=, 8=, 9=]
  {%
  \def\optarga{#1}%
  \def\optargb{#2}%
  \def\optargc{#3}%
  \def\optargd{#4}%
  \def\optarge{#5}%
  \def\optargf{#6}%
  \def\optargg{#7}%
  \def\optargh{#8}%
  \def\optargi{#9}%
  \tupleauxxxnx%
  }
\newcommandx{\tupleauxxox}[9][1=, 2=, 3=, 4=, 5=, 6=, 7=, 8=, 9=]
  {%
  \def\optarga{#1}%
  \def\optargb{#2}%
  \def\optargc{#3}%
  \def\optargd{#4}%
  \def\optarge{#5}%
  \def\optargf{#6}%
  \def\optargg{#7}%
  \def\optargh{#8}%
  \def\optargi{#9}%
  \tupleauxxxox%
  }
\newcommandx{\tupleauxxpx}[9][1=, 2=, 3=, 4=, 5=, 6=, 7=, 8=, 9=]
  {%
  \def\optarga{#1}%
  \def\optargb{#2}%
  \def\optargc{#3}%
  \def\optargd{#4}%
  \def\optarge{#5}%
  \def\optargf{#6}%
  \def\optargg{#7}%
  \def\optargh{#8}%
  \def\optargi{#9}%
  \tupleauxxxpx%
  }
\newcommandx{\tupleauxxqx}[9][1=, 2=, 3=, 4=, 5=, 6=, 7=, 8=, 9=]
  {%
  \def\optarga{#1}%
  \def\optargb{#2}%
  \def\optargc{#3}%
  \def\optargd{#4}%
  \def\optarge{#5}%
  \def\optargf{#6}%
  \def\optargg{#7}%
  \def\optargh{#8}%
  \def\optargi{#9}%
  \tupleauxxxqx%
  }
\newcommandx{\tupleauxxrx}[9][1=, 2=, 3=, 4=, 5=, 6=, 7=, 8=, 9=]
  {%
  \def\optarga{#1}%
  \def\optargb{#2}%
  \def\optargc{#3}%
  \def\optargd{#4}%
  \def\optarge{#5}%
  \def\optargf{#6}%
  \def\optargg{#7}%
  \def\optargh{#8}%
  \def\optargi{#9}%
  \tupleauxxxrx%
  }

\newcommandx{\tupleauxxxjx}[1][1=]
  {%
  \tuplej
    {\argdef{\optarga}{\tuplearga[\argsubscript][\argsuperscript]}}
    {\argdef{\optargb}{\tupleargb[\argsubscript][\argsuperscript]}}
    {\argdef{\optargc}{\tupleargc[\argsubscript][\argsuperscript]}}
    {\argdef{\optargd}{\tupleargd[\argsubscript][\argsuperscript]}}
    {\argdef{\optarge}{\tuplearge[\argsubscript][\argsuperscript]}}
    {\argdef{\optargf}{\tupleargf[\argsubscript][\argsuperscript]}}
    {\argdef{\optargg}{\tupleargg[\argsubscript][\argsuperscript]}}
    {\argdef{\optargh}{\tupleargh[\argsubscript][\argsuperscript]}}
    {\argdef{\optargi}{\tupleargi[\argsubscript][\argsuperscript]}}
    {\argdef{#1}{\tupleargj[\argsubscript][\argsuperscript]}}%
  }
\newcommandx{\tupleauxxxkx}[2][1=, 2=]
  {%
  \tuplek
    {\argdef{\optarga}{\tuplearga[\argsubscript][\argsuperscript]}}
    {\argdef{\optargb}{\tupleargb[\argsubscript][\argsuperscript]}}
    {\argdef{\optargc}{\tupleargc[\argsubscript][\argsuperscript]}}
    {\argdef{\optargd}{\tupleargd[\argsubscript][\argsuperscript]}}
    {\argdef{\optarge}{\tuplearge[\argsubscript][\argsuperscript]}}
    {\argdef{\optargf}{\tupleargf[\argsubscript][\argsuperscript]}}
    {\argdef{\optargg}{\tupleargg[\argsubscript][\argsuperscript]}}
    {\argdef{\optargh}{\tupleargh[\argsubscript][\argsuperscript]}}
    {\argdef{\optargi}{\tupleargi[\argsubscript][\argsuperscript]}}
    {\argdef{#1}{\tupleargj[\argsubscript][\argsuperscript]}}
    {\argdef{#2}{\tupleargk[\argsubscript][\argsuperscript]}}
  }
\newcommandx{\tupleauxxxlx}[3][1=, 2=, 3=]
  {%
  \tuplel
    {\argdef{\optarga}{\tuplearga[\argsubscript][\argsuperscript]}}
    {\argdef{\optargb}{\tupleargb[\argsubscript][\argsuperscript]}}
    {\argdef{\optargc}{\tupleargc[\argsubscript][\argsuperscript]}}
    {\argdef{\optargd}{\tupleargd[\argsubscript][\argsuperscript]}}
    {\argdef{\optarge}{\tuplearge[\argsubscript][\argsuperscript]}}
    {\argdef{\optargf}{\tupleargf[\argsubscript][\argsuperscript]}}
    {\argdef{\optargg}{\tupleargg[\argsubscript][\argsuperscript]}}
    {\argdef{\optargh}{\tupleargh[\argsubscript][\argsuperscript]}}
    {\argdef{\optargi}{\tupleargi[\argsubscript][\argsuperscript]}}
    {\argdef{#1}{\tupleargj[\argsubscript][\argsuperscript]}}
    {\argdef{#2}{\tupleargk[\argsubscript][\argsuperscript]}}
    {\argdef{#3}{\tupleargl[\argsubscript][\argsuperscript]}}
  }
\newcommandx{\tupleauxxxmx}[4][1=, 2=, 3=, 4=]
  {%
  \tuplem
    {\argdef{\optarga}{\tuplearga[\argsubscript][\argsuperscript]}}
    {\argdef{\optargb}{\tupleargb[\argsubscript][\argsuperscript]}}
    {\argdef{\optargc}{\tupleargc[\argsubscript][\argsuperscript]}}
    {\argdef{\optargd}{\tupleargd[\argsubscript][\argsuperscript]}}
    {\argdef{\optarge}{\tuplearge[\argsubscript][\argsuperscript]}}
    {\argdef{\optargf}{\tupleargf[\argsubscript][\argsuperscript]}}
    {\argdef{\optargg}{\tupleargg[\argsubscript][\argsuperscript]}}
    {\argdef{\optargh}{\tupleargh[\argsubscript][\argsuperscript]}}
    {\argdef{\optargi}{\tupleargi[\argsubscript][\argsuperscript]}}
    {\argdef{#1}{\tupleargj[\argsubscript][\argsuperscript]}}
    {\argdef{#2}{\tupleargk[\argsubscript][\argsuperscript]}}
    {\argdef{#3}{\tupleargl[\argsubscript][\argsuperscript]}}
    {\argdef{#4}{\tupleargm[\argsubscript][\argsuperscript]}}
  }
\newcommandx{\tupleauxxxnx}[5][1=, 2=, 3=, 4=, 5=]
  {%
  \tuplen
    {\argdef{\optarga}{\tuplearga[\argsubscript][\argsuperscript]}}
    {\argdef{\optargb}{\tupleargb[\argsubscript][\argsuperscript]}}
    {\argdef{\optargc}{\tupleargc[\argsubscript][\argsuperscript]}}
    {\argdef{\optargd}{\tupleargd[\argsubscript][\argsuperscript]}}
    {\argdef{\optarge}{\tuplearge[\argsubscript][\argsuperscript]}}
    {\argdef{\optargf}{\tupleargf[\argsubscript][\argsuperscript]}}
    {\argdef{\optargg}{\tupleargg[\argsubscript][\argsuperscript]}}
    {\argdef{\optargh}{\tupleargh[\argsubscript][\argsuperscript]}}
    {\argdef{\optargi}{\tupleargi[\argsubscript][\argsuperscript]}}
    {\argdef{#1}{\tupleargj[\argsubscript][\argsuperscript]}}
    {\argdef{#2}{\tupleargk[\argsubscript][\argsuperscript]}}
    {\argdef{#3}{\tupleargl[\argsubscript][\argsuperscript]}}
    {\argdef{#4}{\tupleargm[\argsubscript][\argsuperscript]}}
    {\argdef{#5}{\tupleargn[\argsubscript][\argsuperscript]}}
  }
\newcommandx{\tupleauxxxox}[6][1=, 2=, 3=, 4=, 5=, 6=]
  {%
  \tupleo
    {\argdef{\optarga}{\tuplearga[\argsubscript][\argsuperscript]}}
    {\argdef{\optargb}{\tupleargb[\argsubscript][\argsuperscript]}}
    {\argdef{\optargc}{\tupleargc[\argsubscript][\argsuperscript]}}
    {\argdef{\optargd}{\tupleargd[\argsubscript][\argsuperscript]}}
    {\argdef{\optarge}{\tuplearge[\argsubscript][\argsuperscript]}}
    {\argdef{\optargf}{\tupleargf[\argsubscript][\argsuperscript]}}
    {\argdef{\optargg}{\tupleargg[\argsubscript][\argsuperscript]}}
    {\argdef{\optargh}{\tupleargh[\argsubscript][\argsuperscript]}}
    {\argdef{\optargi}{\tupleargi[\argsubscript][\argsuperscript]}}
    {\argdef{#1}{\tupleargj[\argsubscript][\argsuperscript]}}
    {\argdef{#2}{\tupleargk[\argsubscript][\argsuperscript]}}
    {\argdef{#3}{\tupleargl[\argsubscript][\argsuperscript]}}
    {\argdef{#4}{\tupleargm[\argsubscript][\argsuperscript]}}
    {\argdef{#5}{\tupleargn[\argsubscript][\argsuperscript]}}
    {\argdef{#6}{\tupleargo[\argsubscript][\argsuperscript]}}
  }
\newcommandx{\tupleauxxxpx}[7][1=, 2=, 3=, 4=, 5=, 6=, 7=]
  {%
  \tuplep
    {\argdef{\optarga}{\tuplearga[\argsubscript][\argsuperscript]}}
    {\argdef{\optargb}{\tupleargb[\argsubscript][\argsuperscript]}}
    {\argdef{\optargc}{\tupleargc[\argsubscript][\argsuperscript]}}
    {\argdef{\optargd}{\tupleargd[\argsubscript][\argsuperscript]}}
    {\argdef{\optarge}{\tuplearge[\argsubscript][\argsuperscript]}}
    {\argdef{\optargf}{\tupleargf[\argsubscript][\argsuperscript]}}
    {\argdef{\optargg}{\tupleargg[\argsubscript][\argsuperscript]}}
    {\argdef{\optargh}{\tupleargh[\argsubscript][\argsuperscript]}}
    {\argdef{\optargi}{\tupleargi[\argsubscript][\argsuperscript]}}
    {\argdef{#1}{\tupleargj[\argsubscript][\argsuperscript]}}
    {\argdef{#2}{\tupleargk[\argsubscript][\argsuperscript]}}
    {\argdef{#3}{\tupleargl[\argsubscript][\argsuperscript]}}
    {\argdef{#4}{\tupleargm[\argsubscript][\argsuperscript]}}
    {\argdef{#5}{\tupleargn[\argsubscript][\argsuperscript]}}
    {\argdef{#6}{\tupleargo[\argsubscript][\argsuperscript]}}
    {\argdef{#7}{\tupleargp[\argsubscript][\argsuperscript]}}
  }
\newcommandx{\tupleauxxxqx}[8][1=, 2=, 3=, 4=, 5=, 6=, 7=, 8=]
  {%
  \tupleq
    {\argdef{\optarga}{\tuplearga[\argsubscript][\argsuperscript]}}
    {\argdef{\optargb}{\tupleargb[\argsubscript][\argsuperscript]}}
    {\argdef{\optargc}{\tupleargc[\argsubscript][\argsuperscript]}}
    {\argdef{\optargd}{\tupleargd[\argsubscript][\argsuperscript]}}
    {\argdef{\optarge}{\tuplearge[\argsubscript][\argsuperscript]}}
    {\argdef{\optargf}{\tupleargf[\argsubscript][\argsuperscript]}}
    {\argdef{\optargg}{\tupleargg[\argsubscript][\argsuperscript]}}
    {\argdef{\optargh}{\tupleargh[\argsubscript][\argsuperscript]}}
    {\argdef{\optargi}{\tupleargi[\argsubscript][\argsuperscript]}}
    {\argdef{#1}{\tupleargj[\argsubscript][\argsuperscript]}}
    {\argdef{#2}{\tupleargk[\argsubscript][\argsuperscript]}}
    {\argdef{#3}{\tupleargl[\argsubscript][\argsuperscript]}}
    {\argdef{#4}{\tupleargm[\argsubscript][\argsuperscript]}}
    {\argdef{#5}{\tupleargn[\argsubscript][\argsuperscript]}}
    {\argdef{#6}{\tupleargo[\argsubscript][\argsuperscript]}}
    {\argdef{#7}{\tupleargp[\argsubscript][\argsuperscript]}}
    {\argdef{#8}{\tupleargq[\argsubscript][\argsuperscript]}}
  }
\newcommandx{\tupleauxxxrx}[9][1=, 2=, 3=, 4=, 5=, 6=, 7=, 8=, 9=]
  {%
  \tupler
    {\argdef{\optarga}{\tuplearga[\argsubscript][\argsuperscript]}}
    {\argdef{\optargb}{\tupleargb[\argsubscript][\argsuperscript]}}
    {\argdef{\optargc}{\tupleargc[\argsubscript][\argsuperscript]}}
    {\argdef{\optargd}{\tupleargd[\argsubscript][\argsuperscript]}}
    {\argdef{\optarge}{\tuplearge[\argsubscript][\argsuperscript]}}
    {\argdef{\optargf}{\tupleargf[\argsubscript][\argsuperscript]}}
    {\argdef{\optargg}{\tupleargg[\argsubscript][\argsuperscript]}}
    {\argdef{\optargh}{\tupleargh[\argsubscript][\argsuperscript]}}
    {\argdef{\optargi}{\tupleargi[\argsubscript][\argsuperscript]}}
    {\argdef{#1}{\tupleargj[\argsubscript][\argsuperscript]}}
    {\argdef{#2}{\tupleargk[\argsubscript][\argsuperscript]}}
    {\argdef{#3}{\tupleargl[\argsubscript][\argsuperscript]}}
    {\argdef{#4}{\tupleargm[\argsubscript][\argsuperscript]}}
    {\argdef{#5}{\tupleargn[\argsubscript][\argsuperscript]}}
    {\argdef{#6}{\tupleargo[\argsubscript][\argsuperscript]}}
    {\argdef{#7}{\tupleargp[\argsubscript][\argsuperscript]}}
    {\argdef{#8}{\tupleargq[\argsubscript][\argsuperscript]}}
    {\argdef{#9}{\tupleargr[\argsubscript][\argsuperscript]}}%
  }

\newcommand{\set}[2]
  {\ensuremath{\argint{\{}{\argext{#1}{\allowbreak:\allowbreak}{#2}}{\}}}}

\newcommand{\pow}[1]
  {\ensuremath{2^{#1}}}

\newcommand{\card}[1]
  {\mthempty{\argint{\vert}{#1}{\vert}}}

\newcommand{\dom}
  {\mthargfun{dom}}

\newcommand{\cmp}
  {\ensuremath{\circ}}

\newcommand{\rst}
  {\mthempty{\upharpoonright}}

\newcommandx{\pto}[2][1=, 2=]
  {\ensuremath{\rightharpoonup}}

\newcommandx{\cto}[2][1=, 2=]
  {\:\mthempty{\to}[#1][#2]\:}
\newcommandx{\cpto}[2][1=, 2=]
  {\:\mthempty{\pto}[#1][#2]\:}

\newcommand{\emptyfun}
  {\mthempty{\varnothing}}

\newcommand{\AOmicron}
  {\mthargset{O}}

\newcommand{\SetN}
  {\mthset[1]{N}}
\newcommand{\SetNI}
  {\mthset[1]{N}[\infty]}
\newcommand{\SetZ}
  {\mthset[1]{Z}}

\newcommand{\numcc}[2]
  {\mthempty{[\argb{#1}{#2}]}}
\newcommand{\numco}[2]
  {\mthempty{[\argb{#1}{#2})}}

\DeclareRobustCommand{\sup}
  {\mthfun{sup}}
\DeclareRobustCommand{\min}
  {\mthfun{min}}
\DeclareRobustCommand{\max}
  {\mthfun{max}}

\DeclareMathOperator*{\argmax}
  {\mthfun{arg max}}
\DeclareMathOperator*{\argmin}
  {\mthfun{arg min}}

\newcommandx{\EF}[5][1=, 2=, 3=, 4=, 5=]
  {\txtargname{EF#5{\small\argint{$[$}{#1}{$]$}}}[#2][#3]{#4}}

\newcommandx{\BG}[5][1=, 2=, 3=, 4=, 5=]
  {\txtargname{BG#5{\small\argint{$[$}{#1}{$]$}}}[#2][#3]{#4}}

\newcommandx{\CG}[5][1=, 2=, 3=, 4=, 5=]
  {\txtargname{CG#5{\small\argint{$[$}{#1}{$]$}}}[#2][#3]{#4}}

\newcommandx{\PG}[5][1=, 2=, 3=, 4=, 5=]
  {\txtargname{PG#5{\small\argint{$[$}{#1}{$]$}}}[#2][#3]{#4}}

\newcommandx{\RG}[5][1=, 2=, 3=, 4=, 5=]
  {\txtargname{RG#5{\small\argint{$[$}{#1}{$]$}}}[#2][#3]{#4}}

\newcommandx{\SG}[5][1=, 2=, 3=, 4=, 5=]
  {\txtargname{SG#5{\small\argint{$[$}{#1}{$]$}}}[#2][#3]{#4}}

\newcommandx{\MG}[5][1=, 2=, 3=, 4=, 5=]
  {\txtargname{MG#5{\small\argint{$[$}{#1}{$]$}}}[#2][#3]{#4}}

\newcommandx{\EG}[5][1=, 2=, 3=, 4=, 5=]
  {\txtargname{EG#5{\small\argint{$[$}{#1}{$]$}}}[#2][#3]{#4}}

\newcommandx{\MPG}[5][1=, 2=, 3=, 4=, 5=]
  {\txtargname{MPG#5{\small\argint{$[$}{#1}{$]$}}}[#2][#3]{#4}}

\newcommandx{\DPG}[5][1=, 2=, 3=, 4=, 5=]
  {\txtargname{DPG#5{\small\argint{$[$}{#1}{$]$}}}[#2][#3]{#4}}

\renewcommandx{\SG}[5][1=, 2=, 3=, 4=, 5=]
  {\txtargname{SG#5{\small\argint{$[$}{#1}{$]$}}}[#2][#3]{#4}}

\newcommandx{\BF}[5][1=, 2=, 3=, 4=, 5=]
  {\txtargname{BF#5{\small\argint{$[$}{#1}{$]$}}}[#2][#3]{#4}}

\newcommandx{\QBF}
  {{\txtname{Q}}\BF}

\newcommandx{\FOL}[5][1=, 2=, 3=, 4=, 5=]
  {\txtargname{FOL#5{\small\argint{$[$}{#1}{$]$}}}[#2][#3]{#4}}

\newcommandx{\SOL}[5][1=, 2=, 3=, 4=, 5=]
  {\txtargname{SOL#5{\small\argint{$[$}{#1}{$]$}}}[#2][#3]{#4}}

\newcommandx{\TL}[5][1=, 2=, 3=, 4=, 5=]
  {\txtargname{TL#5{\small\argint{$[$}{#1}{$]$}}}[#2][#3]{#4}}

\newcommandx{\PL}[5][1=, 2=, 3=, 4=, 5=]
  {\txtargname{PL#5{\small\argint{$[$}{#1}{$]$}}}[#2][#3]{#4}}

\newcommandx{\ML}[5][1=, 2=, 3=, 4=, 5=]
  {\txtargname{ML#5{\small\argint{$[$}{#1}{$]$}}}[#2][#3]{#4}}

\newcommandx{\MC}[5][1=, 2=, 3=, 4=, 5=]
  {\txtargname{$\mu$Calculus#5{\small\argint{$[$}{#1}{$]$}}}[#2][#3]{#4}}

\newcommandx{\LTL}[5][1=, 2=, 3=, 4=, 5=]
  {\txtargname{LTL#5{\small\argint{$[$}{#1}{$]$}}}[#2][#3]{#4}}

\newcommandx{\PTL}[5][1=, 2=, 3=, 4=, 5=]
  {\txtargname{PTL#5{\small\argint{$[$}{#1}{$]$}}}[#2][#3]{#4}}

\newcommandx{\CTL}[5][1=, 2=, 3=, 4=, 5=]
  {\txtargname{CTL#5{\small\argint{$[$}{#1}{$]$}}}[#2][#3]{#4}}

\newcommandx{\CTLP}[5][1=, 2=, 3=, 4=, 5=]
  {\txtargname{CTL$^{+}$#5{\small\argint{$[$}{#1}{$]$}}}[#2][#3]{#4}}

\newcommandx{\CTLS}[5][1=, 2=, 3=, 4=, 5=]
  {\txtargname{CTL$^{\star}$#5{\small\argint{$[$}{#1}{$]$}}}[#2][#3]{#4}}

\newcommandx{\STL}[5][1=, 2=, 3=, 4=, 5=]
  {\txtargname{STL#5{\small\argint{$[$}{#1}{$]$}}}[#2][#3]{#4}}

\newcommandx{\STLP}[5][1=, 2=, 3=, 4=, 5=]
  {\txtargname{STL$^{+}$#5{\small\argint{$[$}{#1}{$]$}}}[#2][#3]{#4}}

\newcommandx{\STLS}[5][1=, 2=, 3=, 4=, 5=]
  {\txtargname{STL$^{\star}$#5{\small\argint{$[$}{#1}{$]$}}}[#2][#3]{#4}}

\newcommandx{\ATL}[5][1=, 2=, 3=, 4=, 5=]
  {\txtargname{ATL#5{\small\argint{$[$}{#1}{$]$}}}[#2][#3]{#4}}

\newcommandx{\ATLP}[5][1=, 2=, 3=, 4=, 5=]
  {\txtargname{ATL$^{+}$#5{\small\argint{$[$}{#1}{$]$}}}[#2][#3]{#4}}

\newcommandx{\ATLS}[5][1=, 2=, 3=, 4=, 5=]
  {\txtargname{ATL$^{\star}$#5{\small\argint{$[$}{#1}{$]$}}}[#2][#3]{#4}}

\newcommandx{\SL}[5][1=, 2=, 3=, 4=, 5=]
  {\txtargname{SL#5{\small\argint{$[$}{#1}{$]$}}}[#2][#3]{#4}}

\newcommandx{\LogTime}[4][1=, 2=, 3=, 4=]
  {\txtargname{LogTime#4}[#2][#3]{#1}}
\newcommandx{\LogTimeE}[4][1=, 2=, 3=, 4=]
  {\LogTime[#1][#2][#3][#4]-\EComplexity}
\newcommandx{\LogTimeH}[4][1=, 2=, 3=, 4=]
  {\LogTime[#1][#2][#3][#4]-\HComplexity}
\newcommandx{\LogTimeC}[4][1=, 2=, 3=, 4=]
  {\LogTime[#1][#2][#3][#4]-\CComplexity}

\newcommandx{\LogSpace}[4][1=, 2=, 3=, 4=]
  {\txtargname{LogSpace#4}[#2][#3]{#1}}
\newcommandx{\LogSpaceE}[4][1=, 2=, 3=, 4=]
  {\LogSpace[#1][#2][#3][#4]-\EComplexity}
\newcommandx{\LogSpaceH}[4][1=, 2=, 3=, 4=]
  {\LogSpace[#1][#2][#3][#4]-\HComplexity}
\newcommandx{\LogSpaceC}[4][1=, 2=, 3=, 4=]
  {\LogSpace[#1][#2][#3][#4]-\CComplexity}

\newcommandx{\PTime}[4][1=, 2=, 3=, 4=]
  {\txtargname{PTime#4}[#2][#3]{#1}}
\newcommandx{\PTimeE}[4][1=, 2=, 3=, 4=]
  {\PTime[#1][#2][#3][#4]-\EComplexity}
\newcommandx{\PTimeH}[4][1=, 2=, 3=, 4=]
  {\PTime[#1][#2][#3][#4]-\HComplexity}
\newcommandx{\PTimeC}[4][1=, 2=, 3=, 4=]
  {\PTime[#1][#2][#3][#4]-\CComplexity}

\newcommand{\UPTime}
  {\UComplexity\PTime}

\newcommand{\CoUPTime}
  {\CoComplexity\UPTime}

\newcommand{\NPTime}
  {\NComplexity\PTime}

\newcommand{\CoNPTime}
  {\CoComplexity\NPTime}

\newcommandx{\PSpace}[4][1=, 2=, 3=, 4=]
  {\txtargname{PSpace#4}[#2][#3]{#1}}
\newcommandx{\PSpaceE}[4][1=, 2=, 3=, 4=]
  {\PSpace[#1][#2][#3][#4]-\EComplexity}
\newcommandx{\PSpaceH}[4][1=, 2=, 3=, 4=]
  {\PSpace[#1][#2][#3][#4]-\HComplexity}
\newcommandx{\PSpaceC}[4][1=, 2=, 3=, 4=]
  {\PSpace[#1][#2][#3][#4]-\CComplexity}

\newcommandx{\ExpTime}[4][1=, 2=, 3=, 4=]
  {\txtargname{ExpTime#4}[#2][#3]{#1}}
\newcommandx{\ExpTimeE}[4][1=, 2=, 3=, 4=]
  {\ExpTime[#1][#2][#3][#4]-\EComplexity}
\newcommandx{\ExpTimeH}[4][1=, 2=, 3=, 4=]
  {\ExpTime[#1][#2][#3][#4]-\HComplexity}
\newcommandx{\ExpTimeC}[4][1=, 2=, 3=, 4=]
  {\ExpTime[#1][#2][#3][#4]-\CComplexity}

\newcommandx{\ExpSpace}[4][1=, 2=, 3=, 4=]
  {\txtargname{ExpSpace#4}[#2][#3]{#1}}
\newcommandx{\ExpSpaceE}[4][1=, 2=, 3=, 4=]
  {\ExpSpace[#1][#2][#3][#4]-\EComplexity}
\newcommandx{\ExpSpaceH}[4][1=, 2=, 3=, 4=]
  {\ExpSpace[#1][#2][#3][#4]-\HComplexity}
\newcommandx{\ExpSpaceC}[4][1=, 2=, 3=, 4=]
  {\ExpSpace[#1][#2][#3][#4]-\CComplexity}

\newcommandx{\NonElm}[4][1=, 2=, 3=, 4=]
  {\txtargname{NonElementary#4}[#2][#3]{#1}}
\newcommandx{\NonElmE}[4][1=, 2=, 3=, 4=]
  {\NonElm[#1][#2][#3][#4]-\EComplexity}
\newcommandx{\NonElmH}[4][1=, 2=, 3=, 4=]
  {\NonElm[#1][#2][#3][#4]-\HComplexity}
\newcommandx{\NonElmC}[4][1=, 2=, 3=, 4=]
  {\NonElm[#1][#2][#3][#4]-\CComplexity}

\newcommandx{\NonElmTime}[4][1=, 2=, 3=, 4=]
  {\txtargname{NonElementaryTime#4}[#2][#3]{#1}}
\newcommandx{\NonElmTimeE}[4][1=, 2=, 3=, 4=]
  {\NonElmTime[#1][#2][#3][#4]-\EComplexity}
\newcommandx{\NonElmTimeH}[4][1=, 2=, 3=, 4=]
  {\NonElmTime[#1][#2][#3][#4]-\HComplexity}
\newcommandx{\NonElmTimeC}[4][1=, 2=, 3=, 4=]
  {\NonElmTime[#1][#2][#3][#4]-\CComplexity}

\newcommandx{\NonElmSpace}[4][1=, 2=, 3=, 4=]
  {\txtargname{NonElementarySpace#4}[#2][#3]{#1}}
\newcommandx{\NonElmSpaceE}[4][1=, 2=, 3=, 4=]
  {\NonElmSpace[#1][#2][#3][#4]-\EComplexity}
\newcommandx{\NonElmSpaceH}[4][1=, 2=, 3=, 4=]
  {\NonElmSpace[#1][#2][#3][#4]-\HComplexity}
\newcommandx{\NonElmSpaceC}[4][1=, 2=, 3=, 4=]
  {\NonElmSpace[#1][#2][#3][#4]-\CComplexity}

\newcommandx{\DLHier}[4][2=, 3=, 4=]
  {\mthargset[0]{\Delta#4}[#1][#3]{#2}}
\newcommandx{\DLHierE}[4][2=, 3=, 4=]
  {\DLHier{#1}[#2][#3][#4]-\EComplexity}
\newcommandx{\DLHierH}[4][2=, 3=, 4=]
  {\DLHier{#1}[#2][#3][#4]-\HComplexity}
\newcommandx{\DLHierC}[4][2=, 3=, 4=]
  {\DLHier{#1}[#2][#3][#4]-\CComplexity}

\newcommandx{\ELHier}[4][2=, 3=, 4=]
  {\mthargset[0]{\Sigma#4}[#1][#3]{#2}}
\newcommandx{\ELHierE}[4][2=, 3=, 4=]
  {\ELHier{#1}[#2][#3][#4]-\EComplexity}
\newcommandx{\ELHierH}[4][2=, 3=, 4=]
  {\ELHier{#1}[#2][#3][#4]-\HComplexity}
\newcommandx{\ELHierC}[4][2=, 3=, 4=]
  {\ELHier{#1}[#2][#3][#4]-\CComplexity}

\newcommandx{\ULHier}[4][2=, 3=, 4=]
  {\mthargset[0]{\Pi#4}[#1][#3]{#2}}
\newcommandx{\ULHierE}[4][2=, 3=, 4=]
  {\ULHier{#1}[#2][#3][#4]-\EComplexity}
\newcommandx{\ULHierH}[4][2=, 3=, 4=]
  {\ULHier{#1}[#2][#3][#4]-\HComplexity}
\newcommandx{\ULHierC}[4][2=, 3=, 4=]
  {\ULHier{#1}[#2][#3][#4]-\CComplexity}

\newcommandx{\DBHier}[4][2=, 3=, 4=]
  {\mthargset[3]{\Delta#4}[#1][#3]{#2}}
\newcommandx{\DBHierE}[4][2=, 3=, 4=]
  {\DBHier{#1}[#2][#3][#4]-\EComplexity}
\newcommandx{\DBHierH}[4][2=, 3=, 4=]
  {\DBHier{#1}[#2][#3][#4]-\HComplexity}
\newcommandx{\DBHierC}[4][2=, 3=, 4=]
  {\DBHier{#1}[#2][#3][#4]-\CComplexity}

\newcommandx{\EBHier}[4][2=, 3=, 4=]
  {\mthargset[3]{\Sigma#4}[#1][#3]{#2}}
\newcommandx{\EBHierE}[4][2=, 3=, 4=]
  {\EBHier{#1}[#2][#3][#4]-\EComplexity}
\newcommandx{\EBHierH}[4][2=, 3=, 4=]
  {\EBHier{#1}[#2][#3][#4]-\HComplexity}
\newcommandx{\EBHierC}[4][2=, 3=, 4=]
  {\EBHier{#1}[#2][#3][#4]-\CComplexity}

\newcommandx{\UBHier}[4][2=, 3=, 4=]
  {\mthargset[3]{\Pi#4}[#1][#3]{#2}}
\newcommandx{\UBHierE}[4][2=, 3=, 4=]
  {\UBHier{#1}[#2][#3][#4]-\EComplexity}
\newcommandx{\UBHierH}[4][2=, 3=, 4=]
  {\UBHier{#1}[#2][#3][#4]-\HComplexity}
\newcommandx{\UBHierC}[4][2=, 3=, 4=]
  {\UBHier{#1}[#2][#3][#4]-\CComplexity}

\newcommandx{\DPolHier}[4][2=, 3=, 4=]
  {\DLHier{#1}[#2][\argb{\mathrm{P}}{#3}][#4]}
\newcommandx{\DPolHierE}[4][2=, 3=, 4=]
  {\DPolHier{#1}[#2][#3][#4]-\EComplexity}
\newcommandx{\DPolHierH}[4][2=, 3=, 4=]
  {\DPolHier{#1}[#2][#3][#4]-\HComplexity}
\newcommandx{\DPolHierC}[4][2=, 3=, 4=]
  {\DPolHier{#1}[#2][#3][#4]-\CComplexity}

\newcommandx{\EPolHier}[4][2=, 3=, 4=]
  {\ELHier{#1}[#2][\argb{\mathrm{P}}{#3}][#4]}
\newcommandx{\EPolHierE}[4][2=, 3=, 4=]
  {\EPolHier{#1}[#2][#3][#4]-\EComplexity}
\newcommandx{\EPolHierH}[4][2=, 3=, 4=]
  {\EPolHier{#1}[#2][#3][#4]-\HComplexity}
\newcommandx{\EPolHierC}[4][2=, 3=, 4=]
  {\EPolHier{#1}[#2][#3][#4]-\CComplexity}

\newcommandx{\UPolHier}[4][2=, 3=, 4=]
  {\ULHier{#1}[#2][\argb{\mathrm{P}}{#3}][#4]}
\newcommandx{\UPolHierE}[4][2=, 3=, 4=]
  {\UPolHier{#1}[#2][#3][#4]-\EComplexity}
\newcommandx{\UPolHierH}[4][2=, 3=, 4=]
  {\UPolHier{#1}[#2][#3][#4]-\HComplexity}
\newcommandx{\UPolHierC}[4][2=, 3=, 4=]
  {\UPolHier{#1}[#2][#3][#4]-\CComplexity}

\newcommandx{\DAriHier}[4][2=, 3=, 4=]
  {\DLHier{#1}[#2][\argb{0}{#3}][#4]}
\newcommandx{\DAriHierE}[4][2=, 3=, 4=]
  {\DAriHier{#1}[#2][#3][#4]-\EComplexity}
\newcommandx{\DAriHierH}[4][2=, 3=, 4=]
  {\DAriHier{#1}[#2][#3][#4]-\HComplexity}
\newcommandx{\DAriHierC}[4][2=, 3=, 4=]
  {\DAriHier{#1}[#2][#3][#4]-\CComplexity}

\newcommandx{\EAriHier}[4][2=, 3=, 4=]
  {\ELHier{#1}[#2][\argb{0}{#3}][#4]}
\newcommandx{\EAriHierE}[4][2=, 3=, 4=]
  {\EAriHier{#1}[#2][#3][#4]-\EComplexity}
\newcommandx{\EAriHierH}[4][2=, 3=, 4=]
  {\EAriHier{#1}[#2][#3][#4]-\HComplexity}
\newcommandx{\EAriHierC}[4][2=, 3=, 4=]
  {\EAriHier{#1}[#2][#3][#4]-\CComplexity}

\newcommandx{\UAriHier}[4][2=, 3=, 4=]
  {\ULHier{#1}[#2][\argb{0}{#3}][#4]}
\newcommandx{\UAriHierE}[4][2=, 3=, 4=]
  {\UAriHier{#1}[#2][#3][#4]-\EComplexity}
\newcommandx{\UAriHierH}[4][2=, 3=, 4=]
  {\UAriHier{#1}[#2][#3][#4]-\HComplexity}
\newcommandx{\UAriHierC}[4][2=, 3=, 4=]
  {\UAriHier{#1}[#2][#3][#4]-\CComplexity}

\newcommandx{\DAnaHier}[4][2=, 3=, 4=]
  {\DLHier{#1}[#2][\argb{1}{#3}][#4]}
\newcommandx{\DAnaHierE}[4][2=, 3=, 4=]
  {\DAnaHier{#1}[#2][#3][#4]-\EComplexity}
\newcommandx{\DAnaHierH}[4][2=, 3=, 4=]
  {\DAnaHier{#1}[#2][#3][#4]-\HComplexity}
\newcommandx{\DAnaHierC}[4][2=, 3=, 4=]
  {\DAnaHier{#1}[#2][#3][#4]-\CComplexity}

\newcommandx{\EAnaHier}[4][2=, 3=, 4=]
  {\ELHier{#1}[#2][\argb{1}{#3}][#4]}
\newcommandx{\EAnaHierE}[4][2=, 3=, 4=]
  {\EAnaHier{#1}[#2][#3][#4]-\EComplexity}
\newcommandx{\EAnaHierH}[4][2=, 3=, 4=]
  {\EAnaHier{#1}[#2][#3][#4]-\HComplexity}
\newcommandx{\EAnaHierC}[4][2=, 3=, 4=]
  {\EAnaHier{#1}[#2][#3][#4]-\CComplexity}

\newcommandx{\UAnaHier}[4][2=, 3=, 4=]
  {\ULHier{#1}[#2][\argb{1}{#3}][#4]}
\newcommandx{\UAnaHierE}[4][2=, 3=, 4=]
  {\UAnaHier{#1}[#2][#3][#4]-\EComplexity}
\newcommandx{\UAnaHierH}[4][2=, 3=, 4=]
  {\UAnaHier{#1}[#2][#3][#4]-\HComplexity}
\newcommandx{\UAnaHierC}[4][2=, 3=, 4=]
  {\UAnaHier{#1}[#2][#3][#4]-\CComplexity}

\newcommandx{\DBorHier}[4][2=, 3=, 4=]
  {\DBHier{#1}[#2][\argb{\mathrm{B}}{#3}][#4]}
\newcommandx{\DBorHierE}[4][2=, 3=, 4=]
  {\DBorHier{#1}[#2][#3][#4]-\EComplexity}
\newcommandx{\DBorHierH}[4][2=, 3=, 4=]
  {\DBorHier{#1}[#2][#3][#4]-\HComplexity}
\newcommandx{\DBorHierC}[4][2=, 3=, 4=]
  {\DBorHier{#1}[#2][#3][#4]-\CComplexity}

\newcommandx{\EBorHier}[4][2=, 3=, 4=]
  {\EBHier{#1}[#2][\argb{\mathrm{B}}{#3}][#4]}
\newcommandx{\EBorHierE}[4][2=, 3=, 4=]
  {\EBorHier{#1}[#2][#3][#4]-\EComplexity}
\newcommandx{\EBorHierH}[4][2=, 3=, 4=]
  {\EBorHier{#1}[#2][#3][#4]-\HComplexity}
\newcommandx{\EBorHierC}[4][2=, 3=, 4=]
  {\EBorHier{#1}[#2][#3][#4]-\CComplexity}

\newcommandx{\UBorHier}[4][2=, 3=, 4=]
  {\UBHier{#1}[#2][\argb{\mathrm{B}}{#3}][#4]}
\newcommandx{\UBorHierE}[4][2=, 3=, 4=]
  {\UBorHier{#1}[#2][#3][#4]-\EComplexity}
\newcommandx{\UBorHierH}[4][2=, 3=, 4=]
  {\UBorHier{#1}[#2][#3][#4]-\HComplexity}
\newcommandx{\UBorHierC}[4][2=, 3=, 4=]
  {\UBorHier{#1}[#2][#3][#4]-\CComplexity}

\newcommand{\EComplexity}
  {{\txtname{easy}}}

\newcommand{\HComplexity}
  {{\txtname{hard}}}

\newcommand{\CComplexity}
  {{\txtname{complete}}}

\newcommand{\UComplexity}
  {{\txtname{U}}}

\newcommand{\NComplexity}
  {{\txtname{N}}}

\newcommand{\CoComplexity}
  {{\txtname{Co}}}

\newtheorem{definition}{Definition}

\newtheorem{proposition}{Proposition}
\newtheorem{lemma}{Lemma}

\newtheorem{theorem}{Theorem}

% ...

\newcommandx{\MF}
  {\txtname{MF}}
\newcommandx{\QDR}
  {\txtname{QDR}}

\newcommandx{\qdr}
  {\txtname{qdr}}

\newcommandx{\ParGam}
  {\txtname{ParityGames}}

\newcommandx{\PP}[5][1=, 2=, 3=, 4=, 5=]
  {\txtargname{PP#5{\small\argint{$[$}{#1}{$]$}}}[#2][#3]{#4}}
\newcommandx{\PR}[5][1=, 2=, 3=, 4=, 5=]
  {\txtargname{PR#5{\small\argint{$[$}{#1}{$]$}}}[#2][#3]{#4}}
\newcommandx{\BRIM}[5][1=, 2=, 3=, 4=, 5=]
  {\txtargname{SEPM#5{\small\argint{$[$}{#1}{$]$}}}[#2][#3]{#4}}
\newcommandx{\QDPM}[5][1=, 2=, 3=, 4=, 5=]
  {\txtargname{QDPM#5{\small\argint{$[$}{#1}{$]$}}}[#2][#3]{#4}}

\newcommand{\gamname}{\Game}
\newcommand{\GamName}
  {\mthname{\gamname}}

\newcommand{\movrel}{Mv}
\newcommandx{\MovRel}[3][1=, 2=, 3=]
  {\mthrel{\movrel#3}[#1][#2]}

\newcommand{\winset}{Wn}
\newcommandx{\WinSet}[3][1=, 2=, 3=]
  {\mthset{\winset#3}[#1][#2]}

\newcommand{\prtset}{Pr}
\newcommandx{\PrtSet}[3][1=, 2=, 3=]
  {\mthset{\prtset#3}[#1][#2]}
\newcommand{\prtsym}{p}
\newcommandx{\prtSym}[3][1=, 2=, 3=]
  {\mthsym{\prtsym#3}[#1][#2]}
\newcommand{\prtelm}{p}
\newcommandx{\prtElm}[3][1=, 2=, 3=]
  {\mthelm{\prtelm#3}[#1][#2]}

\newcommand{\prtfun}{pr}
\newcommandx{\prtFun}[4][1=, 2=, 3=, 4=]
  {\mthargfun{\prtfun#4}[#1][#2]{#3}}

\newcommand{\depfun}{dep}
\newcommandx{\depFun}[4][1=, 2=, 3=, 4=]
  {\mthargfun{\depfun#4}[#1][#2]{#3}}

\newcommand{\escfun}{esc}
\newcommandx{\escFun}[4][1=, 2=, 3=, 4=]
  {\mthargfun{\escfun#4}[#1][#2]{#3}}

\newcommand{\intfun}{int}
\newcommandx{\intFun}[4][1=, 2=, 3=, 4=]
  {\mthargfun{\intfun#4}[#1][#2]{#3}}

\newcommand{\styfun}{sty}
\newcommandx{\styFun}[4][1=, 2=, 3=, 4=]
  {\mthargfun{\styfun#4}[#1][#2]{#3}}

\newcommand{\srcfun}{src}
\newcommandx{\srcFun}[4][1=, 2=, 3=, 4=]
  {\mthargfun{\srcfun#4}[#1][#2]{#3}}

\newcommand{\regset}{Rg}
\newcommandx{\RegSet}[3][1=, 2=, 3=]
  {\mthset{\regset#3}[#1][#2]}

\newcommand{\qdset}{QD}
\newcommandx{\QDSet}[3][1=, 2=, 3=]
  {\mthset{\qdset#3}[#1][#2]}

\newcommand{\diselm}{\top}
\newcommandx{\disElm}[3][1=, 2=, 3=]
  {\mthelm{\diselm#3}[#1][#2]}

\newcommand{\ordrel}{\prec}
\newcommandx{\ordRel}[3][1=, 2=, 3=]
  {\mthrel{\ordrel#3}[#1][#2]}

\newcommand{\comrel}{\Yright}
\newcommandx{\comRel}[3][1=, 2=, 3=]
  {\mthrel{\comrel#3}[#1][#2]}

\newcommand{\qryfun}{\Re}
\newcommandx{\qryFun}[4][1=, 2=, 3=, 4=]
  {\mthargfun{\qryfun#4}[#1][#2]{#3}}

\newcommand{\sucfun}{\downarrow}
\newcommandx{\sucFun}[4][1=, 2=, 3=, 4=]
  {\mthargfun{\sucfun#4}[#1][#2]{#3}}

\newcommand{\strset}{Str}
\newcommandx{\StrSet}[3][1=, 2=, 3=]
  {\mthset{\strset#3}[#1][#2]}
\newcommand{\strsym}{\sigma}
\newcommandx{\strSym}[4][1=, 2=, 3=, 4=]
  {\mthargfun{\strsym#4}[#1][#2]{#3}}
\newcommand{\strelm}{\sigma}
\newcommandx{\strElm}[4][1=, 2=, 3=, 4=]
  {\mthargfun{\strelm#4}[#1][#2]{#3}}

\newcommand{\playset}{Plays}
\newcommandx{\PlaySet}[4][1=, 2=, 3=, 4=]
  {\mthset{\playset#4}[#1][#2]{#3}}
\newcommand{\playfun}{play}
\newcommandx{\playFun}[4][1=, 2=, 3=, 4=]
  {\mthargfun{\playfun#4}[#1][#2]{#3}}

\newcommand{\pthset}{Pth}
\newcommandx{\PthSet}[3][1=, 2=, 3=]
  {\mthset{\pthset#3}[#1][#2]}
\newcommand{\pthsym}{\pi}
\newcommandx{\pthSym}[3][1=, 2=, 3=]
  {\mthsym{\pthsym#3}[#1][#2]}
\newcommand{\pthelm}{\pi}
\newcommandx{\pthElm}[3][1=, 2=, 3=]
  {\mthelm{\pthelm#3}[#1][#2]}

\newcommand{\hstset}{Hst}
\newcommandx{\HstSet}[3][1=, 2=, 3=]
  {\mthset{\hstset#3}[#1][#2]}

\newcommand{\prfset}{Pf}
\newcommandx{\PrfSet}[3][1=, 2=, 3=]
  {\mthset{\prfset#3}[#1][#2]}
\newcommand{\prfsym}{f}
\newcommandx{\prfSym}[3][1=, 2=, 3=]
  {\mthsym{\prfsym#3}[#1][#2]}
\newcommand{\prfelm}{f}
\newcommandx{\prfElm}[3][1=, 2=, 3=]
  {\mthelm{\prfelm#3}[#1][#2]}

\newcommand{\refrel}{\trianglelefteq}
\newcommandx{\refRel}[3][1=, 2=, 3=]
  {\mthrel{\refrel#3}[#1][#2]}

\newcommand{\plrrel}{\leq}
\newcommandx{\plrRel}[3][1=, 2=, 3=]
  {\mthrel{\plrrel#3}[#1][#2]}

\newcommand{\gaufun}{\eth}
\newcommandx{\gauFun}[4][1=, 2=, 3=, 4=]
  {\mthargfun{\gaufun#4}[#1][#2]{#3}}

\newcommand{\noropr}{\upharpoonright}
\newcommandx{\norOpr}[4][1=, 2=, 3=, 4=]
  {\mthargfun{\noropr#4}[#1][#2]{#3}}

\newcommand{\evlset}{Ev}
\newcommandx{\EvlSet}[3][1=, 2=, 3=]
  {\mthset{\evlset#3}[#1][#2]}
\newcommand{\evlsym}{\eta}
\newcommandx{\evlSym}[3][1=, 2=, 3=]
  {\mthsym{\evlsym#3}[#1][#2]}
\newcommand{\evlelm}{\eta}
\newcommandx{\evlElm}[3][1=, 2=, 3=]
  {\mthelm{\evlelm#3}[#1][#2]}

\newcommand{\qdeset}{Ev}
\newcommandx{\QDESet}[3][1=, 2=, 3=]
  {\mthset{\qdeset#3}[#1][#2]}
\newcommand{\qdesym}{\mu}
\newcommandx{\qdeSym}[3][1=, 2=, 3=]
  {\mthsym{\qdesym#3}[#1][#2]}
\newcommand{\qdeelm}{\mu}
\newcommandx{\qdeElm}[3][1=, 2=, 3=]
  {\mthelm{\qdeelm#3}[#1][#2]}

\newcommand{\bndfun}{\#}
\newcommandx{\bndFun}[4][1=, 2=, 3=, 4=]
  {\mthargfun{\bndfun#4}[#1][#2]{#3}}

\newcommand{\denot}[1]
  {\mthempty{\argint{\lVert}{#1}{\rVert}}}

% ...

\newcommand{\crtset}{Cr}
\newcommandx{\CrtSet}[3][1=, 2=, 3=]
  {\mthset{\crtset#3}[#1][#2]}

\newcommand{\ifpfun}{ifp}
\newcommandx{\ifpFun}[4][1=, 2=, 3=, 4=]
  {\mthargfun{\ifpfun#4}[#1][#2]{#3}}

\newcommand{\lfpfun}{lfp}
\newcommandx{\lfpFun}[4][1=, 2=, 3=, 4=]
  {\mthargfun{\lfpfun#4}[#1][#2]{#3}}

\newcommand{\rwaset}{RWA}
\newcommandx{\RWASet}[3][1=, 2=, 3=]
  {\mthset{\rwaset#3}[#1][#2]}

\newcommand{\nabfun}{\nabla}
\newcommandx{\nabFun}[4][1=, 2=, 3=, 4=]
  {\mthargfun{\nabfun#4}[#1][#2]{#3}}

\newcommand{\delfun}{\Delta}
\newcommandx{\delFun}[4][1=, 2=, 3=, 4=]
  {\mthargfun{\delfun#4}[#1][#2]{#3}}

% ...

\def\forallcmd#1{\ifx#1\forallcmd\else\defcmd#1\expandafter\forallcmd\fi}

\newcommandx{\DefMacroStructure}[5][2=, 3=, 4=, 5=]
  {
  \DefMacroName{#1}[#2]
  \DefMacroSet{#1}[#3][#4][#5]
  }
\newcommandx{\DefMacroName}[2][2=]
  {
  \csdef{#1Name}{\mthname{\argdef{#2}{#1}}}
  }
\newcommandx{\DefMacroSet}[4][2=, 3=, 4=]
  {
  \csdef{#1Set}{\mthset{\argdef{#2}{#1}}}
  \caselower[q]{#1}
  \DefMacroElm{\thestring}[#3]
  \DefMacroSym{\thestring}[#4]
  }
\newcommandx{\DefMacroElm}[2][2=]
  {
  \csdef{#1Elm}{\mthelm{\argdef{#2}{#1}}}
  \def\defcmd##1{\csdef{##1#1Elm}{\mthelm{##1}}}
  \forallcmd abcdefghijklmnopqrstuvwxyz\forallcmd
  \forallcmd ABCDEFGHIJKLMNOPQRSTUVWXYZ\forallcmd
  }
\newcommandx{\DefMacroSym}[2][2=]
  {
  \csdef{#1Sym}{\mthsym{\argdef{#2}{#1}}}
  \def\defcmd##1{\csdef{##1#1Sym}{\mthsym{##1}}}
  \forallcmd abcdefghijklmnopqrstuvwxyz\forallcmd
  \forallcmd ABCDEFGHIJKLMNOPQRSTUVWXYZ\forallcmd
  }
\newcommandx{\DefMacroFun}[2][2=]
  {
  \csdef{#1Fun}{\mthfun{\argdef{#2}{#1}}}
  }
\newcommandx{\DefMacroRel}[2][2=]
  {
  \csdef{#1Rel}{\mthrel{\argdef{#2}{#1}}}
  }

\DefMacroSym{Plr}[\oplus]
\DefMacroSym{Opp}[\boxminus]
\DefMacroSet{Pos}[Ps][v][v]

\DefMacroSet{Wgh}[Wg][w][w]
\DefMacroFun{wgh}[wg]

\DefMacroStructure{Msr}[M][Ms][\eta][\eta]
\DefMacroFun{pay}[pf]
\DefMacroFun{val}[c]

\DefMacroStructure{MF}[F][][\mu][\mu]
\DefMacroStructure{SMF}[F][][\mu][\mu]
\DefMacroStructure{QDR}[Q][][\varrho][\varrho]

\DefMacroSet{Str}[Str][\sigma][\sigma]

\DefMacroFun{npp}

\DefMacroFun{win}

\DefMacroFun{sol}
\DefMacroFun{lift}

\DefMacroFun{nab}[\nabla]

\DefMacroFun{npm}[npm]
\DefMacroFun{nep}

\DefMacroFun{atr}
\DefMacroFun{pre}
\DefMacroFun{adj}

\DefMacroFun{qsi}[Q]
\DefMacroFun{dmn}[\Delta]

\DefMacroFun{prm}

\DefMacroFun{extmin}
\DefMacroFun{prioqueue}

\DefMacroSet{Frf}[Fr][f][f]

\DefMacroFun{bef}
\DefMacroFun{bep}
\DefMacroFun{prg}

\newcommand{\renewcounter}[2]
  {
  \setcounterref{#1}{#2}
  \addtocounter{#1}{-1}
  }

% ...

\usepackage{tikz}
\usetikzlibrary{arrows,shapes,calc,patterns}

\usepackage{float}
\usepackage{wrapfig}

\usepackage{pgf}
\usepackage{pgfplots}
\usepackage{pgfplotstable}

\pgfplotsset{width=7cm, compat=newest}

\pgfdeclarepatternformonly{mydots}
{\pgfqpoint{-1pt}{-1pt}}
{\pgfqpoint{1pt}{1pt}}
{\pgfqpoint{2.1pt}{2.1pt}}%
{
  \pgfpathcircle{\pgfqpoint{0pt}{0pt}}{.5pt}
  \pgfusepath{fill}
}
\tikzstyle{every node} =
  [draw = none, fill = white, thin]
\tikzstyle{every edge} +=
  [black, thick]

\tikzstyle{noall} =
  [draw = none, fill = none]
\tikzstyle{nodraw} =
  [draw = none, fill = white]
\tikzstyle{nofill} =
  [draw = black, fill = none]

\tikzstyle{nonode} =
  [draw = white]
\tikzstyle{cnode} =
  [circle, draw = black]
\tikzstyle{snode} =
  [regular polygon, regular polygon sides = 4, draw = gray!50]
\tikzstyle{lnode} =
  [diamond, draw = gray!75]
\tikzstyle{pnode} =
  [regular polygon, regular polygon sides = 5, draw = gray]

\AfterEndPreamble
  {

  \newcommand{\figmpgexma}
    {
    \begin{center}
    \footnotesize
    \scalebox{0.750}[0.750]
    {
    \begin{tikzpicture}
      [node distance = 5em, bend angle = 22.5, inner sep = 0.20em, minimum size
      = 3em]

      \tikzset{every loop/.style = {max distance = 1.5em}}

      \node [snode]
            (A)
            []
            {$\aSym/\mathbf{k}$};
      \node [cnode]
            (B)
            [right of = A]
            {$\bSym/\mathbf{0}$};
      \node [snode]
            (C)
            [below of = A]
            {$\cSym/\mathbf{0}$};
      \node [cnode]
            (D)
            [right of = C]
            {$\dSym/\mathbf{1}$};

      \path[->]
        (A) edge  [loop left]
                  ()
            edge  []
                  (B)
            edge  []
                  (D)
        (B) edge  [loop right]
                  ()
        (C) edge  []
                  (A)
            edge  [bend left]
                  (D)
        (D) edge  [bend left]
                  (C)
        ;

      \begin{scope}
        [very thick, draw = blue, fill = gray, fill opacity = 0.075]

        \filldraw
          ($(C) + (0, 0.75)$)
            to [out = 180, in = 90]
          ($(C) + (-0.9, 0)$)
            to [out = 270, in = 180]
          ($(C) + (0, -0.75)$)
            to [out = 0, in = 180]
          ($(D) + (0, -0.75)$)
            to [out = 0, in = 270]
          ($(D) + (0.9, 0)$)
            to [out = 90, in = 0]
          ($(D) + (0, 0.75)$)
            to [out = 180, in = 0]
          ($(C) + (0, 0.75)$)
          ;

      \end{scope}

    \end{tikzpicture}
    }
    \end{center}
    }

  \newcommand{\figmpgexmb}
    {
    \begin{center}
    \footnotesize
    \scalebox{0.750}[0.750]
    {
    \begin{tikzpicture}
      [node distance = 5em, bend angle = 22.5, inner sep = 0.20em, minimum size
      = 3em]

      \tikzset{every loop/.style = {max distance = 1.5em}}

      \node [snode]
            (A)
            []
            {$\aSym/\mathbf{3}$};
      \node [cnode]
            (B)
            [right of = A]
            {$\bSym/\mathbf{-1}$};
      \node [snode]
            (C)
            [below of = A]
            {$\cSym/\mathbf{1}$};
      \node [snode]
            (D)
            [right of = C]
            {$\dSym/\mathbf{0}$};
      \node [cnode]
            (E)
            [below of = C]
            {$\eSym/\!\mathbf{-3}$};
      \node [cnode]
            (F)
            [below of = D]
            {$\fSym/\mathbf{0}$};

      \path[->]
        (A) edge  [bend right = 45]
                  (E)
        (B) edge  [very thick, blue, dashed]
                  (A)
            edge  []
                  (C)
        (C) edge  []
                  (A)
            edge  []
                  (F)
        (D) edge  []
                  (B)
            edge  []
                  (C)
        (E) edge  []
                  (C)
        (F) edge  [loop below]
                  ()
            edge  [very thick, blue, dashed]
                  (D)
        ;

       \begin{scope}
         [very thick, draw = blue, fill = gray, fill opacity = 0.075]

         \filldraw
           ($(C) + (0, 0.75)$)
             to [out = 180, in = 90]
           ($(C) + (-0.75, 0)$)
             to [out = 270, in = 180]
           ($(C) + (0, -0.75)$)
             to [out = 0, in = 90]
           ($(F) + (-0.75, 0)$)
             to [out = 270, in = 180]
           ($(F) + (0, -0.90)$)
             to [out = 0, in = 270]
           ($(F) + (0.75, 0)$)
             to [out = 90, in = 270]
           ($(D) + (0.75, 0)$)
             to [out = 90, in = 0]
           ($(D) + (0, 0.75)$)
             to [out = 180, in = 0]
           ($(C) + (0, 0.75)$)
           ;

       \end{scope}

    \end{tikzpicture}
    }
    \end{center}
    }

  \newcommand{\figmpgexmc}
    {
    \begin{center}
    \footnotesize
    \scalebox{0.750}[0.750]
    {
    \begin{tikzpicture}
      [node distance = 5em, bend angle = 22.5, inner sep = 0.20em, minimum size
      = 3em]

      \tikzset{every loop/.style = {max distance = 1.5em}}

      \node [snode]
            (A)
            []
            {$\aSym/\mathbf{3}$};
      \node [cnode]
            (B)
            [right of = A]
            {$\bSym/\mathbf{-3}$};
      \node [snode]
            (C)
            [right of = B]
            {$\cSym/\mathbf{2}$};
      \node [snode]
            (D)
            [below of = A]
            {$\dSym/\mathbf{1}$};
      \node [cnode]
            (E)
            [below of = C]
            {$\eSym/\!\mathbf{0}$};

      \path[->]
        (A) edge  [bend right]
                  (B)
        (B) edge  [bend right]
                  (A)
        (C) edge  []
                  (B)
            edge  [bend right]
                  (E)
        (D) edge  []
                  (A)
            edge  [bend right]
                  (E)
        (E) edge  [bend right]
                  (C)
            edge  [bend right]
                  (D)
          ;

    \end{tikzpicture}
    }
    \end{center}
    }

  \newcommand{\figsima}
    {
    \begin{center}
    \footnotesize
    \scalebox{0.600}[0.680]
    {
    \begin{tikzpicture}
      [node distance = 5em, bend angle = 45, inner sep = 0.3em, minimum size =
      3.2em]

      \tikzset{every loop/.style = {max distance = 1.5em}}

      \node [snode]
            (A)
            []
            {$\aSym/\mathbf{k}$};
      \node [snode]
            (B)
            [right of = A]
            {$\!\!\bSym/\mathbf{-1}\!\!$};
      \node [cnode]
            (C)
            [right of = B]
            {$\cSym/\mathbf{0}$};
      \node [snode]
            (D)
            [right of = C]
            {$\!\!\dSym/\mathbf{-1}\!\!$};
      \node [snode]
            (E)
            [below of = A]
            {$\eSym/\mathbf{0}$};
      \node [snode]
            (F)
            [below of = C, xshift = -3em]
            {$\fSym/\mathbf{2}$};
      \node [snode]
            (G)
            [below of = C, xshift = +3em]
            {$\gSym/\mathbf{2}$};
      \node [minimum size = 2.5em]
            (L)
            [below right of = B, xshift = -1em, yshift = -5.25em]
            {\large $\mathbf{(0): \text{ The \MPG}}$};

      \path[->]
        (A) edge  []
                  (E)
        (B) edge  []
                  (A)
            edge  []
                  (C)
        (C) edge  []
                  (D)
            edge  []
                  (F)
        (D) edge  []
                  (G)
        (E) edge  [loop right]
                  ()
        (F) edge  []
                  (B)
            edge  []
                  (G)
        (G) edge  []
                  (C)
        ;

    \end{tikzpicture}
    }
    \end{center}
    }

  \newcommand{\figsimb}
    {
    \begin{center}
    \footnotesize
    \scalebox{0.600}[0.680]
    {
    \begin{tikzpicture}
      [node distance = 5em, bend angle = 45, inner sep = 0.3em, minimum size =
      3.2em]

      \tikzset{every loop/.style = {max distance = 1.5em}}

      \node [snode]
            (A)
            []
            {$\aSym/\mathbf{k}$};
      \node [snode]
            (B)
            [right of = A]
            {$\bSym/\mathbf{0}$};
      \node [cnode]
            (C)
            [right of = B]
            {$\cSym/\mathbf{0}$};
      \node [snode]
            (D)
            [right of = C]
            {$\dSym/\mathbf{0}$};
      \node [snode]
            (E)
            [below of = A]
            {$\eSym/\mathbf{0}$};
      \node [snode]
            (F)
            [below of = C, xshift = -3em]
            {$\fSym/\mathbf{2}$};
      \node [snode]
            (G)
            [below of = C, xshift = +3em]
            {$\gSym/\mathbf{2}$};
      \node [minimum size = 2.5em]
            (L)
            [below right of = B, xshift = -1em, yshift = -5.25em]
            {\large $\mathbf{(1):\qdrElm[1] = \prgFun[0](\qdrElm[0]) =
            \prgFun[+](\prgFun[0](\qdrElm[0]))}$};

      \path[->]
        (A) edge  []
                  (E)
        (B) edge  []
                  (A)
            edge  []
                  (C)
        (C) edge  []
                  (D)
            edge  []
                  (F)
        (D) edge  []
                  (G)
        (E) edge  [loop right]
                  ()
        (F) edge  [very thick, red, dashed]
                  (B)
            edge  []
                  (G)
        (G) edge  [very thick, red, dashed]
                  (C)
        ;

    \end{tikzpicture}
    }
  \end{center}
  }

  \newcommand{\figsimc}
    {
    \begin{center}
    \footnotesize
    \scalebox{0.600}[0.680]
    {
    \begin{tikzpicture}
      [node distance = 5em, bend angle = 45, inner sep = 0.3em, minimum size =
      3.2em]

      \tikzset{every loop/.style = {max distance = 1.5em}}

      \node [snode]
            (A)
            []
            {$\aSym/\mathbf{k}$};
      \node [snode]
            (B)
            [right of = A]
            {$\bSym/\mathbf{0}$};
      \node [cnode]
            (C)
            [right of = B]
            {$\cSym/\mathbf{2}$};
      \node [snode]
            (D)
            [right of = C]
            {$\dSym/\mathbf{1}$};
      \node [snode]
            (E)
            [below of = A]
            {$\eSym/\mathbf{0}$};
      \node [snode]
            (F)
            [below of = C, xshift = -3em]
            {$\fSym/\mathbf{2}$};
      \node [snode]
            (G)
            [below of = C, xshift = +3em]
            {$\gSym/\mathbf{2}$};
      \node [minimum size = 2.5em]
            (L)
            [below right of = B, xshift = -1em, yshift = -5.25em]
            {\large $\mathbf{(2):\qdrElm[2] = \prgFun[0](\qdrElm[1])}$};

      \path[->]
        (A) edge  []
                  (E)
        (B) edge  []
                  (A)
            edge  []
                  (C)
        (C) edge  []
                  (D)
            edge  [very thick, blue, dashed]
                  (F)
        (D) edge  [very thick, red, dashed]
                  (G)
        (E) edge  [loop right]
                  ()
        (F) edge  [very thick, red, dashed]
                  (B)
            edge  []
                  (G)
        (G) edge  [very thick, red]
                  (C)
        ;

      \begin{scope}
        [very thick, draw = blue, fill = gray, fill opacity = 0.075]

        \filldraw
          ($(G) + (-0.50, -0.80)$)
            to [out = 0, in = 180]
          ($(G) + (1.15, -0.80)$)
            to [out = 0, in = 270]
          ($(D) + (0.80, -2.10)$)
            to [out = 90, in = 270]
          ($(D) + (0.80, 0.50)$)
            to [out = 90, in = 0]
          ($(D) + (0.50, 0.80)$)
            to [out = 180, in = 0]
          ($(D) + (-0.50, 0.80)$)
            to [out = 180, in = 90]
          ($(D) + (-0.80, 0.50)$)
            to [out = 270, in = 0]
          ($(G) + (-0.50, 0.80)$)
            to [out = 180, in = 90]
          ($(G) + (-0.80, 0.50)$)
            to [out = 270, in = 90]
          ($(G) + (-0.80, -0.50)$)
            to [out = 270, in = 180]
          ($(G) + (-0.50, -0.80)$)
          ;

      \end{scope}

    \end{tikzpicture}
    }
  \end{center}
  }

  \newcommand{\figsimd}
    {
    \begin{center}
    \footnotesize
    \scalebox{0.600}[0.680]
    {
    \begin{tikzpicture}
      [node distance = 5em, bend angle = 45, inner sep = 0.3em, minimum size =
      3.2em]

      \tikzset{every loop/.style = {max distance = 1.5em}}

      \node [snode]
            (A)
            []
            {$\aSym/\mathbf{k}$};
      \node [snode]
            (B)
            [right of = A]
            {$\bSym/\mathbf{0}$};
      \node [cnode]
            (C)
            [right of = B]
            {$\!\!\cSym/\mathbf{2}\!\!$};
      \node [snode]
            (D)
            [right of = C]
            {$\dSym/\mathbf{3}$};
      \node [snode]
            (E)
            [below of = A]
            {$\eSym/\mathbf{0}$};
      \node [snode]
            (F)
            [below of = C, xshift = -3em]
            {$\fSym/\mathbf{2}$};
      \node [snode]
            (G)
            [below of = C, xshift = +3em]
            {$\gSym/\mathbf{4}$};
      \node [minimum size = 2.5em]
            (L)
            [below right of = B, xshift = -1em, yshift = -5.25em]
            {\large $\mathbf{(3):\qdrElm[3] = \prgFun[+](\qdrElm[2])}$};

      \path[->]
        (A) edge  []
                  (E)
        (B) edge  []
                  (A)
            edge  []
                  (C)
        (C) edge  [very thick, blue]
                  (D)
            edge  [very thick, blue, dashed]
                  (F)
        (D) edge  [very thick, red, dashed]
                  (G)
        (E) edge  [loop right]
                  ()
        (F) edge  [very thick, red, dashed]
                  (B)
            edge  []
                  (G)
        (G) edge  [very thick, red, dashed]
                  (C)
        ;

    \end{tikzpicture}
    }
  \end{center}
  }

  \newcommand{\figsime}
    {
    \begin{center}
    \footnotesize
    \scalebox{0.600}[0.680]
    {
    \begin{tikzpicture}
      [node distance = 5em, bend angle = 45, inner sep = 0.3em, minimum size =
      3.2em]

      \tikzset{every loop/.style = {max distance = 1.5em}}

      \node [snode]
            (A)
            []
            {$\aSym/\mathbf{k}$};
      \node [snode]
            (B)
            [right of = A]
            {$\bSym/\mathbf{1}$};
      \node [cnode]
            (C)
            [right of = B]
            {$\!\!\cSym/\mathbf{2}\!\!$};
      \node [snode]
            (D)
            [right of = C]
            {$\dSym/\mathbf{3}$};
      \node [snode]
            (E)
            [below of = A]
            {$\eSym/\mathbf{0}$};
      \node [snode]
            (F)
            [below of = C, xshift = -3em]
            {$\fSym/\mathbf{2}$};
      \node [snode]
            (G)
            [below of = C, xshift = +3em]
            {$\gSym/\mathbf{4}$};
      \node [minimum size = 2.5em]
            (L)
            [below right of = B, xshift = -1em, yshift = -5.25em]
            {\large $\mathbf{(4):\qdrElm[4] = \prgFun[0](\qdrElm[3])}$};

      \path[->]
        (A) edge  []
                  (E)
        (B) edge  []
                  (A)
            edge  [very thick, red, dashed]
                  (C)
        (C) edge  [very thick, blue]
                  (D)
            edge  [very thick, blue, dashed]
                  (F)
        (D) edge  [very thick, red, dashed]
                  (G)
        (E) edge  [loop right]
                  ()
        (F) edge  [very thick, red]
                  (B)
            edge  []
                  (G)
        (G) edge  [very thick, red, dashed]
                  (C)
        ;

      \begin{scope}
        [very thick, draw = blue, fill = gray, fill opacity = 0.075]

        \filldraw
          ($(B) + (-0.50, 0.80)$)
            to [out = 180, in = 90]
          ($(B) + (-0.80, 0.50)$)
            to [out = 270, in = 90]
          ($(B) + (-0.80, -0.50)$)
            to [out = 270, in = 180]
          ($(B) + (-0.50, -0.80)$)
            to [out = 0, in = 90]
          ($(F) + (-0.80, -0.50)$)
            to [out = 270, in = 180]
          ($(F) + (-0.50, -0.80)$)
            to [out = 0, in = 180]
          ($(G) + (0.50, -0.80)$)
            to [out = 0, in = 270]
          ($(G) + (0.80, -0.50)$)
            to [out = 90, in = 180]
          ($(D) + (0.50, -0.80)$)
            to [out = 0, in = 270]
          ($(D) + (0.80, -0.50)$)
            to [out = 90, in = 270]
          ($(D) + (0.80, 0.50)$)
            to [out = 90, in = 0]
          ($(D) + (0, 0.80)$)
            to [out = 180, in = 0]
          ($(B) + (-0.50, 0.80)$)
          ;

      \end{scope}

    \end{tikzpicture}
    }
  \end{center}
  }

  \newcommand{\figsimf}
    {
    \begin{center}
    \footnotesize
    \scalebox{0.600}[0.680]
    {
    \begin{tikzpicture}
      [node distance = 5em, bend angle = 45, inner sep = 0.3em, minimum size =
      3.2em]

      \tikzset{every loop/.style = {max distance = 1.5em}}

      \node [snode]
            (A)
            []
            {$\aSym/\mathbf{k}$};
      \node [snode]
            (B)
            [right of = A]
            {$\!\!\!\bSym/\mathbf{k}\!\!-\!\!\mathbf{1}\!\!\!\!$};
      \node [cnode]
            (C)
            [right of = B]
            {$\cSym/\infty$};
      \node [snode]
            (D)
            [right of = C]
            {$\!\!\dSym/\infty\!$};
      \node [snode]
            (E)
            [below of = A]
            {$\eSym/\mathbf{0}$};
      \node [snode]
            (F)
            [below of = C, xshift = -3em]
            {$\!\!\!\fSym/\mathbf{k}\!\!+\!\!\mathbf{1}\!\!\!\!$};
      \node [snode]
            (G)
            [below of = C, xshift = +3em]
            {$\!\gSym/\infty\!\!$};
      \node [minimum size = 2.5em]
            (L)
            [below right of = B, xshift = -1em, yshift = -5.25em]
            {\large $\mathbf{(5):\qdrElm[5] = \prgFun[+](\qdrElm[4])}$};

      \path[->]
        (A) edge  []
                  (E)
        (B) edge  [very thick, red, dashed]
                  (A)
            edge  []
                  (C)
        (C) edge  [very thick, blue, dashed]
                  (D)
            edge  []
                  (F)
        (D) edge  [very thick, red, dashed]
                  (G)
        (E) edge  [loop right]
                  ()
        (F) edge  [very thick, red, dashed]
                  (B)
            edge  []
                  (G)
        (G) edge  [very thick, red, dashed]
                  (C)
        ;

    \end{tikzpicture}
    }
  \end{center}
  }

  \newcommand{\figsctplt}
    {
    \begin{center}
      \begin{tikzpicture}[every node/.style={scale=0.75},xscale=1,yscale=0.8]
        \hspace{-0.1cm}
        \begin{axis}
          [
          scatter/classes={%
            10={gray},%
            20={black},%
            40={blue},%
            80={red}},%
          xshift = 3.25in,
          width=8cm, height=7.5cm,
          xmode=log, ymode=log,
          xmin=11e-3, xmax=0.6,
          ymin=11e-3, ymax=1e4,
          name=axis2,
          ylabel near ticks, xlabel near ticks,
          every axis y label/.style={at={(axis description cs:-0.12,0.5)},
          rotate=90},
          every axis x label/.style={at={(axis description cs:0.5,-0.1)},
          rotate=0},
          ylabel={\BRIM}, xlabel={\QDPM},
          enlargelimits=false
          ]
          \addplot[scatter, only marks, scatter src=explicit symbolic, mark
            size=.9, mark=+] table[y=BRIM, x=PMP,
            meta=META]{results-5k-15k-10.txt};
          \addplot[scatter, only marks, scatter src=explicit symbolic, mark
            size=.9, mark=o] table[y=BRIM, x=PMP,
            meta=META]{results-5k-15k-20.txt};
          \addplot[scatter, only marks, scatter src=explicit symbolic, mark
            size=.9, mark=star] table[y=BRIM, x=PMP,
            meta=META]{results-5k-15k-40.txt};
          \addplot[scatter, only marks, scatter src=explicit symbolic, mark
            size=.9, mark=diamond] table[y=BRIM, x=PMP,
            meta=META]{results-5k-15k-80.txt};
          \addplot[color=blue!70!black, line width=.7, domain=0.011:0.6,
            mark=none] {\x} node[below,pos=0.96]{\scriptsize$\times1$};
          \addplot[color=red!50!black, line width=.7, domain=0.0011:0.6,
            mark=none, dashed] {10*\x}
            node[below,pos=0.97]{\scriptsize$\times10$};
          \addplot[color=red!50!black, line width=.7, domain=0.00011:0.6,
            mark=none, dashed] {100*\x}
            node[below,pos=0.973]{\scriptsize$\times10^2$};
          \addplot[color=red!50!black, line width=.7, domain=0.000011:0.6,
            mark=none, dashed]
            {1000*\x}node[below,pos=0.98]{\scriptsize$\times10^3$};
          \addplot[color=red!50!black, line width=.7, domain=0.0000011:0.6,
            mark=none, dashed]
            {10000*\x}node[below,pos=0.982]{\scriptsize$\times10^4$};
          \addplot[color=red!50!black, line width=.7, domain=0.00000011:0.1,
            mark=none, dashed]
            {100000*\x}node[below,pos=0.982]{\scriptsize$\times10^5$};
        \end{axis}
      \end{tikzpicture}
    \end{center}
    }

      \newcommand{\expmoves}
    {
      \begin{center}
        \begin{tikzpicture}[every node/.style={scale=0.75}, xscale=1.25,yscale=0.95]
          \begin{axis}[
            ymode=log,
            xmin=1, ymin=1,
            ymax=3e5,
            ymajorgrids=true,
            xmajorgrids=true,
            xtick=data,
            stack plots=y,
            area style,
            extra y ticks={20000, 280000},
            xtick={5, 10, 15, 20, 30, 40, 50, 60, 70, 80, 90,100},
            enlarge x limits=false,
            enlarge y limits=upper,
            legend style={nodes={scale=0.5, transform shape}},
            legend pos=north west,
            yticklabel style = {font=\tiny,xshift=0.5ex},
            xticklabel style = {font=\tiny,yshift=0.5ex}
            ]
        \addplot[fill=blue] table [x=GAME, y=PMP] {two-moves-results.txt} \closedcycle;
            \addlegendentry{\QDPM}
        \addplot[pattern=mydots, pattern color=red, draw=orange!30!black] table [x=GAME, y=BRIM] {two-moves-results.txt} \closedcycle;
            \addlegendentry{\BRIM}
      \end{axis}
    \end{tikzpicture}
        \begin{tikzpicture}[every node/.style={scale=0.75},  xscale=1.25,yscale=0.95]
          \begin{axis}[
            ymode=log,
            xmin=1, ymin=1,
            xtick=data,
            stack plots=y,
            area style,
            ymajorgrids=true,
            xmajorgrids=true,
            ytick={1e1,1e2,3.4e2,1e3, 1e4,230000},
            xtick={5, 10, 15, 20, 30, 40, 50, 60, 70, 80, 90,100},
            y tick label style={/pgf/number format/sci},
            enlarge x limits=false,
            enlarge y limits=upper,
            yticklabel style = {font=\tiny,xshift=0.5ex},
            xticklabel style = {font=\tiny,yshift=0.5ex},
            legend style={nodes={scale=0.5, transform shape}},
            legend pos=north west
            ]
            \addplot[fill=blue] table [x=GAME, y=PMP] {results.txt} \closedcycle;
            \addlegendentry{\QDPM}
        \addplot[pattern=mydots, pattern color=red, draw=orange!30!black] table [x=GAME, y=BRIM] {results.txt} \closedcycle;
            \addlegendentry{\BRIM}
      \end{axis}
        \end{tikzpicture}
      \end{center}
    }

  }

\usepackage{multirow}

% ...

\AfterEndPreamble
  {

  \newcommand{\tabsimexm}
    {
    \setlength{\tabcolsep}{5.50pt}
    \begin{tabular}{c|c|c}
      \begin{minipage}[t]{0.3\textwidth}
        \vspace{-.30em}
        \figsima
        \vspace{-0.5em}
      \end{minipage} &
      \begin{minipage}[t]{0.3\textwidth}
        \vspace{-.30em}
        \figsimb
        \vspace{-0.5em}
      \end{minipage} &
      \begin{minipage}[t]{0.3\textwidth}
        \vspace{-0.30em}
        \figsimc
        \vspace{-0.5em}
      \end{minipage}
      \\\hline
      \begin{minipage}[t]{0.3\textwidth}
        \vspace{-0.30em}
        \figsimd
        \vspace{-1em}
      \end{minipage} &
      \begin{minipage}[t]{0.3\textwidth}
        \vspace{-0.30em}
        \figsime
        \vspace{0.1pt}
        \vspace{-1em}
      \end{minipage} &
      \begin{minipage}[t]{0.3\textwidth}
        \vspace{-0.30em}
        \figsimf
        \vspace{-1em}
      \end{minipage}
    \end{tabular}
    }

  }

\usepackage[ruled,vlined]{algorithm2e}

\DontPrintSemicolon

\SetKw{Signature}{signature}
\SetKwFor{Function}{function}{}{}
\SetKwFor{Procedure}{procedure}{}{}

\SetKwFor{Let}{let}{in}{}

\AfterEndPreamble
  {

  \newcommand{\algprg}
    {
    \begin{algorithm}[H]
      \caption{\label{alg:prg} Progress Operator}
      \Signature{$\prgFun[+] \colon \QDRSet \pto \QDRSet$} \;
      \Function{$\prgFun[+](\qdrElm)$}
        {
        \nl $\QSet \gets \dmnFun(\qdrElm)$ \;
        \nl \While{$\escFun(\qdrElm, \QSet) \neq \emptyset$}
          {
          \nl $\ESet \gets \bepFun(\qdrElm, \QSet)$ \;
          \nl $\qdrElm \gets \liftFun(\qdrElm, \ESet, \dual{\QSet})$ \;
          \nl $\QSet \gets \QSet \setminus \ESet$ \;
          }
        \nl $\qdrElm \gets \winFun(\qdrElm, \QSet)$ \;
        \nl \Return $\qdrElm$ \;
        }
    \end{algorithm}
    }

  \newcommand{\algsol}
    {
    \SetInd{0.5em}{0.5em}
    \begin{algorithm}[H]
      \caption{\label{alg:sol} \MPG Solver}
      \Signature{$\solFun \colon \MPG \to \QDRSet$} \;
      \Procedure{$\solFun(\GamName)$}
        {
        \nl $\qdrElm \gets (\{ \posElm \in \PosSet \mapsto 0 \}, \emptyfun)$ \;
        \nl $\cFun \gets \{ \posElm \in \PosSet[\OppSym] \mapsto \card{\set{
            \uposElm \in \MovRel(\posElm) }{ \mfElm[{\qdrElm}](\posElm) \geq
            \mfElm[{\qdrElm}](\uposElm) + \posElm}} \}\!\!\!\!\!\!\!\!\!\!\!\!$
            \;
        \nl $(\NSet[0], \NSet[+]) \gets (\set{ \posElm \in \PosSet }{
            \wghFun(\posElm) > 0 }, \emptyset)$ \;
        \nl \While{$\NSet[0] \neq \emptyset \lor \NSet[+] \neq \emptyset$}
          {
          \nl $(\NSet[0], \ASet) \gets \prgFun[0](\NSet[0])$ \;
          \nl $\NSet[+] \gets \NSet[+] \cup \ASet$ \;
          \nl $(\ASet, \NSet[+]) \gets \prgFun[+](\NSet[+])$ \;
          \nl $\NSet[0] \gets \NSet[0] \cup \ASet$ \;
          }
        \nl \Return $\qdrElm$ \;
        }
    \end{algorithm}
    }

  \newcommand{\algprgzer}
    {
    \SetInd{0.5em}{0.5em}
    \begin{algorithm}[H]
      \caption{\label{alg:prgzer} Efficient Progress Zero Operator}
      \Signature{$\prgFun[0] \colon \pow{\PosSet} \to \pow{\PosSet} \times
        \pow{\PosSet}$} \;
      \Procedure{$\prgFun[0](\NSet)$}
        {
        \nl $\ZSet \gets \emptyset$ \;
        \nl $\der{\mfElm} \gets \{ \posElm \in \NSet \mapsto
            \mfElm[{\qdrElm}](\posElm) \}$ \;
        \nl $\qdrElm \gets \sup \{ \qdrElm, \liftFun(\qdrElm, \NSet, \PosSet)
            \}$ \;
        \nl $\cFun \gets {\cFun}[ \posElm \in \NSet \cap \PosSet[\OppSym]
            \mapsto \card{\set{ \uposElm \in \MovRel(\posElm) }{
            \mfElm[{\qdrElm}](\posElm) \geq \mfElm[{\qdrElm}](\uposElm) +
            \posElm}} ]$ \;
        \nl \ForEach{$(\posElm, \uposElm) \in \MovRel; \: \uposElm \in \NSet; \:
            \mfElm[{\qdrElm}](\posElm) < \mfElm[{\qdrElm}](\uposElm) + \posElm$}
          {
          \nl \eIf{$\posElm \in \PosSet[\PlrSym]$}
            {
            \nl $\ZSet \gets \ZSet \cup \posElm$ \;
            }
            {
            \nl \lIf{$\posElm \not\in \NSet \land \mfElm[{\qdrElm}](\posElm)
                \geq \der{\mfElm}(\uposElm) + \posElm$} {$\cFun(\posElm) \gets
                \cFun(\posElm) - 1$}
            \nl \lIf{$\cFun(\posElm) = 0$} {$\ZSet \gets \ZSet \cup \posElm$}
            }
          }
        \nl \Return $(\ZSet \cap \mfElm[{\qdrElm}][-1](0), \ZSet \setminus
            \mfElm[{\qdrElm}][-1](0))$ \;
        }
    \end{algorithm}
    }

  \newcommand{\algdmn}
    {
    \SetInd{0.5em}{0.5em}
    \begin{algorithm}[H]
      \caption{\label{alg:dmn} Efficient Quasi Dominion Operator}
      \Signature{$\dmnFun \colon \pow{\PosSet} \to \pow{\PosSet}$} \;
      \Procedure{$\dmnFun(\NSet)$}
        {
        \nl $\QSet \gets \emptyset$ \;
        \nl $\dFun \gets \emptyfun$ \;
        \nl \While{$\NSet \not\subseteq \QSet$}
          {
          \nl $\QSet \gets \QSet \cup \NSet$ \;
          \nl $\NSet \gets \preFun(\NSet)$ \;
          }
        \nl \Return $\QSet$ \;
        }
      \Signature{$\preFun \colon \pow{\PosSet} \to \pow{\PosSet}$} \;
      \Procedure{$\preFun(\NSet)$}
        {
        \nl $\ZSet \gets \emptyset$ \;
        \nl $\dFun \gets {\dFun}[ \posElm \in (\MovRel[][-1](\NSet) \cap
            \PosSet[\OppSym]) \setminus \dom{\dFun} \mapsto \cFun(\posElm) ]$ \;
        \nl \ForEach{$(\posElm, \uposElm) \in \MovRel; \: \uposElm \in \NSet$}
          {
          \nl \eIf{$\posElm \in \PosSet[\PlrSym]$}
            {
            \nl \lIf{$\strElm[\qdrElm](\posElm) = \uposElm$} {$\ZSet \gets \ZSet
                \cup \posElm$}
            }
            {
            \nl \lIf{$\mfElm[{\qdrElm}](\posElm) \geq
                \mfElm[{\qdrElm}](\uposElm) + \posElm$} {$\dFun(\posElm) \gets
                \dFun(\posElm) - 1$}
            \nl \lIf{$\dFun(\posElm) = 0$} {$\ZSet \gets \ZSet \cup \posElm$}
            }
          }
        \nl \Return $\ZSet$ \;
        }
    \end{algorithm}
    }

  \newcommand{\algprgpls}
    {
    \SetInd{0.5em}{0.5em}
    \begin{algorithm}[H]
      \caption{\label{alg:prgpls} Efficient Progress Plus Operator}
      \Signature{$\prgFun[+] \colon \pow{\PosSet} \to \pow{\PosSet} \times
        \pow{\PosSet}$} \;
      \Procedure{$\prgFun[+](\NSet)$}
        {
        \nl $\ZSet \gets \emptyset$ \;
        \nl $\QSet \gets \dmnFun(\qdrElm, \cFun, \NSet)$ \;
        \nl $\der{\mfElm} \gets \{ \posElm \in \QSet \mapsto
            \mfElm[{\qdrElm}](\posElm) \}$ \;
        \nl $\TSet \gets \{ (\posElm, \befFun(\mfElm[\qdrElm], \QSet, \posElm))
            \in \escFun(\qdrElm, \QSet) \times \SetN \}$ \;
        \nl $\gFun \gets \{ \posElm \in \QSet \cap \PosSet[\PlrSym] \mapsto
            \card{\set{ \uposElm \in \MovRel(\posElm) \cap \QSet }{
            \strElm[\qdrElm](\posElm) = \uposElm \lor \mfElm[{\qdrElm}](\posElm)
            < \mfElm[{\qdrElm}](\uposElm) + \posElm}} \}$ \;
        \nl \While{$\TSet \neq \emptyset$}
          {
          \nl $(\ESet, \TSet) \gets \extminFun(\TSet)$ \;
          \nl $\qdrElm \gets \liftFun(\qdrElm, \ESet, \dual{\QSet})$ \;
          \nl $\QSet \gets \QSet \setminus \ESet$ \;
          \nl $\cFun \gets {\cFun}[ \posElm \in \ESet \cap \PosSet[\OppSym]
              \mapsto \card{\set{ \uposElm \in \MovRel(\posElm) }{
              \mfElm[{\qdrElm}](\posElm) \geq \mfElm[{\qdrElm}](\uposElm) +
              \posElm}} ]$ \;
          \nl \ForEach{$(\posElm, \uposElm) \in \MovRel; \: \uposElm \in \ESet$}
            {
            \nl \uIf{$\posElm \in \QSet$}
              {
              \nl \eIf{$\posElm \in \PosSet[\OppSym]$}
                {
                \nl {$\TSet \gets \TSet \cup (\posElm,
                    \mfElm[{\qdrElm}](\uposElm) + \posElm -
                    \mfElm[{\qdrElm}](\posElm))$} \;
                }
                {
                \nl \lIf{$\strElm[\qdrElm](\posElm) = \uposElm \lor
                    \mfElm[{\qdrElm}](\posElm) < \der{\mfElm}(\uposElm) +
                    \posElm$} {$\gFun(\posElm) \gets \gFun(\posElm) - 1$}
                \nl \lIf{$\gFun(\posElm) = 0$} {$\TSet \gets \TSet \cup
                    (\posElm, \befFun(\mfElm[\qdrElm], \QSet, \posElm))$}
                }
              }
            \nl \ElseIf{$\mfElm[{\qdrElm}](\posElm) <
                \mfElm[{\qdrElm}](\uposElm) + \posElm$}
              {
              \nl \uIf{$\posElm \in \PosSet[\PlrSym]$}
                {
                \nl $\ZSet \gets \ZSet \cup \posElm$ \;
                }
              \nl \ElseIf{$\mfElm[{\qdrElm}](\posElm) = 0$}
                {
                \nl \lIf{$\der{\mfElm}(\uposElm) + \posElm = 0$}
                    {$\cFun(\posElm) \gets \cFun(\posElm) - 1$}
                \nl \lIf{$\cFun(\posElm) = 0$} {$\ZSet \gets \ZSet \cup
                    \posElm$}
                }
              }
            }
          }
        \nl $\qdrElm \gets \winFun(\qdrElm, \QSet)$ \;
        \nl \ForEach{$(\posElm, \uposElm) \in \MovRel; \: \uposElm \in \QSet$}
          {
          \nl \uIf{$\posElm \in \PosSet[\PlrSym]$}
            {
            \nl $\ZSet \gets \ZSet \cup \posElm$ \;
            }
          \nl \ElseIf{$\mfElm[{\qdrElm}](\posElm) = 0$}
            {
            \nl \lIf{$\der{\mfElm}(\uposElm) + \posElm = 0$} {$\cFun(\posElm)
                \gets \cFun(\posElm) - 1$}
            \nl \lIf{$\cFun(\posElm) = 0$} {$\ZSet \gets \ZSet \cup
                \posElm$}
            }
          }
        \nl \Return $(\ZSet \cap \mfElm[{\qdrElm}][-1](0), \ZSet \setminus
            \mfElm[{\qdrElm}][-1](0))$ \;
        }
    \end{algorithm}
    }

  }

\hypersetup
  {
  pdftitle  = {Solving Mean-Payoff Games via Quasi Dominions},
  pdfauthor = {M. Benerecetti, D. Dell'Erba, F. Mogavero}
  }

\begin{document}

  \title{Solving Mean-Payoff Games via Quasi Dominions}

  \author
    {
    \IEEEauthorblockN
      {Massimo Benerecetti \& Daniele Dell'Erba \& Fabio Mogavero \\
      \{massimo.benerecetti, daniele.dellerba, fabio.mogavero\}@unina.it}
    \IEEEauthorblockA
      {Universit\`a degli Studi di Napoli Federico II, Napoli, Italy}
    }

  \maketitle



\begin{abstract}

  We propose a novel algorithm for the solution of \emph{mean-payoff games} that
  merges together two seemingly unrelated concepts introduced in the context of
  parity games, \emph{small progress measures} and \emph{quasi dominions}.
  We show that the integration of the two notions can be highly beneficial and
  significantly speeds up convergence to the problem solution.
  Experiments show that the resulting algorithm performs orders of magnitude
  better than the asymptotically-best solution algorithm currently known,
  without sacrificing on the worst-case complexity.

\end{abstract}




\begin{section}{Introduction}

  In this article we consider the problem of solving \emph{mean-payoff games},
  namely infinite-duration perfect-information two-player games played on
  weighted directed graphs, each of whose vertexes is controlled by one of the
  two players.
  The game starts at an arbitrary vertex and, during its evolution, each player
  can take moves at the vertexes it controls, by choosing one of the outgoing
  edges.
  The moves selected by the two players induce an infinite sequence of vertices,
  called play.
  The payoff of any prefix of a play is the sum of the weights of its edges.
  A play is winning if it satisfies the game objective, called
  \emph{mean-payoff objective}, which requires that the limit of the
  \emph{mean payoff}, taken over the prefixes lengths, never falls below a given
  \emph{threshold} $\nu$.
  \\\indent
  Mean-payoff games have been first introduced and studied by Ehrenfeucht and
  Mycielski in~\cite{EM79}, who showed that positional strategies suffice to
  obtain the optimal value.
  A slightly generalized version was also considered by Gurvich~\etal
  in~\cite{GKK88}.
  Positional determinacy entails that the decision problem for these games lies
  in $\NPTime \cap \CoNPTime$~\cite{ZP96}, and it was later shown to belong to
  $\UPTime \cap \CoUPTime$~\cite{Jur98}, being $\UPTime$ the class of
  unambiguous non-deterministic polynomial time.
  This result gives the problem a rather peculiar complexity status, shared by
  very few other problems, such as integer
  factorization~\cite{FK92},~\cite{AKS04} and parity games~\cite{Jur98}.
  Despite various attempts~\cite{GKK88,ZP96,Pis99,DG06,BV07}, no polynomial-time
  algorithm for the mean-payoff game problems is known so far.
  \\\indent
  A different formulation of the game objective allows to define another class
  of quantitative games, known as \emph{energy games}.
  The \emph{energy objective} requires that, given an initial value $\cElm$,
  called \emph{credit}, the sum of $\cElm$ and the \emph{payoff} of every prefix
  of the play never falls below $0$.
  These games, however, are tightly connected to mean-payoff games, as the two
  type of games have been proved to be log-space equivalent~\cite{BFLMS08}.
  They are also related to other more complex forms of quantitative games.
  In particular, unambiguous polynomial-time reductions~\cite{Jur98} exist from
  these games to \emph{discounted payoff}~\cite{ZP96} and \emph{simple
  stochastic games}~\cite{Con92}.
  \\\indent
  Recently, a fair amount of work in formal verification has been directed to
  consider, besides correctness properties of computational systems, also
  quantitative specifications, in order to express performance measures and
  resource requirements, such as quality of service, bandwidth and power
  consumption and, more generally, bounded resources.
  Mean-payoff and energy games also have important practical applications in
  system verification and synthesis.
  In~\cite{BCHJ09} the authors show how quantitative aspects, interpreted as
  penalties and rewards associated to the system choices, allow for expressing
  optimality requirements encoded as mean-payoff objectives for the automatic
  synthesis of systems that also satisfy parity objectives.
  With similar application contexts in mind,~\cite{BCHK11} and~\cite{BBFR13}
  further contribute to that effort, by providing complexity results and
  practical solutions for the verification and automatic synthesis of reactive
  systems from quantitative specifications expressed in linear time temporal
  logic extended with mean-payoff and energy objectives.
  Further applications to temporal networks have been studied in~\cite{CR15}
  and~\cite{CPR17}.
  Consequently, efficient algorithms to solve mean-payoff games become essential
  ingredients to tackle these problems in practice.
  \\\indent
  Several algorithms have been devised in the past for the solution of the
  decision problem for mean-payoff games, which asks whether there exists a
  strategy for one of the players that grants the mean-payoff objective.
  The very first deterministic algorithm was proposed in~\cite{ZP96}, where it
  is shown that the problem can be solved with $\AOmicron{n^{3} \cdot m \cdot
  \WElm}$ arithmetic operations, with $n$ and $m$ the number of positions and
  moves, respectively, and $\WElm$ the maximal absolute weight in the game.
  A strategy improvement approach, based on iteratively adjusting a randomly
  chosen initial strategy for one player until a winning strategy is obtained,
  is presented in~\cite{Sch08a}, which has an exponential upper bound.
  The algorithm by Lifshits and Pavlov~\cite{LP07}, which runs in time
  $\AOmicron{n \cdot m \cdot 2^{n} \cdot \log_{2} \WElm}$, computes the
  ``potential'' of each game position, which corresponds to the initial credit
  that the player needs in order to win the game from that position.
  Algorithms based on the solution of linear feasibility problems over the
  tropical semiring have been also provided in~\cite{ABG14,ABG14a,ABGJ15}.
  The best known deterministic algorithm to date, which requires $\AOmicron{n
  \cdot m \cdot \WElm}$ arithmetic operations, was proposed by Brim
  \etal~\cite{BCDGR11}.
  They adapt to energy and mean-payoff games the notion of progress
  measures~\cite{Kla91}, as applied to parity games in~\cite{Jur00}.
  The approach was further developed in~\cite{CR17} to obtain the same
  complexity bound for the optimal strategy synthesis problem.
  A strategy-improvement refinement of this technique has been introduced
  in~\cite{BC12}.
  Finally, Bjork \etal~\cite{BSV04} proposed a randomized strategy-improvement
  based algorithm running in time $\min \{ \AOmicron{n^2 \cdot m \cdot \WElm},
  2^{\AOmicron{\sqrt{n} \cdot \log n}} \}$.
  \\\indent
  Our contribution is a novel mean-payoff progress measure approach that
  enriches such measures with the notion of \emph{quasi dominions}, originally
  introduced in~\cite{BDM16} for parity games.
  These are sets of positions with the property that as long as the opponent
  chooses to play to remain in the set, it loses the game for sure, hence its
  best choice is always to try to escape.
  A quasi dominion from where is not possible escaping is a winning set for the
  other player.
  Progress measure approaches, such as the one of~\cite{BCDGR11}, typically
  focus on finding the best choices of the opponent and little information is
  gathered on the other player.
  In this sense, they are intrinsically asymmetric.
  Enriching the approach with quasi dominions can be viewed as a way to also
  encode the best choices of the player, information that can be exploited to
  speed up convergence significantly.
  The main difficulty here is that suitable lift operators in the new setting do
  not enjoy monotonicity.
  Such a property makes proving completeness of classic progress measure
  approaches almost straightforward, as monotonic operators do admit a least
  fixpoint.
  Instead, the lift operator we propose is only inflationary (specifically,
  non-decreasing) and, while still admitting fixpoints~\cite{Bou49,Wit50}, need
  not have a least one.
  Hence, providing a complete solution algorithm proves more challenging.
  The advantages, however, are significant.
  On the one hand, the new algorithm still enjoys the same worst-case complexity
  of the best known algorithm for the problem proposed in~\cite{BCDGR11}.
  On the other hand, we show that there exist families of games on which the
  classic approach requires a number of operations that can be made arbitrarily
  larger than the one required by the new approach.
  Experimental results also witness the fact that this phenomenon is by no means
  isolated, as the new algorithm performs orders of magnitude better than the
  algorithm developed in~\cite{BCDGR11}.

\end{section}




\begin{section}{Mean-Payoff Games}
  \label{sec:mpg}

  A two-player turn-based \emph{arena} is a tuple $\AName =
  \tuplec{\PosSet[\PlrSym]}{\PosSet[\OppSym]}{\MovRel}$, with $\PosSet[\PlrSym]
  \cap \PosSet[\OppSym] = \emptyset$ and $\PosSet \defeq \PosSet[\PlrSym] \cup
  \PosSet[\OppSym]$, such that $\tupleb{\PosSet}{\MovRel}$ is a finite directed
  graph without sinks.
  $\PosSet[\PlrSym]$ (\resp, $\PosSet[\OppSym]$) is the set of positions of
  player $\PlrSym$ (\resp, $\OppSym$) and $\MovRel \subseteq \PosSet \times
  \PosSet$ is a left-total relation describing all possible moves.
  A \emph{path} in $\VSet \subseteq \PosSet$ is a finite or infinite sequence
  $\pthElm \in \PthSet(\VSet)$ of positions in $\VSet$ compatible with the move
  relation, \ie, $(\pthElm_{i}, \pthElm_{i + 1}) \in \MovRel$, for all $i \in
  \numco{0}{\card{\pthElm} - 1}$.
  A positional \emph{strategy} for player $\alpha \in \{ \PlrSym, \OppSym \}$ on
  $\VSet \subseteq \PosSet$ is a function
  $\strElm[\alpha] \in \StrSet[\alpha](\VSet) \subseteq (\VSet \cap
  \PosSet[\alpha]) \to \PosSet$, mapping each $\alpha$-position $\posElm$ in the
  domain of $\strElm[\alpha]$ to position $\strElm[\alpha](\posElm)$ compatible
  with the move relation, \ie,
  $(\posElm, \strElm[\alpha](\posElm)) \in \MovRel$.
  With $\StrSet[\alpha](\VSet)$ we denote the set of all $\alpha$-strategies on
  $\VSet$, while $\StrSet[\alpha]$ denotes $\bigcup_{\VSet \subseteq \PosSet}
  \StrSet[\alpha](\VSet)$.
  A \emph{play} in $\VSet \subseteq \PosSet$ from a position $\posElm \in \VSet$
  \wrt a pair of strategies $(\strElm[\PlrSym], \strElm[\OppSym]) \in
  \StrSet[\PlrSym](\VSet) \times \StrSet[\OppSym](\VSet)$, called
  \emph{$((\strElm[\PlrSym], \strElm[\OppSym]), \posElm)$-play}, is a path
  $\pthElm \in \PthSet(\VSet)$ such that $\pthElm[0] = \posElm$ and, for all $i
  \in \numco{0}{\card{\pthElm} - 1}$, if $\pthElm_{i} \in \PosSet[\PlrSym]$ then
  $\pthElm_{i + 1} = \strElm[\PlrSym](\pthElm_{i})$ else $\pthElm_{i + 1} =
  \strElm[\OppSym](\pthElm_{i})$.
  The \emph{play function} $\playFun : (\StrSet[\PlrSym](\VSet) \times
  \StrSet[\OppSym](\VSet)) \times \VSet \to \PthSet(\VSet)$ returns, for each
  position $\posElm \in \VSet$ and pair of strategies $(\strElm[\PlrSym],
  \strElm[\OppSym]) \in \StrSet[\PlrSym](\VSet) \times \StrSet[\OppSym](\VSet)$,
  the maximal $((\strElm[\PlrSym], \strElm[\OppSym]), \posElm)$-play
  $\playFun((\strElm[\PlrSym], \strElm[\OppSym]), \posElm)$.
  If a pair $(\strElm[\PlrSym], \strElm[\OppSym]) \in \StrSet[\PlrSym](\VSet)
  \times \StrSet[\OppSym](\VSet)$ induces a finite play starting from position
  $\posElm \in \VSet$, then $\playFun((\strElm[\PlrSym], \strElm[\OppSym]),
  \posElm)$ identifies the maximal prefix of that play that is contained in
  $\VSet$.

  A \emph{mean-payoff game} (\MPG for short) is a tuple $\GamName =
  \tuplec{\AName}{\WghSet}{\wghFun}$, where $\AName$ is an arena, $\WghSet
  \subset \SetZ$ is a finite set of integer weights, and $\wghFun \colon
  \PosSet \to \WghSet$ is a \emph{weight function} assigning a weight to each
  position.
  $\PosSet[][+]$ (\resp, $\PosSet[][-]$) denotes the set of positive-weight
  positions (\resp, non-positive-weight positions).
  For convenience, we shall refer to non-positive weights as negative weights.
  Notice that this definition of \MPG is equivalent to the classic formulation
  in which the weights label the moves, instead.
  The weight function naturally extends to paths, by setting $\wghFun(\pthElm)
  \defeq \sum_{i = 0}^{\card{\pthElm} - 1} \wghFun(\pthElm_{i})$.
  The goal of player $\PlrSym$ (\resp, $\OppSym$) is to maximize (\resp,
  minimize) $\vFun(\pthElm) \defeq \liminf_{i \to \infty} \frac{1}{i} \cdot
  \wghFun(\pthElm_{\leq i})$, where $\pthElm_{\leq i}$ is the prefix up to index
  $i$.
  Given a threshold $\nu$, a set of positions $\VSet \subseteq \PosSet$ is a
  $\PlrSym$-\emph{dominion}, if there exists a $\PlrSym$-strategy
  $\strElm[\PlrSym] \in \StrSet[\PlrSym](\VSet)$ such that, for all
  $\OppSym$-strategies $\strElm[\OppSym] \in \StrSet[\OppSym](\VSet)$ and
  positions $\posElm \in \VSet$, the induced play $\pthElm =
  \playFun((\strElm[\PlrSym], \strElm[\OppSym]), \posElm)$ satisfies
  $\vFun(\pthElm) > \nu$.
  The pair of winning regions $(\WinSet[\PlrSym],\WinSet[\OppSym])$ forms a
  $\nu$-mean partition.
  Assuming $\nu$ integer, the $\nu$-mean partition problem is equivalent to the
  $0$-mean partition one, as we can subtract $\nu$ to the weights of all the
  positions.
  As a consequence, the \MPG decision problem can be equivalently restated as
  deciding whether player $\PlrSym$ (\resp, $\OppSym$) has a strategy to enforce
  $\liminf_{i \to \infty} \frac{1}{i} \cdot \wghFun(\pthElm_{\leq i}) > 0$
  (\resp, $\liminf_{i \to \infty} \frac{1}{i} \cdot \wghFun(\pthElm_{\leq i})
  \leq 0$), for all the resulting plays $\pi$.

\end{section}




\begin{section}{Solving Mean-Payoff Games via Progress Measures}
  \label{sec:solprgmsr}

  The abstract notion of progress measure~\cite{Kla91} has been introduced as a
  way to encode global properties on paths of a graph by means of simpler local
  properties of adjacent vertexes.
  In the context of {\MPG}s, the graph property of interest, called
  \emph{mean-payoff property}, requires that the mean payoff of every infinite
  path in the graph be non-positive.
  More precisely, in game theoretic terms, a \emph{mean-payoff progress measure}
  witnesses the existence of strategy $\strElm[\OppSym]$ for player $\OppSym$
  such that each path in the graph induced by fixing that strategy on the arena
  satisfies the desired property.
  A mean-payoff progress measure associates with each vertex of the underlying
  graph a value, called \emph{measures}, taken from the set of extended natural
  numbers $\SetNI \defeq \SetN \cup \{\infty\}$, endowed with an ordering
  relation $\leq$ and an addition operation $+$, which extend the standard
  ordering and addition over the naturals in the usual way.
  Measures are associated with positions in the game and the measure of a
  position $\posElm$ can intuitively be interpreted as an estimate of the payoff
  that player $\PlrSym$ can enforce on the plays starting in $\posElm$.
  In this sense, they measure ``how far'' $\posElm$ is from satisfying the
  mean-payoff property, with the maximal measure $\infty$ denoting failure of
  the property for $\posElm$.
  More precisely, the $\OppSym$-strategy induced by a progress measure ensures
  that measures do not increase along the paths of the induced graph.
  This, in turn, ensures that every path eventually gets trapped in a
  non-positive-weight cycle, thereby witnessing a win for player $\OppSym$.

  To obtain a progress measure, one starts from some suitable association of
  position of the game with measures.
  The local information encoded by these measures is then propagated back along
  the edges of the underlying graph so as to associate with each position the
  information gathered along plays of some finite length starting from that
  position.
  The propagation process is performed according to the following intuition.
  The measures of positions adjacent to $\posElm$ are propagated back to
  $\posElm$ only if those measures push $\posElm$ further away from the
  property.
  This propagation is achieved by means of a measure stretch operation $+$,
  which adds, when appropriate, the weight of an adjacent position to the
  measure of a given position.
  This is established by comparing the measure of $\posElm$ with those of its
  adjacent positions, since, for each position $\posElm$, the mean-payoff
  property is defined in terms of the sum of the weights encountered along the
  plays from that position.
  The process ends when no position can be pushed further away from the property
  and each position is not dominated by any, respectively one, of its adjacents,
  depending on whether that position belongs to player $\PlrSym$ or to player
  $\OppSym$, respectively.
  The positions that did not reach measure $\infty$ are those from which player
  $\OppSym$ can win game and the set of measures currently associated with such
  positions forms a mean-payoff progress measure for the game.

  To make the above intuitions precise, we introduce the notion of measure
  function, progress measure, and an algorithm for computing progress measures
  correctly.
  It is worth noticing that the progress-measure based approach as described
  in~\cite{BCDGR11}, called \BRIM from now on, can be easily recast equivalently
  in the form below.
  A \emph{measure function} $\mfElm \colon \PosSet \cto \SetNI$ maps each
  position $\posElm$ in the game to a suitable measure $\mfElm(\posElm)$.
  The order $\leq$ of the measures naturally induces a pointwise partial order
  $\sqsubseteq$ on the measure functions defined in the usual way, namely, for
  any two measure functions $\mfElm[1]$ and $\mfElm[2]$, we write $\msrElm[1]
  \sqsubseteq \msrElm[2]$ if $\mfElm[1](\posElm) \leq \mfElm[2](\posElm)$, for
  all positions $\posElm$.
  The set of measure functions over a measure space, together with the induced
  ordering $\sqsubseteq$, forms a \emph{measure-function space}.

  \begin{definition}[Measure-Function Space]
    \label{def:msrfunspc}
    The \emph{measure-function space} is the \emph{partial order} $\MFName
    \defeq \tupleb {\MFSet} {\sqsubseteq}$ whose components are defined as
    reported in the following:
    \begin{enumerate}
      \item\label{def:msrfunspc(set)}
        $\MFSet \defeq \PosSet \to \SetNI$ is the set of all functions $\mfElm
        \in \MFSet$, called \emph{measure functions}, mapping each position
        $\posElm \in \PosSet$ to a measure $\mfElm(\posElm) \in \SetNI$;
      \item\label{def:msrfunspc(ord)}
        for all $\mfElm[1], \mfElm[2] \in \MFSet$, it holds that $\mfElm[1]
        \sqsubseteq \mfElm[2]$ if $\mfElm[1](\posElm) \leq \mfElm[2](\posElm)$,
        for all positions $\posElm \in \PosSet$.
    \end{enumerate}
    The \emph{$\PlrSym$-denotation} (\resp, \emph{$\OppSym$-denotation}) of a
    measure function $\mfElm \in \MFSet$ is the set $\denot{\mfElm}[\PlrSym]
    \defeq \mfElm[][-1](\infty)$ (\resp, $\denot{\mfElm}[\OppSym] \defeq
    \dual{\mfElm[][-1](\infty)}$) of all positions having maximal (\resp,
    non-maximal) measure associated within $\mfElm$.
  \end{definition}

  Assuming that a given position $\posElm$ has an adjacent with measure $\eta$,
  a measure update of $\eta$ \wrt $\posElm$ is obtained by the stretch operator
  $+ \colon \SetNI \times \PosSet \to \SetNI$, defined as
  \[
    \msrElm + \posElm \defeq \max \{ 0, \msrElm + \wghFun(\posElm) \},
  \]
  which corresponds to the payoff estimate that the given position will obtain
  by choosing to follow the move leading to the prescribed adjacent.

  A \emph{mean-payoff progress measure} is such that the measure associated with
  each game position $\posElm$ needs not be increased further in order to beat
  the actual payoff of the plays starting from $\posElm$.
  In particular, it can be defined by taking into account the opposite attitude
  of the two players in the game.
  While the player $\PlrSym$ tries to push toward higher measures, the player
  $\OppSym$ will try to keep the measures as low as possible.
  A measure function in which the payoff of each $\PlrSym$-position (\resp,
  $\OppSym$-position) $\posElm$ is not dominated by the payoff of all (\resp,
  some of) its adjacents augmented with the weight of $\posElm$ itself meets the
  requirements.

  \begin{definition}[Progress Measure]
    \label{def:prgmsr}
    A measure function $\mfElm \in \MFSet$ is a \emph{progress measure} if the
    following two conditions hold true, for all positions $\posElm \in \PosSet$:
    \begin{enumerate}
      \item\label{def:prgmsr(plr)}
        $\mfElm(\uposElm) + \posElm \leq \mfElm(\posElm)$, for all adjacents
        $\uposElm \in \MovRel(\posElm)$ of $\posElm$, if $\posElm \in
        \PosSet[\PlrSym]$;
      \item\label{def:prgmsr(opp)}
        $\mfElm(\uposElm) + \posElm \leq \mfElm(\posElm)$, for some adjacent
        $\uposElm \in \MovRel(\posElm)$ of $\posElm$, if $\posElm \in
        \PosSet[\OppSym]$.
    \end{enumerate}
  \end{definition}

  The following theorem states the fundamental property of progress measures,
  namely, that every position associated with a non-maximal value is won by
  player $\OppSym$.

  \begin{theorem}[Progress Measure]
    \label{thm:prgmsr}
    Let $\mfElm \in \MFSet$ be a progress measure.
    Then, $\denot{\mfElm}[\OppSym] \subseteq \WinSet[\OppSym]$.
  \end{theorem}

  In order to obtain a progress measure from a given measure function, one can
  iteratively adjust the current measure values in such a way to force the
  progress condition above among adjacent positions.
  To this end, we define the \emph{lift operator} $\liftFun \colon \MFSet \to
  \MFSet$ as follows:
  \[
    \liftFun (\mfElm)(\posElm) \defeq
    \begin{cases}
      \max {\set{ \mfElm(\wposElm) + \posElm }{ \wposElm \in \MovRel(\posElm)
      }},
      & \text{if } \posElm \in \PosSet[\PlrSym]; \\
      \min {\set{ \mfElm(\wposElm) + \posElm }{ \wposElm \in \MovRel(\posElm)
      }},
      & \text{otherwise}.
    \end{cases}
  \]
  Note that the lift operator is clearly monotone and, therefore, admits a least
  fixpoint.
  A mean-payoff progress measure can, then, be obtained by repeatedly applying
  this operator until a fixpoint is reached, starting from the minimal measure
  function $\mfElm[0] \defeq \{ \posElm \in \PosSet \mapsto 0 \}$ that assigns
  measure $0$ to all the positions in the game.
  The following \emph{solver operator} applied to $\mfElm[0]$ computes the
  desired solution:
  \[
    \solFun \defeq \lfpFun\, \mfElm \,.\, \liftFun(\mfElm) \colon \MFSet \to
    \MFSet.
  \]
  Observe that the measures generated by the procedure outlined above have a
  fairly natural interpretation.
  Each positive measure, indeed, under-approximates the weight that player
  $\PlrSym$ can enforce along finite prefixes of the plays from the
  corresponding positions.
  This follows from the fact that, while player $\PlrSym$ maximizes its measures
  along the outgoing moves, player $\OppSym$ minimizes them.
  In this sense, each positive measure witnesses the existence of a
  positively-weighted finite prefix of a play that player $\PlrSym$ can enforce.
  Let $\SElm \defeq \sum \set{ \wghFun(\posElm) \in \SetN }{ \posElm \in \PosSet
  \land \wghFun(\posElm) > 0 }$ be the sum of all the positive weights in the
  game.
  Clearly, the maximal payoff of a simple play in the underlying graph cannot
  exceed $\SElm$.
  Therefore, a measure greater than $\SElm$ witnesses the existence of a cycle
  whose payoff diverges to infinity and is won, thus, by player $\PlrSym$.
  Hence, any measure strictly greater than $\SElm$ can be substituted with the
  value $\infty$.
  This observation established the termination of the algorithm and is
  instrumental to its completeness proof.
  Indeed, at the fixpoint, the measures actually coincide with the highest
  payoff player $\PlrSym$ is able to guarantee.
  Soundness and completeness of the above procedure have been established
  in~\cite{BCDGR11}, where the authors also show that, despite the algorithm
  requiring $\AOmicron{n \cdot \SElm} = \AOmicron{n^{2} \cdot \WElm}$ lift
  operations in the worst-case, with $n$ the number of positions and $\WElm$ the
  maximal positive weight in the game, the overall cost of these lift operations
  is $\AOmicron{\SElm \cdot m \cdot \log \SElm} = \AOmicron{n \cdot m \cdot
  \WElm \cdot \log(n \cdot W)}$, with $m$ the number of moves and
  $\AOmicron{\log \SElm}$ the cost of each arithmetic operation necessary to
  compute the stretch of the measures.

\end{section}




\begin{section}{Solving Mean-Payoff Games via Quasi Dominions}
  
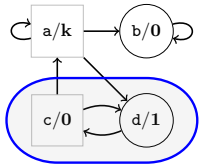
\label{sec:solqsidom}

  \begin{wrapfigure}[6]{i}{0.155\textwidth}
    \vspace{-1.75em}
    \figmpgexma
    \vspace{-0.75em}
    \caption{\label{fig:mpgexma} An \MPG.}
  \end{wrapfigure}
  Let us consider the simple example game depicted in Figure~\ref{fig:mpgexma},
  where the shape of each position indicates the owner, circles for player
  $\PlrSym$ and square for its opponent $\OppSym$, and, in each label of the
  form $\ell / \wghElm$, the letter $\wghElm$ corresponds to the associated
  weight, where we assume $k > 1$.
  Starting from the smallest measure function $\mfElm[0] = \{ \aSym, \bSym,
  \cSym, \dSym \mapsto 0 \}$, the first application of the lift operator returns
  $\mfElm[1] = \{ \aSym \mapsto k; \bSym, \cSym \mapsto 0; \dSym \mapsto 1 \} =
  \liftFun(\mfElm[0])$.
  After that step, the following iterations of the fixpoint alternatively
  updates positions $\cSym$ and $\dSym$, since the other ones already satisfy
  the progress condition.
  Being $\cSym \in \PosSet[\OppSym]$, the lift operator chooses for it the
  measure computed along the move $(\cSym, \dSym)$, thus obtaining
  $\mfElm[2](\cSym) = \liftFun(\mfElm[1])(\cSym) = \mfElm[1](\dSym) = 1$.
  Subsequently, $\dSym$ is updated to $\mfElm[3](\dSym) =
  \liftFun(\mfElm[2])(\dSym) = \mfElm[2](\cSym) + 1 = 2$.
  A progress measure is obtained after exactly $2k + 1$ iterations, when the
  measure of $\cSym$ reaches value $k$ and $\dSym$ value $k + 1$.
  Note, however, that the choice of the move $(\cSym, \dSym)$ is clearly a
  losing strategy for player $\OppSym$, as remaining in the highlighted region
  would make the payoff from position $\cSym$ diverge.
  Therefore, the only reasonable choice for player $\OppSym$ is to exit from
  that region by taking the move leading to position $\aSym$.
  An operator able to diagnose this phenomenon early on could immediately
  discard the move $(\cSym, \dSym)$ and jump directly to the correct payoff
  obtained by choosing the move to position $\aSym$.
  As we shall see, such an operator might lose the monotonicity property and
  recovering the completeness of the resulting approach will prove more
  involved.

  In the rest of this article we shall devise a progress operator that does
  precisely that.
  To this end, we start by providing a notion of \emph{quasi dominion},
  originally introduced for parity games in~\cite{BDM16}, which can be exploited
  in the context of {\MPG}s.

  \begin{definition}[Quasi Dominion]
    \label{def:qsidom}
    An arbitrary set of positions $\QSet \subseteq \PosSet$ is a \emph{quasi
    $\PlrSym$-dominion} if there exists a $\PlrSym$-strategy $\strElm[\PlrSym]
    \in \StrSet[\PlrSym](\QSet)$, called \emph{$\PlrSym$-witness for $\QSet$},
    such that, for all $\OppSym$-strategies $\strElm[\OppSym] \in
    \StrSet[\OppSym](\QSet)$ and positions $\posElm \in \QSet$, the induced play
    $\pthElm = \playFun((\strElm[\PlrSym], \strElm[\OppSym]), \posElm)$, called
    \emph{$(\strElm[\PlrSym], \posElm)$-play in $\QSet$}, satisfies
    $\wghFun(\pthElm) > 0$.
    If the condition $\wghFun(\pthElm) > 0$ holds only for infinite plays
    $\pthElm$, then $\QSet$ is called \emph{weak quasi $\PlrSym$-dominion}.
  \end{definition}

  Essentially, a quasi $\PlrSym$-dominion consists in a set $\QSet$ of positions
  starting from which player $\PlrSym$ can force plays in $\QSet$ of positive
  weight.
  Analogously, any infinite play that player $\PlrSym$ can force in a weak quasi
  $\PlrSym$-dominion has positive weight.
  Clearly, any quasi $\PlrSym$-dominion is also a weak quasi $\PlrSym$-dominion.
  Moreover, the latter are closed under subsets, while the former are not.
  It is an immediate consequence of the definition above that all infinite
  plays induced by the $\PlrSym$-witness, if any, necessarily have infinite
  weight and, thus, are winning for player $\PlrSym$.
  Indeed, every such a play $\pthElm$ is regular, \ie it can be decomposed into
  a prefix $\pthElm'$ and a simple cycle $(\pthElm'')^{\omega}$, \ie $\pthElm =
  \pthElm' (\pthElm'')^{\omega}$, since the strategies we are considering are
  memoryless.
  Now, $\wghFun((\pthElm'')^{\omega}) > 0$, so, $\wghFun(\pthElm'') > 0$, which
  implies $\wghFun((\pthElm'')^{\omega}) = \infty$.
  Hence, $\wghFun(\pthElm) = \infty$.

  \begin{proposition}
    \label{prp:win}
    Let $\QSet$ be a weak quasi $\PlrSym$-dominion with $\strElm[\PlrSym] \in
    \StrSet[\PlrSym](\QSet)$ one of its $\PlrSym$-witnesses and $\QSet[][\star]
    \subseteq \QSet$.
    Then, for all $\OppSym$-strategies $\strElm[\OppSym] \in
    \StrSet[\OppSym](\QSet[][\star])$ and positions $\posElm \in \QSet[][\star]$
    the following holds: if the $(\strElm[\PlrSym]_{\rst \QSet[][\star]},
    \posElm)$-play $\pthElm = \playFun((\strElm[\PlrSym]_{\rst \QSet[][\star]},
    \strElm[\OppSym]), \posElm)$ is infinite, then $\wghFun(\pthElm) = \infty$.
  \end{proposition}

  From Proposition~\ref{prp:win}, it directly follows that, if a weak quasi
  $\PlrSym$-dominion $\QSet$ is \emph{closed} \wrt its \emph{$\PlrSym$-witness},
  namely all the induced plays are infinite, then it is a $\PlrSym$-dominion,
  hence is contained in $\WinSet[\PlrSym]$.

  Consider again the example of Figure~\ref{fig:mpgexma}.
  The set of position $\QSet \defeq \{ \aSym, \cSym, \dSym \}$ forms a quasi
  $\PlrSym$-dominion whose $\PlrSym$-witness is the only possible
  $\PlrSym$-strategy mapping position $\dSym$ to $\cSym$.
  Indeed, any infinite play remaining in $\QSet$ forever and compatible with
  that strategy (\eg, the play from position $\cSym$ when player $\OppSym$
  chooses the move from $\cSym$ leading to $\dSym$ or the one from $\aSym$ to
  itself or the one from $\aSym$ to $\dSym$) grants an infinite payoff.
  Any finite compatible play, instead, ends in position $\aSym$ (\eg, the play
  from $\cSym$ when player $\OppSym$ chooses the move from $\cSym$ to $\aSym$
  and then one from $\aSym$ to $\bSym$) giving a payoff of at least $k > 0$.
  On the other hand, $\QSet[][\star] \defeq \{ \cSym, \dSym \}$ is only a weak
  quasi $\PlrSym$-dominion, as player $\OppSym$ can force a play of weight $0$
  from position $\cSym$, by choosing the exiting move $(\cSym, \aSym)$.
  However, the internal move $(\cSym, \dSym)$ would lead to an infinite play in
  $\QSet[][\star]$ of infinite weight.

  The crucial observation here is that the best choice for player $\OppSym$ in
  any position of a (weak) quasi $\PlrSym$-dominion is to exit from it as soon
  as it can, while the best choice for player $\PlrSym$ is to remain inside it
  as long as possible.
  The idea of the algorithm we propose in this section is to precisely exploit
  the information provided by the quasi dominions in the following way.
  Consider the example above.
  In position $\aSym$ player $\OppSym$ must choose to exit from $\QSet = \{
  \aSym, \cSym, \dSym \}$, by taking the move $(\aSym, \bSym)$, without changing
  its measure, which would corresponds to its weight $k$.
  On the other hand, the best choice for player $\OppSym$ in position $\cSym$ is
  to exit from the weak quasi-dominion $\QSet[][\star] = \{ \cSym, \dSym \}$, by
  choosing the move $(\cSym, \aSym)$ and lifting its measure from $0$ to $k$.
  Note that this contrasts with the minimal measure-increase policy for player
  $\OppSym$ employed in~\cite{BCDGR11}, which would keep choosing to leave
  $\cSym$ in the quasi-dominion by following the move to $\dSym$, which gives
  the minimal increase in measure of value $1$.
  Once $\cSym$ is out of the quasi-dominion, though, the only possible move for
  player $\PlrSym$ is to follow $\cSym$, taking measure $k + 1$.
  The resulting measure function is a progress measure and the solution has,
  thus, been reached.

  In order to make this intuitive idea precise, we need to be able to identify
  quasi dominions first.
  Interestingly enough, the measure functions $\mfElm$ defined in the previous
  section do allow to identify a quasi dominion, namely the set of positions
  $\dual{\mfElm[][-1](0)}$ having positive measure.
  Indeed, as observed at the end of that section, a positive measure witnesses
  the existence of a positively-weighted finite play that player $\PlrSym$ can
  enforce from that position onward, which is precisely the requirement of
  Definition~\ref{def:qsidom}.
  In the example of Figure~\ref{fig:mpgexma}, $\dual{\mfElm[0][-1](0)} =
  \emptyset$ and $\dual{\mfElm[1][-1](0)} = \{ \aSym, \cSym, \dSym\}$ are both
  quasi dominions, the first one \wrt the empty $\PlrSym$-witness and the second
  one \wrt the $\PlrSym$-witness $\strElm[\PlrSym](\dSym) = \cSym$.

  We shall keep the quasi-dominion information in pairs $(\mfElm, \strElm)$,
  called \emph{quasi-dominion representations} (\emph{\qdr}, for short),
  composed of a measure function $\mfElm$ and a $\PlrSym$-strategy $\strElm$,
  which corresponds to one of the $\PlrSym$-witnesses of the set of positions
  with positive measure in $\mfElm$.
  The connection between these two components is formalized in the definition
  below that also provides the partial order over which the new algorithm
  operates.

  \begin{definition}[\QDR Space]
    \label{def:qsidomrepspc}
    The \emph{quasi-dominion-representation space} is the \emph{partial order}
    $\QDRName \defeq \tupleb{\QDRSet}{\sqsubseteq}$, whose components are
    defined as prescribed in the following:
    \begin{enumerate}
      \item\label{def:qsidomrepspc(set)}
        $\QDRSet \subseteq \MFSet \times \StrSet[\PlrSym]$ is the set of all
        pairs $\qdrElm \defeq (\mfElm[\qdrElm], \strElm[\qdrElm]) \in \QDRSet$,
        called \emph{quasi-dominion-representations}, composed of a measure
        function $\mfElm[\qdrElm] \in \MFSet$ and a $\PlrSym$-strategy
        $\strElm[\qdrElm] \in \StrSet[\PlrSym](\qsiFun(\qdrElm))$, where
        $\qsiFun(\qdrElm) \defeq \dual{\mfElm[\qdrElm][-1](0)}$, for which the
        following four conditions hold:
        \begin{enumerate}
          \item\label{def:qsidomrepspc(set:qsi)}
            $\qsiFun(\qdrElm)$ is a quasi $\PlrSym$-dominion enjoying
            $\strElm[\qdrElm]$ as a $\PlrSym$-witness;
          \item\label{def:qsidomrepspc(set:dom)}
            $\denot{\mfElm[\qdrElm]}[\PlrSym]$ is a $\PlrSym$-dominion;
          \item\label{def:qsidomrepspc(set:plr)}
            $\mfElm[\qdrElm](\posElm) \leq
            \mfElm[\qdrElm](\strElm[\qdrElm](\posElm)) + \posElm$, for all
            $\PlrSym$-positions $\posElm \in \qsiFun(\qdrElm) \cap
            \PosSet[\PlrSym]$;
          \item\label{def:qsidomrepspc(set:opp)}
            $\mfElm[\qdrElm](\posElm) \leq \mfElm[\qdrElm](\uposElm) + \posElm$,
            for all $\OppSym$-positions $\posElm \in \qsiFun(\qdrElm) \cap
            \PosSet[\OppSym]$ and adjacents $\uposElm \in \MovRel(\posElm)$;
        \end{enumerate}
      \item\label{def:qsidomrepspc(ord)}
        for all $\qdrElm[1], \qdrElm[2] \in \QDRSet$, it holds that $\qdrElm[1]
        \sqsubseteq \qdrElm[2]$ if $\mfElm[{\qdrElm[1]}] \sqsubseteq
        \mfElm[{\qdrElm[2]}]$ and $\strElm[{\qdrElm[1]}](\posElm) =
        \strElm[{\qdrElm[2]}](\posElm)$, for all $\PlrSym$-positions $\posElm
        \in \qsiFun(\qdrElm[1]) \cap \PosSet[\PlrSym]$ with
        $\mfElm[{\qdrElm[1]}](\posElm) = \mfElm[{\qdrElm[2]}](\posElm)$.
    \end{enumerate}
    The \emph{$\alpha$-denotation} $\denot{\qdrElm}[\alpha]$ of a \qdr
    $\qdrElm$, with $\alpha \in \{ \PlrSym, \OppSym \}$, is the
    $\alpha$-denotation $\denot{\mfElm[\qdrElm]}[\alpha]$ of its measure
    function.
  \end{definition}

  Condition~\ref{def:qsidomrepspc(set:qsi)} is obvious.
  Condition~\ref{def:qsidomrepspc(set:dom)}, instead, requires that every
  position with infinite measure is indeed won by player $\PlrSym$ and is
  crucial to guarantee the completeness of the algorithm.
  Finally, Conditions~\ref{def:qsidomrepspc(set:plr)}
  and~\ref{def:qsidomrepspc(set:opp)} ensure that every positive measure
  under approximates the actual weight of some finite play within the induced
  quasi dominion.
  This is formally captured by the following proposition, which can be easily
  proved by induction on the length of the play.

  \begin{proposition}
    \label{prp:qdrpth}
    Let $\qdrElm$ be a \qdr and $\posElm \pthElm \uposElm$ a finite path
    starting at position $\posElm \in \PosSet$ and terminating in position
    $\uposElm \in \PosSet$ compatible with the $\PlrSym$-strategy
    $\strElm[\qdrElm]$.
    Then, $\mfElm[\qdrElm](\posElm) \leq \wghFun(\posElm \pthElm) +
    \mfElm[\qdrElm](\uposElm)$.
  \end{proposition}

  It is immediate to see that every MPG admits a non-trivial \QDR space, since
  the pair $(\mfElm[0], \strElm[0])$, with $\mfElm[0]$ the smallest measure
  function and $\strElm[0]$ the empty strategy, trivially satisfies all the
  required conditions.

  \begin{proposition}
    \label{prp:qdrnonemp}
    Every \MPG has a non-empty \QDR space associated with it.
  \end{proposition}

  The solution procedure we propose, called \QDPM from now on, can intuitively
  be broken down as an alternation of two phases.
  The first one tries to lift the measures of positions outside the quasi
  dominion $\qsiFun(\qdrElm)$ in order to extend it, while the second one lifts
  the positions inside $\qsiFun(\qdrElm)$ that can be forced to exit from it by
  player $\OppSym$.
  The algorithm terminates when no new position can be absorbed within the quasi
  dominion and no measure needs to be lifted to allow the $\OppSym$-winning
  positions to exit from it, when possible.
  To this end, we define a controlled lift operator $\liftFun \colon \QDRSet
  \!\times\! \pow{\PosSet} \!\times\! \pow{\PosSet} \pto \QDRSet$ that works on
  {\qdr}s and takes two additional parameters, a source and a target set of
  positions.
  The intended meaning is that we want to restrict the application of the lift
  operation to the positions in the source set $\SSet$, while using only the
  moves leading to the target set $\TSet$.
  The different nature of the two types of lifting operations is reflected in
  the actual values of the source and target parameters.
  \[
    \liftFun (\qdrElm, \SSet, \TSet) \defeq \qdrElm[][\star], \text{ where}
  \]
  \[
    \mfElm[{\qdrElm[][\star]}](\posElm) \defeq
    \begin{cases}
      \max {\set{ \mfElm[\qdrElm](\uposElm) + \posElm }{ \uposElm \in
      \MovRel(\posElm) \cap \TSet }},
      & \text{if } \posElm \in \SSet \cap \PosSet[\PlrSym]; \\
      \min {\set{ \mfElm[\qdrElm](\uposElm) + \posElm }{ \uposElm \in
      \MovRel(\posElm) \cap \TSet }},
      & \text{if } \posElm \in \SSet \cap \PosSet[\OppSym]; \\
      \mfElm[\qdrElm](\posElm),
      & \text{otherwise};
    \end{cases}
  \]
  and, for all $\PlrSym$-positions $\posElm \in \qsiFun(\qdrElm[][\star]) \cap
  \PosSet[\PlrSym]$,
  \[
    \strElm[ {\qdrElm[][\star]} ](\posElm) \in \argmax_{\uposElm \in
    \MovRel(\posElm) \cap \TSet}\: \mfElm[\qdrElm](\uposElm) + \posElm, \text{
    if } \mfElm[{\qdrElm[][\star]}](\posElm) \neq \mfElm[\qdrElm](\posElm),
    \text{ and } \strElm[ {\qdrElm[][\star]} ](\posElm) =
    \strElm[\qdrElm](\posElm), \text{ otherwise} .
  \]
  Except for the restriction on the outgoing moves considered, which are those
  leading to the targets in $\TSet$, the lift operator acts on the measure
  component of a \qdr very much like the original lift operator does.
  In order to ensure that the result is still a \qdr, however, the lift operator
  must also update the $\PlrSym$-witness of the quasi dominion.
  This is required to guarantee that Conditions~\ref{def:qsidomrepspc(set:qsi)}
  and~\ref{def:qsidomrepspc(set:plr)} of Definition~\ref{def:qsidomrepspc} are
  preserved.
  If the measure of a $\PlrSym$-position $\posElm$ is not affected by the lift,
  the $\PlrSym$-witness must not change for that position.
  On the other hand, if the application of the lift operation increases the
  measure, then the $\PlrSym$-witness on $\posElm$ needs to be updated to any
  move $(\posElm, \uposElm)$ that grants measure
  $\mfElm[{\qdrElm[][\star]}](\posElm)$ to $\posElm$.
  In principle, more than one such move may exist and any one of them can serve
  the purpose as witness.

  The solution algorithm can then be expressed as the inflationary
  fixpoint~\cite{Bou49,Wit50} of the composition of the two phases mentioned
  above, defined by the progress operators $\prgFun[0]$ and $\prgFun[+]$.
  \[
    \solFun \defeq \ifpFun\, \qdrElm \,.\, \prgFun[+](\prgFun[0](\qdrElm))
    \colon \QDRSet \pto \QDRSet.
  \]
  The first phase is computed by the operator $\prgFun[0] \colon \QDRSet \pto
  \QDRSet$, defined as follows:
  \[
    \prgFun[0](\qdrElm) \defeq \sup \{ \qdrElm, \liftFun(\qdrElm,
    \dual{\qsiFun(\qdrElm)}, \PosSet) \}.
  \]
  This operator is responsible of enforcing the progress condition on the
  positions outside the quasi dominion $\qsiFun(\qdrElm)$ that do not satisfy
  the inequalities between the measures along a move leading to
  $\qsiFun(\qdrElm)$ itself.
  It does that by applying the lift operator with $\dual{\qsiFun(\qdrElm)}$ as
  source and no restrictions on the moves.
  Those position that acquire a positive measure in this phase contribute to
  enlarging the current quasi dominion.
  Observe that the strategy component of the \qdr is updated so that it is a
  $\PlrSym$-witness of the new quasi dominion.
  To guarantee that measures never decrease, the supremum \wrt the \QDR-space
  ordering is taken as result.

  \begin{lemma}
    \label{lmm:prgzer}
    Let $\qdrElm \in \QDRSet$ be a fixpoint of $\prgFun[0]$.
    Then, $\mfElm[\qdrElm]$ is a progress measure over
    $\dual{\qsiFun(\qdrElm)}$.
  \end{lemma}

  The second phase, instead, implements the mechanism intuitively described
  above, while analyzing the simple example of Figure~\ref{fig:mpgexma}.
  This is achieved by the operator $\prgFun[+]$ reported in
  Algorithm~\ref{alg:prg}.
  The procedure iteratively examines the current quasi dominion by lifting the
  measures of the positions that must exit from it.
  Specifically, it processes $\qsiFun(\qdrElm)$ layer by layer, starting from
  the outer layer of positions that must escape from.
  The process ends when a, possibly empty, closed weak quasi dominion is
  obtained.
  Recall that all the positions in a closed weak quasi dominion are necessarily
  winning for player $\PlrSym$, due to Proposition~\ref{prp:win}.
  We distinguish two sets of positions in $\qsiFun(\qdrElm)$.
  Those that already satisfy the progress condition and those that do not.
  The measures of first ones already witness an escape route from
  $\qsiFun(\qdrElm)$.
  The other ones, instead, are those whose current choice is to remain inside
  it.
  For instance, when considering the measure function $\mfElm[2]$ in the example
  of Figure~\ref{fig:mpgexma}, position $\aSym$ belongs to the first set, while
  positions $\cSym$ and $\dSym$ to the second one, since the choice of $\cSym$
  is to follow the internal move $(\cSym, \dSym)$.

  Since the only positions that change measure are those in the second set, only
  such positions need to be examined.
  To identify them, which form a weak quasi dominion $\dmnFun(\qdrElm)$ strictly
  contained in $\qsiFun(\qdrElm)$, we proceed as follows.
  First, we collect the set $\nppFun(\qdrElm)$ of positions in
  $\qsiFun(\qdrElm)$ that do not satisfy the progress condition, called the
  \emph{non-progress positions}.
  Then, we compute the set of positions that will have no choice other than
  reaching $\nppFun(\qdrElm)$.
  The non-progress positions are computed as follows.
  \begin{linenomath}
  \begin{align*}
    \nppFun(\qdrElm)
    & \defeq {\set{ \posElm \in \qsiFun(\qdrElm) \cap \PosSet[\PlrSym] }{
    \exists \uposElm \in \MovRel(\posElm) \,.\, \mfElm[{\qdrElm}](\posElm) <
    \mfElm[{\qdrElm}](\uposElm) + \posElm }} \\
    & \,\cup {\set{ \posElm \in \qsiFun(\qdrElm) \cap \PosSet[\OppSym] }{
    \forall \uposElm \in \MovRel(\posElm) \,.\, \mfElm[{\qdrElm}](\posElm) <
    \mfElm[{\qdrElm}](\uposElm) + \posElm }}.
  \end{align*}
  \end{linenomath}
  The remaining positions in $\dmnFun(\qdrElm)$ are collected as the
  inflationary fixpoint of the following operator.
  \begin{linenomath}
  \begin{align*}
    \preFun(\qdrElm, \QSet)
    & \defeq \QSet \cup {\set{ \posElm \in \qsiFun(\qdrElm) \cap
    \PosSet[\PlrSym] }{ \strElm[{\qdrElm}](\posElm) \in \QSet }} \\
    & \,\cup {\set{ \posElm \in \qsiFun(\qdrElm) \cap \PosSet[\OppSym] }{
    \forall \uposElm \in \MovRel(\posElm) \setminus \QSet \,.\,
    \mfElm[{\qdrElm}](\posElm) < \mfElm[{\qdrElm}](\uposElm) + \posElm }}.
  \end{align*}
  \end{linenomath}
  The final result is
  \[
    \dmnFun(\qdrElm) \defeq (\ifpFun\, \QSet \,.\, \preFun(\qdrElm,
    \QSet))(\nppFun(\qdrElm))
  \]
  Intuitively, $\dmnFun(\qdrElm)$ contains all the $\PlrSym$-positions that are
  forced to reach $\nppFun(\qdrElm)$ via the quasi-dominion $\PlrSym$-witness
  and all the $\OppSym$-positions that can only avoid reaching
  $\nppFun(\qdrElm)$ by strictly increasing their measure, which player
  $\OppSym$ wants obviously to prevent.

  It is important to observe that, from a functional view-point, the progress
  operator $\prgFun[+]$ would work just as well if applied to
  the entire quasi dominion $\qsiFun(\qdrElm)$, since it would simply leave
  unchanged the measure of those positions that already satisfy the progress
  condition.
  However, it is crucial that only the positions in $\dmnFun(\qdrElm)$ are
  processed in order to achieve the best asymptotic complexity bound known to
  date.
  We shall reiterate on this point later on.

  \begin{wrapfigure}[11]{r}{0.38\textwidth}
    \vspace{-2.25em}
    \algprg
    \vspace{-2.0em}
  \end{wrapfigure}
  At each iteration of the while-loop of Algorithm~\ref{alg:prg}, let $\QSet$
  denote the current (weak) quasi dominion, initially set to $\dmnFun(\qdrElm)$
  (Line~1).
  It first identifies the positions in $\QSet$ that can immediately escape from
  it (Line~2).
  Those are \emph{(i)} all the $\OppSym$-position with a move leading outside of
  $\QSet$ and \emph{(ii)} the $\PlrSym$-positions $\posElm$ whose
  $\PlrSym$-witness $\strElm[\qdrElm]$ forces $\posElm$ to exit from $\QSet$,
  namely $\strElm[\qdrElm](\posElm) \not\in \QSet$, and that cannot strictly
  increase their measure by choosing to remain in $\QSet$.
  While the condition for $\OppSym$-position is obvious, the one for
  $\PlrSym$-positions require some explanation.
  The crucial observation here is that, while player $\PlrSym$ does indeed
  prefer to remain in the quasi dominion, it can only do so while ensuring that
  by changing strategy it does not enable infinite plays within $\QSet$ that are
  winning for the adversary.
  In other words, the new $\PlrSym$-strategy must still be a $\PlrSym$-witness
  for $\QSet$ and this can only be ensured if the new choice strictly increases
  its measure.
  The operator $\escFun \colon \QDRSet \!\times\! \pow{\PosSet} \to
  \pow{\PosSet}$ formalizes the idea:
  \begin{linenomath}
  \begin{align*}
    \escFun(\qdrElm, \QSet)
    & \defeq {\set{ \posElm \in \QSet \cap \PosSet[\OppSym] }{ \MovRel(\posElm)
    \setminus \QSet \neq \emptyset }} \\
    & \:\cup {\set{ \posElm \in \QSet \cap \PosSet[\PlrSym] }{
    \strElm[\qdrElm](\posElm) \not\in \QSet \land \forall \uposElm \in
    \MovRel(\posElm) \cap \QSet \,.\, \mfElm[\qdrElm](\uposElm) + \posElm \leq
    \mfElm[\qdrElm](\posElm)) }}.
  \end{align*}
  \end{linenomath}

  \begin{wrapfigure}[9]{i}{0.195\textwidth}
    \vspace{-1.50em}
    \figmpgexmb
    \vspace{-1.0em}
    \caption{\label{fig:mpgexmb} Another \MPG.}
  \end{wrapfigure}
  Consider, for instance, the example in Figure~\ref{fig:mpgexmb} and a \qdr
  $\qdrElm$ such that $\mfElm[\qdrElm] = \{ \aSym \mapsto 3; \bSym \mapsto 2;
  \cSym, \dSym, \fSym \mapsto 1; \eSym \mapsto 0 \}$ and $\strElm[\qdrElm] = \{
  \bSym \mapsto \aSym; \fSym \mapsto \dSym \}$.
  In this case, we have $\QSet[\qdrElm] = \{ \aSym, \bSym, \cSym, \dSym, \fSym
  \}$ and $\dmnFun(\qdrElm) = \{ \cSym, \dSym, \fSym \}$, since $\cSym$ is the
  only non-progress positions, $\dSym$ is forced to follow $\cSym$ in order to
  avoid the measure increase required to reach $\bSym$, and $\fSym$ is forced by
  the $\PlrSym$-witness to reach $\dSym$.
  Now, consider the situation where the current weak quasi dominion is $\QSet =
  \{ \cSym, \fSym \}$, \ie after $\dSym$ has escaped from $\dmnFun(\qdrElm)$.
  The escape set of $\QSet$ is $\{\cSym, \fSym\}$.
  To see why the $\PlrSym$-position $\fSym$ is escaping, observe that
  $\mfElm[\qdrElm](\fSym) + \fSym = 1 = \mfElm[\qdrElm](\fSym)$ and that,
  indeed, should player $\PlrSym$ choose to change its strategy and take the
  move $(\fSym, \fSym)$ to remain in $\QSet$, it would obtain an infinite play
  with payoff $0$, thus violating the definition of weak quasi dominion.

  Before proceeding, we want to stress an easy consequence of the definition of
  the notion of escape set and Conditions~\ref{def:qsidomrepspc(set:plr)}
  and~\ref{def:qsidomrepspc(set:opp)} of Definition~\ref{def:qsidomrepspc}, \ie,
  that every escape position of the quasi dominion $\qsiFun(\qdrElm)$ can only
  assume its weight as possible measure inside a \qdr $\qdrElm$, as reported is
  the following proposition.
  This observation, together with Proposition~\ref{prp:qdrpth}, precisely
  ensures that the measure of a position $\posElm \in \qsiFun(\qdrElm)$ is an
  under approximation of the weight of all finite plays leaving
  $\qsiFun(\qdrElm)$.

  \begin{proposition}
    \label{prp:qdresc}
    Let $\qdrElm$ be a \qdr.
    Then, $\mfElm[\qdrElm](\posElm) = \wghFun(\posElm) > 0$, for all $\posElm
    \in \escFun(\qdrElm, \qsiFun(\qdrElm))$.
  \end{proposition}

  Now, going back to the analysis of the algorithm, if the escape set is
  non-empty, we need to select the escape positions that need to be lifted in
  order to satisfy the progress condition.
  The main difficulty is to do so in such a way that the resulting measure
  function still satisfies Condition~\ref{def:qsidomrepspc(set:opp)} of
  Definition~\ref{def:qsidomrepspc}, for all the $\OppSym$-positions with
  positive measure.
  The problem occurs when a $\OppSym$-position can exit either immediately or
  passing through a path leading to another position in the escape set.
  Consider again the example above, where $\QSet = \dmnFun(\qdrElm) = \{ \cSym,
  \dSym, \fSym \}$.
  If position $\dSym$ immediately escapes from $\QSet$ using the move $(\dSym,
  \bSym)$, it would change its measure to $\mfElm'(\dSym) = \mfElm(\bSym) +
  \dSym = 2 > \mfElm(\dSym) = 1$.
  Now, position $\cSym$ has two ways to escape, either directly with move
  $(\cSym, \aSym)$ or by reaching the other escape position $\dSym$ passing
  through $\fSym$.
  The first choice would set its measure to $\mfElm(\aSym) + \cSym = 4$.
  The resulting measure function, however, would not satisfy
  Condition~\ref{def:qsidomrepspc(set:opp)} of
  Definition~\ref{def:qsidomrepspc}, as the new measure of $\cSym$ would be
  greater than $\mfElm'(\dSym) + \cSym = 2$, preventing to obtain a \qdr.
  Similarly, if position $\dSym$ escapes from $\QSet$ passing through $\cSym$
  via the move $(\cSym, \aSym)$, we would have $\mfElm''(\dSym) =
  \mfElm''(\cSym) + \dSym = (\mfElm(\aSym) + \cSym) + \dSym = 4 > 2 =
  \mfElm(\bSym) + \dSym$, still violating
  Condition~\ref{def:qsidomrepspc(set:opp)}.
  Therefore, in this specific case, the only possible way to escape is to reach
  $\bSym$.
  The solution to this problem is simply to lift in the current iteration only
  those positions that obtain the lowest possible measure increase, hence
  position $\dSym$ in the example, leaving the lift of $\cSym$ to some
  subsequent iteration of the algorithm that would choose the correct escape
  route via $\dSym$.
  To do so, we first compute the minimal measure increase, called the
  \emph{best-escape forfeit}, that each position in the escape set would obtain
  by exiting the quasi dominion immediately.
  The positions with the lowest possible forfeit, called \emph{best-escape
  positions}, can all be lifted at the same time.
  The intuition is that the measure of all the positions that escape from a
  (weak) quasi dominion will necessarily be increased of at least the minimal
  best-escape forfeit.
  This observation is at the core of the proof of Theorem~\ref{thm:tot} (see
  the appendix) ensuring that the desired properties of {\qdr}s are preserved by
  the operator $\prgFun[+]$.
  The set of best-escape positions is computed by the operator $\bepFun \colon
  \QDRSet \!\times\! \pow{\PosSet} \to \pow{\PosSet}$ as follows:
  \[
    \bepFun(\qdrElm, \QSet) \defeq \argmin_{\posElm \in \escFun(\qdrElm, \QSet)}
    \befFun(\mfElm[\qdrElm], \QSet, \posElm),
  \]
  where the operator $\befFun \colon \MFSet \!\times\! \pow{\PosSet} \!\times\!
  \PosSet \to \SetN[\infty]$ computes, for each position $\posElm$ in a quasi
  dominion $\QSet$, its best-escape forfeit:
  \[
    \befFun(\mfElm, \QSet, \posElm) \defeq
    \begin{cases}
      \max {\set{ \mfElm(\uposElm) + \posElm - \mfElm(\posElm) }{ \uposElm \in
      \MovRel(\posElm) \setminus \QSet }},
      & \text{if } \posElm \in \PosSet[\PlrSym]; \\
      \min {\set{ \mfElm(\uposElm) + \posElm - \mfElm(\posElm) }{ \uposElm \in
      \MovRel(\posElm) \setminus \QSet }},
      & \text{otherwise}.
    \end{cases}
  \]
  In our example, $\befFun(\mfElm, \QSet, \cSym) = \mfElm(\aSym) + \cSym -
  \mfElm(\cSym) = 4 - 1 = 3$, while $\befFun(\mfElm, \QSet, \dSym) =
  \mfElm(\bSym) + \dSym - \mfElm(\dSym) = 2 - 1 = 1$.
  Therefore, $\bepFun(\qdrElm, \QSet) = \{ \dSym \}$.

  Once the set $\ESet$ of best-escape positions is identified (Line~3 of the
  algorithm), the procedure simply lifts them restricting the possible moves to
  those leading outside the current quasi dominion (Line~4).
  Those positions are, then, removed from the set (Line~5), thus obtaining a
  smaller weak quasi dominion ready for the next iteration.

  The algorithm terminates when the (possibly empty) current quasi dominion
  $\QSet$ is closed.
  By virtue of Proposition~\ref{prp:win}, all those positions belong to
  $\WinSet[\PlrSym]$ and their measure is set to $\infty$ by means of the
  operator $\winFun \colon \QDRSet \!\times\! \pow{\PosSet} \pto \QDR$ (Line~6),
  which also computes the winning $\PlrSym$-strategy on those positions.
  \[
    \winFun (\qdrElm, \QSet) \defeq \qdrElm[][\star]\text{, where}
  \]
  \[
    \mfElm[{\qdrElm[][\star]}] \defeq {\mfElm[\qdrElm]}[\QSet \mapsto \infty]
  \]
  and, for all $\PlrSym$-positions $\posElm \in \qsiFun(\qdrElm[][\star]) \cap
  \PosSet[\PlrSym]$,
  \[
    \strElm[ {\qdrElm[][\star]} ](\posElm) \in \argmax_{\uposElm \in
    \MovRel(\posElm) \cap \QSet}\: \mfElm[\qdrElm](\uposElm) + \posElm, \text{
    if } \strElm[\qdrElm](\posElm) \not\in \QSet \text{ and }
    \strElm[ {\qdrElm[][\star]} ](\posElm) = \strElm[\qdrElm](\posElm), \text{
    otherwise}.
  \]
  Observe that, since we know that every $\PlrSym$-position $\posElm \in \QSet
  \cap  \PosSet[\PlrSym]$, whose current $\PlrSym$-witness leads outside
  $\QSet$, is not an escape position, any move $(\posElm, \uposElm)$ within
  $\QSet$ that grants the maximal stretch $\mfElm[\qdrElm](\uposElm) + \posElm$
  strictly increases its measure and, therefore, is a possible choice for a
  $\PlrSym$-witness of the $\PlrSym$-dominion $\QSet$.

  At this point, it should be quite evident that the progress operator
  $\prgFun[+]$ is responsible of enforcing the progress condition on the
  positions inside the quasi dominion $\qsiFun(\qdrElm)$, thus, the following
  necessarily holds.

  \begin{lemma}
    \label{lmm:prgpls}
    Let $\qdrElm \in \QDRSet$ be a fixpoint of $\prgFun[+]$.
    Then, $\mfElm[\qdrElm]$ is a progress measure over $\qsiFun(\qdrElm)$.
  \end{lemma}

  \begin{wrapfigure}[6]{i}{0.225\textwidth}
    \vspace{-1.50em}
    \figmpgexmc
    \vspace{-1.0em}
    \caption{\label{fig:mpgexmc} Yet another \MPG.}
  \end{wrapfigure}
  We now exemplify the lack of monotonicity of the progress operator
  $\prgFun[+]$.
  To do so, consider the game of Figure~\ref{fig:mpgexmc} and the following two
  {\qdr}s $\qdrElm[1]$ and $\qdrElm[2]$ defined via their components:
  $\mfElm[{\qdrElm[1]}] = \{ \aSym \mapsto 3; \bSym \mapsto 0; \cSym \mapsto 2;
  \dSym, \eSym \mapsto 1 \}$ and $\strElm[{\qdrElm[1]}] = \{ \eSym \mapsto
  \dSym \}$; $\mfElm[{\qdrElm[2]}] = \{ \aSym \mapsto 3; \bSym \mapsto 0; \cSym,
  \eSym \mapsto 2; \dSym \mapsto 1 \}$ and $\strElm[{\qdrElm[2]}] = \{ \eSym
  \mapsto \cSym \}$.
  Obviously, $\qdrElm[1] \sqsubset \qdrElm[2]$.
  However, $\qdrElm[1][\star] \defeq \prgFun[+](\qdrElm[1]) \not\sqsubseteq
  \qdrElm[2][\star] \defeq \prgFun[+](\qdrElm[2])$.
  Indeed, $\mfElm[{\qdrElm[1][\star]}] = \{ \aSym \mapsto 3; \bSym \mapsto 0;
  \cSym \mapsto 2; \dSym, \eSym \mapsto 4 \}$, while
  $\mfElm[{\qdrElm[2][\star]}] = \{ \aSym \mapsto 3; \bSym \mapsto 0;
  \cSym, \eSym \mapsto 2; \dSym, \mapsto 3 \}$, which implies that
  $\qdrElm[2][\star] \sqsubset \qdrElm[1][\star]$.
  Moreover, $\qdrElm[1][\star]$ is already a progress measure, while
  $\qdrElm[2][\star]$ requires another application of $\prgFun[+]$ in order to
  solve the game, since $\qdrElm[1][\star] = \prgFun[+](\qdrElm[2][\star])$.

  In order to prove the correctness of the proposed algorithm, we first need to
  ensure that any quasi-dominion space $\QDRName$ is indeed closed under the
  operators $\prgFun[0]$ and $\prgFun[+]$. This is established by the following
  theorem, which states that the operators are total functions on that space.

  \begin{theorem}[Totality]
    \label{thm:tot}
    The progress operators $\prgFun[0]$ and $\prgFun[+]$ are total inflationary
    functions.
  \end{theorem}

  Since both operators are inflationary, so is their composition, which admits
  fixpoint.
  Therefore, the operator $\solFun$ is well defined.
  Moreover, following the same considerations discussed at the end of
  Section~\ref{sec:solprgmsr}, it can be proved the fixpoint is obtained after
  at most $n \cdot (\SElm + 1)$ iterations.
  Let $\ifpFun[k] \,\XElm\, . \,\FFun(\XElm)$ denote the $k$-th iteration of an
  inflationary operator $\FFun$.
  Then, we have the following theorem.

  \begin{theorem}[Termination]
    \label{thm:ter}
    The solver operator $\solFun \defeq \ifpFun\, \qdrElm \,.\,
    \prgFun[+](\prgFun[0](\qdrElm))$ is a well-defined total function.
    Moreover, for every $\qdrElm \in \QDRSet$ it holds that $\solFun(\qdrElm) =
    (\ifpFun[k]\, \qdrElm[][\star] \,.\,
    \prgFun[+](\prgFun[0](\qdrElm[][\star])))(\qdrElm)$, for some index $k \leq
    n \cdot (\SElm + 1)$, where $n$ is the number of positions in the \MPG and
    $\SElm \defeq \sum \set{ \wghFun(\posElm) \in \SetN }{ \posElm \in \PosSet
    \land \wghFun(\posElm) > 0 }$ the total sum of its positive weights.
  \end{theorem}

  Consider, as a final example, the game depicted in Figure~\ref{fig:sim}, with
  $k > 2$, where the numbers denote the weights of the positions of the game, in
  the picture labeled $(0)$, and the measures assigned by the procedure, in the
  remaining ones.
  Each picture also features both the $\PlrSym$-witness strategy in dashed blue
  and the best counter $\OppSym$-strategy in dashed red for the current quasi
  dominion.
  Moreover, solid colored moves are moves along which the measure strictly
  increases.
  Below each picture, we also indicate the phase, $\prgFun[0]$ or $\prgFun[+]$,
  that produces the displayed result.
  \begin{figure}[htbp]
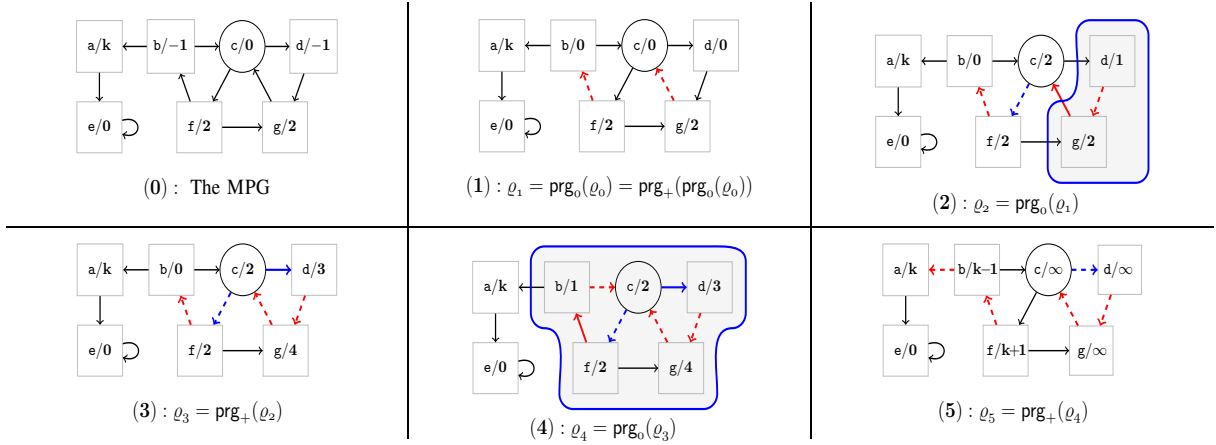

    \vspace{-0.50em}
    \tabsimexm
    \caption{\label{fig:sim} A simulation.}
    \vspace{-0.50em}
  \end{figure}
  The computation starts from the initial \qdr $\qdrElm[0] = (\mfElm[0],
  \strElm[0])$, assigning measure $0$ to all the positions of the game with the
  associated empty strategy.
  The first iteration applies $\prgFun[0]$ to $\qdrElm[0]$, which lifts
  positions $\aSym$, $\fSym$, and $\gSym$ to their respective weights, leading
  to $\qdrElm[1]$ as shown in Picture~$(1)$.
  At this point, $\qsiFun(\qdrElm[1]) = \{ \aSym, \fSym, \gSym\}$ but
  $\dmnFun(\qdrElm[1])$ is empty, as all those positions already satisfy the
  progress condition, thus, $\prgFun[+]$ does nothing.
  In the next iteration, $\prgFun[0]$ applied to $\qdrElm[1]$ results in the
  lifting of positions $\cSym$ and $\dSym$, as reported in Picture~$(2)$.
  Position $\cSym$ is a $\PlrSym$-position and the lift operator chooses
  $(\cSym, \fSym)$ as its strategy.
  The resulting quasi-dominion is $\qsiFun(\qdrElm[2]) = \{ \aSym, \cSym, \dSym,
  \fSym, \gSym \}$ and $\dmnFun(\qdrElm[2]) = \{ \dSym, \gSym \}$, with $\gSym$
  the only escape position that is also non-progress.
  The measure of $\gSym$ is lifted to $\mfElm[2](\cSym) + \gSym = 4$.
  Finally, it is the turn of position $\dSym$ to be lifted to $\mfElm[2](\gSym)
  + \dSym = 3$.
  Picture~$(3)$ shows the resulting \qdr $\qdrElm[3]$.
  The final iteration first applies $\prgFun[0]$ to $\qdrElm[3]$
  (Picture~$(4)$), lifting position $\bSym$ to measure $1$ via the move $(\bSym,
  \cSym)$.
  This change of measure triggers another application of $\prgFun[+]$, as
  position $\fSym$ is now non-progress.
  The resulting \qdr $\qdrElm[4]$ is such that $\qsiFun(\qdrElm[4]) = \{ \aSym,
  \bSym, \cSym, \dSym, \fSym, \gSym \}$ and $\dmnFun(\qdrElm[4]) = \{ \bSym,
  \cSym, \dSym, \fSym, \gSym \}$.
  The only escape position is $\bSym$, which is lifted directly to measure $k -
  1$.
  In the remaining set $\{ \cSym, \dSym, \fSym, \gSym \}$, the only escape
  position is $\fSym$, which is lifted to measure $k + 1$.
  The resulting weak quasi dominion $\{ \cSym, \dSym, \gSym\}$, however, is
  closed, since $\mfElm[{\qdrElm[4]}](\cSym) = 2 < \mfElm[{\qdrElm[4]}](\dSym) +
  \cSym = 3$.
  Therefore, player $\PlrSym$ changes strategy and chooses the move
  $(\cSym, \dSym)$.
  Since no escape positions remain, the set $\{ \cSym, \dSym, \gSym \}$ is
  winning for player $\PlrSym$ and the $\winFun$ operator lifts all their
  measures to $\infty$, leading to $\qdrElm[5]$ in Picture~$(5)$.
  The measure function $\mfElm[5]$ is now a progress measure and the algorithm
  terminates.
  The total number of single measure updates for \QDPM to reach the fixpoint on
  the example of Figure~\ref{fig:sim} is $13$, regardless of the value of the
  maximal weight $k$ in the game assigned to position $\aSym$.

  On the other hand, it can easily be proved that \BRIM~\cite{BCDGR11} requires
  $3k + 8$ applications of its lift operator to compute a progress measure, for
  a total of $5k + 9$ measure updates.
  Indeed, the first two evaluations of $\liftFun$, starting from $\mfElm[0]$,
  lead to $\mfElm[2] = \{ \aSym \mapsto k; \bSym, \eSym \mapsto 0; \cSym, \fSym,
  \gSym \mapsto 2; \dSym \mapsto 1 \}$, as in Picture~$(2)$, and require $5$
  measure lifts.
  Then, the algorithm iteratively increases the measures of $\bSym$, $\gSym$,
  $\dSym$, $\fSym$, and $\cSym$ by applying $3(k - 1)$ times the lift operator,
  for a total of $5(k - 1)$ measure lifts: $\mfElm[3i] = {\mfElm[3i - 1]}[\bSym
  \mapsto i; \gSym \mapsto i + 3]$, $\mfElm[3i + 1] = {\mfElm[3i]}[\dSym, \fSym
  \mapsto i + 2]$, and $\mfElm[3i + 2] = {\mfElm[3i + 1]}[\cSym \mapsto i + 2]$,
  for all $i \in \numcc{1}{k - 1}$.
  At this point, $\bSym$ and $\fSym$ have obtained measures $k - 1$ and $k + 1$,
  respectively, which suffice to satisfy the progress relation along the moves
  $(\fSym, \bSym)$ and $(\bSym, \aSym)$.
  However, the $\OppSym$-position $\gSym$ does not satisfy such a relation along
  its unique move $(\gSym, \cSym)$, since $\mfElm[3k - 1](\gSym) = k + 2 <
  \mfElm[3k - 1](\cSym) + \gSym = (k + 1) + 2 = k + 3$.
  Therefore, other six applications of $\liftFun$ are needed before $\gSym$ can
  exceed the bound $\SElm = \wghFun(\aSym) + \wghFun(\fSym) + \wghFun(\gSym) =
  k + 4$.
  Each one of them modifies the measure of one position only, for a total of $6$
  lifts: $\mfElm[3(k + i)] = {\mfElm[3(k + i) - 1]}[\gSym \mapsto k + 3 + i]$,
  $\mfElm[3(k + i) + 1] = {\mfElm[3(k + i)]}[\dSym \mapsto k + 2 + i]$, and
  $\mfElm[3(k + i) + 2] = {\mfElm[3(k + i) + 1]}[\cSym \mapsto k + 2 + i]$, for
  $i \in \{ 0, 1 \}$.
  At this point, we have $\mfElm[3k + 6] = {\mfElm[3k + 5]}[\gSym \mapsto
  \infty]$, $\mfElm[3k + 7] = {\mfElm[3k + 6]}[\dSym \mapsto \infty]$, and,
  finally, $\mfElm[3k + 8] = {\mfElm[3k + 7]}[\cSym \mapsto \infty]$, which
  contribute with the remaining $3$ lifts.
  From this observation, the next result immediately follows.

  \begin{theorem}[Efficiency]
    \label{thm:eff}
    An infinite family of {\MPG}s $\{ \GamName[k] \}_{k}$ exists on which \QDPM
    requires a constant number of measure updates, while \BRIM requires
    $\AOmicron{k}$ such updates.
  \end{theorem}

  From Theorem~\ref{thm:prgmsr}, together with Lemmas~\ref{lmm:prgzer}
  and~\ref{lmm:prgpls}, it follows that the solution provided by the algorithm
  is indeed a progress measure, hence establishing soundness.

  \begin{theorem}[Soundness]
    \label{thm:snd}
    $\denot{\solFun(\qdrElm)}[\OppSym] \subseteq \WinSet[\OppSym]$, for every
    $\qdrElm \in \QDRSet$.
  \end{theorem}

  On the other hand, Theorem~\ref{thm:ter}, together with
  Condition~\ref{def:qsidomrepspc(set:dom)} of
  Definition~\ref{def:qsidomrepspc}, ensures that all the positions with
  infinite measure are winning for player $\PlrSym$, hence the algorithm is also
  complete.

  \begin{theorem}[Completeness]
    \label{thm:com}
    $\denot{\solFun(\qdrElm)}[\PlrSym] \subseteq \WinSet[\PlrSym]$, for every
    $\qdrElm \in \QDRSet$.
  \end{theorem}

  The following lemma ensures that each execution of the operator $\prgFun[+]$
  strictly increases the measure of all the positions in $\dmnFun(\qdrElm)$.

  \begin{lemma}
    \label{lmm:qdrchg}
    Let $\qdrElm[][\star] \defeq \prgFun[+](\qdrElm)$, for some $\qdrElm \in
    \QDRSet$.
    Then, $\mfElm[{\qdrElm[][\star]}](\posElm) > \mfElm[\qdrElm](\posElm)$, for
    all positions $\posElm \in \dmnFun(\qdrElm)$.
  \end{lemma}

  Recall that each position can at most be lifted $\SElm + 1 = \AOmicron{ n
  \cdot \WElm}$ times and, by the previous lemma, the complexity of $\solFun$
  only depends on the cumulative cost of such lift operations.
  We can express, then, the total cost as the sum, over the set of positions in
  the game, of the cost of all the lift operations performed on that positions.
  Each such operation can be computed in time linear in the number of incoming
  and outgoing moves of the corresponding lifted position $\posElm$, namely
  $\AOmicron{(\card{\MovRel(\posElm)} + \card{\MovRel[][-1](\posElm)}) \cdot
  \log \SElm}$, with $\AOmicron{\log \SElm}$ the cost of each arithmetic
  operation involved.
  Summing all up, the actual asymptotic complexity of the procedure can,
  therefore, be expressed as $\AOmicron{n \cdot m \cdot \WElm \cdot \log(n \cdot
  \WElm)}$.

  \begin{theorem}[Complexity]
    \label{thm:cmp}
    \QDPM requires time $\AOmicron{n \cdot m \cdot \WElm \cdot \log(n \cdot
    \WElm)}$ to solve an \MPG with $n$ positions, $m$ moves, and maximal
    positive weight $\WElm$.
  \end{theorem}

\end{section}




\begin{section}{Experimental Evaluation}

  \begin{wrapfigure}[11]{r}{0.48\textwidth}
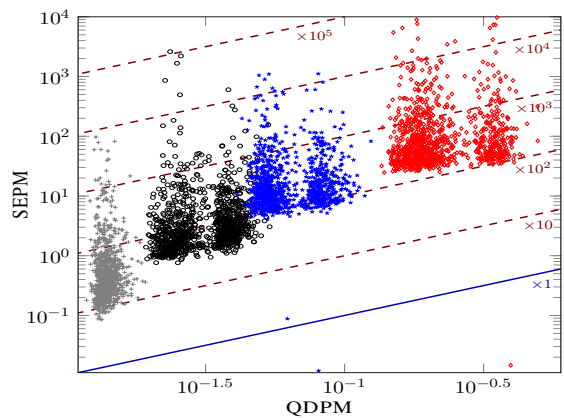

    \vspace{-5.5em}
    \figsctplt
    \vspace{-1em}
    \caption{\label{fig:expevl}\! Experiments on random games with $5000$
      positions.}
  \end{wrapfigure}
  In order to assess the effectiveness of the proposed approach, we implemented
  both \QDPM and \BRIM~\cite{BCDGR11}, the most efficient known solution to the
  problem and the more closely related one to \QDPM, in \texttt{C++} within
  \textsc{Oink}~\cite{Dij18}.
  \textsc{Oink} has been developed as a framework to compare parity game
  solvers.
  However, extending the framework to deal with {\MPG}s is not difficult.
  The form of the arenas of the two types of games essentially coincide, the
  only relevant difference being that {\MPG}s allow negative numbers to label
  game positions.
  We ran the two solvers against randomly generated {\MPG}s of various
  sizes.~\footnote{The experiments were carried out on a 64-bit 3.9GHz quad-core
  machine, with \textsc{Intel} i5-6600K processor and 8GB of RAM, running
  \textsc{Ubuntu}~18.04.}

  Figure~\ref{fig:expevl} compares the solution time, expressed in seconds, of
  the two algorithms on $4000$ games, each with $5000$ positions and randomly
  assigned weights in the range $[-15000,15000]$.
  The scale of both axes is logarithmic.
  The experiments are divided in $4$ clusters, each containing $1000$ games.
  The benchmarks in different clusters differ in the maximal number $m$ of
  outgoing moves per position, with $m \in \{ 10, 20, 40, 80 \}$.
  These experiments clearly show that \QDPM substantially outperforms \BRIM.
  Most often, the gap between the two algorithms is between two and three orders
  of magnitude, as indicated by the dashed diagonal lines.
  It also shows that \BRIM is particularly sensitive to the density of the
  underlying graph, as its performance degrades significantly as the number of
  moves increases.
  The maximal solution time was $8940$ sec. for \BRIM and $0.5$ sec. for \QDPM.

  Figure~\ref{fig:exptwo}, instead, compares the two algorithms fixing the
  maximal out-degree of the underlying graphs to $2$, in the left-hand picture,
  and to $40$, in the right-hand one, while increasing the number of positions
  from $10^3$ to $10^5$ along the x-axis.
  Each picture displays the performance results on $2800$ games.
  Each point shows the total time to solve $100$ randomly generated games with
  that given number of positions, which increases by $1000$ up to size $2 \cdot
  10^3$ and by $10000$, thereafter.
  In both pictures the scale is logarithmic.
  For the experiments in the right-hand picture we had to set a timeout for
  \BRIM to 45 minutes per game, which was hit most of the times on the bigger
  ones.

    \begin{figure}[t]
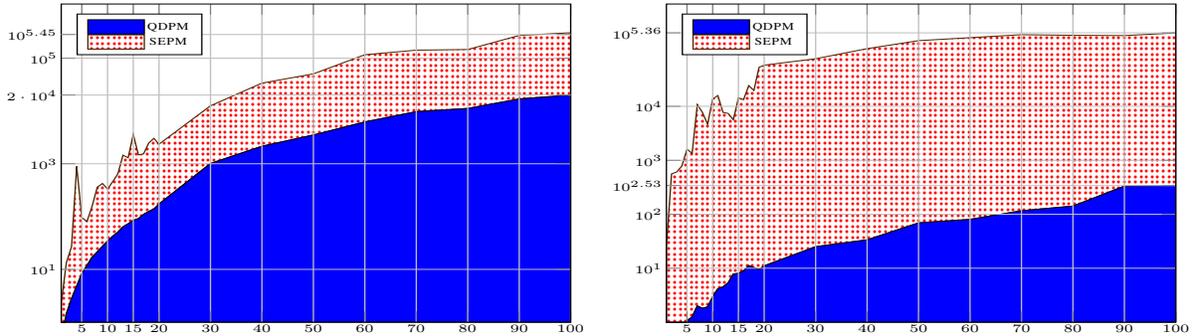

    \expmoves
    \vspace{-1em}
    \caption{\label{fig:exptwo} Total solution times in seconds of \BRIM and
      \QDPM on $5600$ random games.}
    \vspace{-.7em}
  \end{figure}

  Once again, the \QDPM significantly outperforms \BRIM on both kinds of
  benchmarks, with a gap of more than an order of magnitude on the first ones,
  and a gap of more than three orders of magnitude on the second ones.
  The results also confirm that the performance gap grows considerably as the
  number of moves per position increases.

  We are not aware of actual concrete benchmarks for {\MPG}s.
  However, exploiting the standard encoding of parity games into mean-payoff
  games~\cite{Jur98}, we can compare the behavior of \BRIM and \QDPM on concrete
  verification problems encoded as parity games.
  For completeness, Table~\ref{tab:prty} reports some experiments on such
  problems.
  \begin{wraptable}[11]{R}{0.52\textwidth}
    \vspace{-1.2em}
    \begin{center}
      \scriptsize
      \scalebox{1.00}[1.00] {
        \begin{tabular}{|l|r|r||r|r|}
          \hline & & & & \\[-.7em]
          Benchmark & Positions & Moves & \BRIM & \QDPM
          \\ \hline
          \hline  & & & & \\[-.7em]
          Elevator 1  & $144$  & $234$     & $0.04$     & $\mathbf{0.0010}$
          \\
          Elevator 2  & $564$   & $950$   & $8.80$     & $\mathbf{0.0042}$
          \\
          Elevator 3  & $2688$ & $4544$  & $4675.71$  & $\mathbf{0.0064}$
          \\ \hline & & & & \\[-.7em]
          Lang. Incl. 1  & $170$  & $1094$   & $3.18$  & $\mathbf{0.0021}$
          \\
          Lang. Incl. 2  & $304$  & $1222$   & $16.76$  & $\mathbf{0.0019}$
          \\
          Lang. Incl. 3  & $428$  & $878$   & $20.25$  & $\mathbf{0.0033}$
          \\
          Lang. Incl. 4  & $628$  & $1538$   & $135.51$  & $\mathbf{0.0029}$
          \\
          Lang. Incl. 5  & $509$  & $2126$   & $148.37$  & $\mathbf{0.0034}$
          \\
          Lang. Incl. 6  & $835$  & $2914$   & $834.90$  & $\mathbf{0.0051}$
          \\
          Lang. Incl. 7  & $1658$ & $4544$   & $2277.87$  & $\mathbf{0.0100}$
          \\
          \hline
        \end{tabular}
      } \vspace{0.125em}
      \caption{\label{tab:prty} \small Experiments
        on concrete verification problems.}
    \end{center}
  \end{wraptable}

  The table reports the execution times, expressed in seconds, required by the
  two algorithms to solve instances of two classic verification problems: the
  Elevator Verification and the Language Inclusion problems.
  These two benchmarks are included in the PGSolver~\cite{FL09} toolkit and are
  often used as benchmarks for parity games solvers.
  The first benchmark is a \emph{verification under fairness} constraints of a
  simple model of an elevator, while the second one encodes the \emph{language
  inclusion} problem between a non-deterministic B\"uchi automaton and a
  deterministic one.
  The results on various instances of those problems confirm that \QDPM
  significantly outperforms the classic progress measure approach.
  Note also that the translation into {\MPG}s, which encodes priorities as
  weights whose absolute value is exponential in the values of the priorities,
  leads to games with weights of high magnitude.
  Hence, the results in Table~\ref{tab:prty} provide further evidence that \QDPM
  is far less dependent on the absolute value of the weights.
  They also show that \QDPM can be very effective for the solution of real-world
  qualitative verification problems.
  It is worth noting, though, that the translation from parity to {\MPG}s gives
  rise to weights that are exponentially distant from each other~\cite{Jur98}.
  As a consequence, the resulting benchmarks are not necessarily representative
  of {\MPG}s, being a very restricted subclass.
  Nonetheless, they provide evidence of the applicability of the approach in
  practical scenarios.


\end{section}









\begin{section}{Concluding Remarks}

  We proposed a novel solution algorithm for the decision problem of {\MPG}s
  that integrates progress measures and quasi dominions.
  We argue that the integration of these two concepts may offer significant
  speed up in convergence to the solution, at no additional computational cost.
  This is evidenced by the existence of a family of games on which the combined
  approach can perform arbitrarily better than a classic progress measure based
  solution.
  Experimental results also show that the introduction of quasi dominions can
  often reduce solution times up to three order of magnitude, suggesting that
  the approach may be very effective in practical applications as well.
  We believe that the integration approach we devised is general enough to be
  applied to other types of games.
  In particular, the application of quasi dominions in conjunction with progress
  measure based approaches, such as those of~\cite{JL17} and~\cite{FJSSW17}, may
  lead to practically efficient quasi polynomial algorithms for parity games and
  their quantitative extensions.

\end{section}




  \bibliographystyle{IEEEtranS}
  \bibliography{References}

  \appendices



\begin{section}{Proofs}
  \label{app:prf}

  In this appendix, we collect some supplementary material, providing three
  further lemmas (Lemmas~\ref{lmm:qdrstr},~\ref{lmm:qdrbnd},
  and~\ref{lmm:fixpnt}), and the proofs of the theorems and lemmas introduced in
  Sections~\ref{sec:solprgmsr} and~\ref{sec:solqsidom}.

  \renewcounter{theorem}{thm:prgmsr}
  \begin{theorem}[Progress Measure]
    Let $\mfElm \in \MFSet$ be a progress measure.
    Then, $\denot{\mfElm}[\OppSym] \subseteq \WinSet[\OppSym]$.
  \end{theorem}
  \begin{proof}
    Consider a $\OppSym$-strategy $\strElm[\OppSym] \in \StrSet[\OppSym]$ for
    which all measures $\mfElm(\posElm)$ of positions $\posElm \in
    \denot{\mfElm}[\OppSym] \cap \PosSet[\OppSym]$ are a progress at $\posElm$
    \wrt the measures $\mfElm(\strElm[\OppSym](\posElm))$ of their adjacents
    $\strElm[\OppSym](\posElm)$, formally, $\mfElm(\strElm[\OppSym](\posElm)) +
    \posElm \leq \mfElm(\posElm)$.
    The existence of such a strategy is ensured by the fact that $\mfElm$ is a
    progress measure.
    Indeed, by Condition~\ref{def:prgmsr(opp)} of Definition~\ref{def:prgmsr},
    there necessarily exists a adjacent $\uposElm[][\star] \in
    \MovRel(\posElm)$ of $\posElm$ such that $\mfElm(\uposElm[][\star]) +
    \posElm \leq \mfElm(\posElm)$.
    Now, it can be shown that $\strElm[\OppSym]$ is a winning strategy for
    player~$\OppSym$ from all the positions in $\denot{\mfElm}[\OppSym]$, which
    implies that $\denot{\mfElm}[\OppSym] \subseteq \WinSet[\OppSym]$.
    To do this, let us consider a $\PlrSym$-strategy $\strElm[\PlrSym] \in
    \StrSet[\PlrSym]$ and the associated play $\pthElm =
    \playFun((\strElm[\PlrSym], \strElm[\OppSym]), \posElm)$ starting at a
    position $\posElm \in \denot{\mfElm}[\OppSym]$.
    Assume, by contradiction, that $\pthElm$ is won by player~$\PlrSym$.
    Since the game $\GamName$ is finite, $\pthElm$ must contain a finite simple
    cycle, and so a finite simple path, with strictly positive total weight sum.
    In other words, there exist two natural numbers $h \in \SetN$ and $k \in
    \SetN[+]$ such that $(\pthElm)_{h} = (\pthElm)_{h + k}$ and $\wghFun(\rho)
    = \sum_{i = h}^{h + k - 1} \wghFun((\pthElm)_{i}) > 0$, where $\rho \defeq
    ((\pthElm)_{\geq h})_{< h + k}$ is the simple path named above.
    Now, recall that, $((\pthElm)_{i}, (\pthElm)_{i + 1}) \in \MovRel$, for
    all indexes $i \in \SetN$.
    Thus, by both conditions of Definition~\ref{def:prgmsr}, and the notion of
    play, we have that
    \[
      \mfElm((\pthElm)_{i + 1}) + (\pthElm)_{i} \leq \mfElm((\pthElm)_{i}).
    \]
    Via a trivial induction, it is immediate to see that $\mfElm((\pthElm)_{i})
    \leq \SElm$, where $\SElm \defeq \sum \set{ \wghFun(\posElm) \in \SetN }{
    \posElm \in \PosSet \land \wghFun(\posElm) > 0 } < \infty$, for all $i \in
    \SetN$, since $\mfElm((\pthElm)_{0}) = \mfElm(\posElm) \neq \infty$, being
    $\posElm \in \denot{\mfElm}[\OppSym]$.
    As a consequence, due to the definition of the measure stretch operator, it
    holds that
    \[
      \mfElm((\pthElm)_{i + 1}) + \wghFun((\pthElm)_{i}) \leq
      \mfElm((\pthElm)_{i}) \leq \SElm.
    \]
    Hence, by summing together all the inequalities having indexes $i \in \SetN$
    with $h \leq i < h + k$, we obtain
    \[
      \sum_{i = h + 1}^{h + k} \payFun[\mfElm]((\pthElm)_{i}) +
      \sum_{i = h}^{h + k - 1} \wghFun((\pthElm)_{i}) \leq
      \sum_{i = h}^{h + k - 1} \payFun[\mfElm]((\pthElm)_{i}) < \infty,
    \]
    which simplifies in $\wghFun(\rho) = \sum_{i = h}^{h + k - 1}
    \wghFun((\pthElm)_{i}) \leq 0$, since $\payFun[\mfElm]((\pthElm)_{h + k}) =
    \payFun[\mfElm]((\pthElm)_{h})$.
    However, this contradicts the above assumption $\wghFun(\rho) > 0$.
    Therefore, $\strElm[\OppSym]$ is a winning strategy for player~$\OppSym$
    on $\denot{\mfElm}[\OppSym]$ as required by the theorem statement.
  \end{proof}

  In the following, for a finite path $\pthElm$ of an \MPG, we denote by
  $\fst{\pthElm}$ and $\lst{\pthElm}$ the first and last positions of $\pthElm$,
  respectively.

  \begin{lemma}
    \label{lmm:qdrstr}
    Let $\qdrElm \in \QDRSet$ and $\strElm[][\star] \in
    \StrSet[\PlrSym](\qsiFun(\qdrElm))$ a $\PlrSym$-strategy such that, if
    $\strElm[][\star](\posElm) \neq \strElm[\qdrElm](\posElm)$, then
    $\mfElm[\qdrElm](\posElm) < \mfElm[\qdrElm](\strElm[][\star](\posElm)) +
    \posElm$, for all positions $\posElm \in \qsiFun(\qdrElm) \cap
    \PosSet[\PlrSym]$.
    Then, $\strElm[][\star]$ is a $\PlrSym$-witness for $\qsiFun(\qdrElm)$.
  \end{lemma}
  \begin{proof}
    The proof proceed by induction on the number $i \defeq \card{\DSet}$ of the
    positions in $\DSet \defeq \set{ \posElm \in \PosSet[\PlrSym] }{
    \strElm[][\star](\posElm) \neq \strElm[\qdrElm](\posElm) }$ on which the two
    strategies $\strElm[][\star]$ and $\strElm[\qdrElm]$ differ.
    The base case $i = 0$ is immediate, since $\qdrElm$ is a \qdr.
    Therefore, assume $i > 0$, let $\posElm \in \DSet$, and consider the
    strategy $\der{\strElm} \in \StrSet[\PlrSym](\qsiFun(\qdrElm))$ such that
    $\der{\strElm}(\posElm) = \strElm[\qdrElm](\posElm)$ and
    $\der{\strElm}(\uposElm) = \strElm[][\star](\uposElm)$, for all positions
    $\uposElm \in \qsiFun(\qdrElm) \cap \PosSet[\PlrSym]$ with $\uposElm \neq
    \posElm$.
    By the inductive hypothesis, we have that $\der{\strElm}$ is a
    $\PlrSym$-witness for the quasi dominion $\qsiFun(\qdrElm)$.
    Now, consider an arbitrary path $\pthElm$ compatible with the
    $\PlrSym$-strategy $\strElm[][\star]$.
    If $\pthElm$ does not meet $\posElm$, it is necessarily compatible with the
    $\PlrSym$-strategy $\der{\strElm}$, thus, $\wghFun(\pthElm) > 0$.
    If $\pthElm$ meets $\posElm$ once, then it can be decomposed as $\pthElm'
    \posElm \pthElm''$, where $\pthElm'$ and $\pthElm''$ are paths not meeting
    $\posElm$, where only the first can be possibly empty.
    On the one hand, if $\pthElm''$ is infinite, by Proposition~\ref{prp:win},
    we have $\wghFun(\pthElm'') = \infty$ and, so, $\wghFun(\pthElm) =
    \wghFun(\pthElm' \posElm \pthElm'') = \infty$.
    On the other hand, if $\pthElm''$ is finite, then, by
    Propositions~\ref{prp:qdrpth} and~\ref{prp:qdresc}, we have that $0 <
    \mfElm[\qdrElm](\fst{\pthElm' \posElm}) \leq \wghFun(\pthElm') +
    \mfElm[\qdrElm](\lst{\pthElm' \posElm}) = \wghFun(\pthElm') +
    \mfElm[\qdrElm](\posElm)$ and $\mfElm[\qdrElm](\fst{\pthElm''}) \leq
    \wghFun(\pthElm'')$, since both
    $\pthElm'$ and $\pthElm''$ are compatible with $\der{\strElm}$.
    Moreover, $\mfElm[\qdrElm](\posElm) <
    \mfElm[\qdrElm](\strElm[][\star](\posElm)) + \posElm =
    \mfElm[\qdrElm](\fst{\pthElm''}) + \posElm =
    \mfElm[\qdrElm](\fst{\pthElm''}) + \wghFun(\posElm)$.
    Now, by putting all things together, we have $0 <
    \mfElm[\qdrElm](\fst{\pthElm' \posElm}) \leq \wghFun(\pthElm') +
    \mfElm[\qdrElm](\posElm) < \wghFun(\pthElm') +
    \mfElm[\qdrElm](\fst{\pthElm''}) + \wghFun(\posElm) \leq \wghFun(\pthElm') +
    \wghFun(\posElm) + \wghFun(\pthElm'') = \wghFun(\pthElm' \posElm
    \pthElm'')$, \ie, $\wghFun(\pthElm) > 0$.
    Finally, consider the case where $\pthElm$ meets $\posElm$ more than once
    and, so, infinitely many times, due to the regularity of the path, which is
    in its turn due to the memoryless strategies.
    Then, $\pthElm$ can be written as $\pthElm' (\posElm \pthElm'')^{\omega} =
    \pthElm' \posElm (\pthElm'' \posElm)^{\omega}$, where $\pthElm'$ and
    $\pthElm''$ are possibly empty paths not meeting $\posElm$.
    First observe that the finite path $\pthElm'' \posElm$ is compatible with
    $\der{\strElm}$, thus, by Proposition~\ref{prp:qdrpth}, we have that
    $\mfElm[\qdrElm](\fst{\pthElm'' \posElm}) \leq \wghFun(\pthElm'') +
    \mfElm[\qdrElm](\posElm)$.
    Moreover, $\mfElm[\qdrElm](\posElm) < \mfElm[\qdrElm](\fst{\pthElm''
    \posElm}) + \wghFun(\posElm)$, as already shown above.
    Hence, $\mfElm[\qdrElm](\posElm) < \mfElm[\qdrElm](\fst{\pthElm''}) +
    \wghFun(\posElm) \leq \wghFun(\pthElm'') + \wghFun(\posElm) +
    \mfElm[\qdrElm](\posElm) = \wghFun(\pthElm'' \posElm) +
    \mfElm[\qdrElm](\posElm)$, which implies $\wghFun(\pthElm'' \posElm) > 0$.
    As a consequence, $\wghFun((\pthElm'' \posElm)^{\omega}) = \infty$ and, so,
    $\wghFun(\pthElm) = \wghFun(\pthElm' \posElm (\pthElm'' \posElm)^{\omega}) >
    0$.
    Summing up, $\strElm[][\star]$ is a $\PlrSym$-witness for the quasi dominion
    $\qsiFun(\qdrElm)$ as required by the lemma statement.
  \end{proof}

  \renewcounter{theorem}{thm:tot}
  \begin{theorem}[Totality]
    The progress operators $\prgFun[0]$ and $\prgFun[+]$ are total inflationary
    functions.
  \end{theorem}
  \begin{proof}
    The proof proceeds by showing that, for each $\qdrElm \in \QDRSet$, the
    elements $\prgFun[0](\qdrElm)$ and $\prgFun[+](\qdrElm)$ are \qdr too.
    We also prove that $\qdrElm \sqsubseteq \prgFun[0](\qdrElm)$ and $\qdrElm
    \sqsubseteq \prgFun[+](\qdrElm)$.
    The two operators are analyzed separately.
    \begin{itemize}
      \item
        \textbf{[$\prgFun[0]$].}
        Let $\qdrElm[][\star] \defeq \prgFun[0](\qdrElm) = \sup \{ \qdrElm,
        \liftFun(\qdrElm, \dual{\qsiFun(\qdrElm)}, \PosSet) \} \sqsupseteq
        \qdrElm$.
        It is obvious, so, that $\prgFun[0]$ is inflationary.
        Consider now a position $\posElm \in \qsiFun(\qdrElm[][\star])$.
        Recall that $\mfElm[{\qdrElm[][\star]}](\posElm) > 0$.
        If $\posElm \in \qsiFun(\qdrElm)$, by definition of the lift operator,
        it holds that $\mfElm[{\qdrElm[][\star]}](\posElm) =
        \mfElm[\qdrElm](\posElm)$ and $\strElm[{\qdrElm[][\star]}](\posElm) =
        \strElm[\qdrElm](\posElm)$, thus the appropriate condition between
        Conditions~\ref{def:qsidomrepspc(set:plr)}
        and~\ref{def:qsidomrepspc(set:opp)} of Definition~\ref{def:qsidomrepspc}
        is verified, since $\qdrElm \in \QDRSet$.
        Thus, assume $\posElm \in \dual{\qsiFun(\qdrElm)}$.
        If $\posElm \in \PosSet[\PlrSym]$, we have that
        $\mfElm[{\qdrElm[][\star]}](\posElm) = \max \set{
        \mfElm[\qdrElm](\uposElm) + \posElm }{ \uposElm \in \MovRel(\posElm) } =
        \mfElm[\qdrElm](\strElm[{\qdrElm[][\star]}](\posElm)) + \posElm =
        \mfElm[{\qdrElm[][\star]}](\strElm[{\qdrElm[][\star]}](\posElm)) +
        \posElm$, since $\strElm[{\qdrElm[][\star]}](\posElm) \in
        \qsiFun(\qdrElm)$.
        As a consequence, Condition~\ref{def:qsidomrepspc(set:plr)} is
        satisfied.
        If $\posElm \in \PosSet[\OppSym]$, instead, we have that
        $\mfElm[{\qdrElm[][\star]}](\posElm) = \min \set{
        \mfElm[\qdrElm](\uposElm) + \posElm }{ \uposElm \in \MovRel(\posElm) }$,
        which implies $\mfElm[{\qdrElm[][\star]}](\posElm) \leq
        \mfElm[\qdrElm](\uposElm) + \posElm =
        \mfElm[{\qdrElm[][\star]}](\uposElm) + \posElm$, for all adjacents
        $\uposElm \in \MovRel(\posElm)$, as required by
        Condition~\ref{def:qsidomrepspc(set:opp)}.
        To complete the proof that $\prgFun[0]$ is a total function from
        $\QDRSet$ to itself, we need to show that $\qdrElm[][\star]$ satisfies
        Conditions~\ref{def:qsidomrepspc(set:dom)}
        and~\ref{def:qsidomrepspc(set:qsi)} too.
        It is immediate to see that $\denot{\mfElm[\qdrElm]}[\PlrSym] \subseteq
        \denot{\mfElm[{\qdrElm[][\star]}]}[\PlrSym]$.
        Since $\qdrElm$ is a \qdr, $\denot{\mfElm[\qdrElm]}[\PlrSym]$ is a
        $\PlrSym$-dominion.
        Moreover, for all positions $\posElm \in
        \denot{\mfElm[{\qdrElm[][\star]}]}[\PlrSym] \setminus
        \denot{\mfElm[\qdrElm]}[\PlrSym]$, it holds that
        $\strElm[{\qdrElm[][\star]}](\posElm) \in
        \denot{\mfElm[\qdrElm]}[\PlrSym]$, if $\posElm \in \PosSet[\PlrSym]$,
        and $\MovRel(\posElm) \subseteq \denot{\mfElm[\qdrElm]}[\PlrSym]$,
        otherwise.
        Therefore, $\denot{\mfElm[{\qdrElm[][\star]}]}[\PlrSym]$ is necessarily
        a $\PlrSym$-dominion, so Condition~\ref{def:qsidomrepspc(set:dom)} is
        verified.
        Finally, let us focus on Condition~\ref{def:qsidomrepspc(set:qsi)} and
        consider a $(\strElm[{\qdrElm[][\star]}], \posElm)$-play $\posElm
        \pthElm$.
        If, on the one hand, $\pthElm$ is infinite and does not meet $\posElm$,
        thanks to Proposition~\ref{prp:win}, we have $\wghFun(\pthElm) =
        \infty$, thus $\wghFun(\posElm \pthElm) = \infty$ and, so,
        $\wghFun(\posElm \pthElm) > 0$.
        If $\pthElm$ is finite, instead, it holds that $\lst{\pthElm} \in
        \escFun(\qdrElm, \qsiFun(\qdrElm))$ and, so,
        $\mfElm[{\qdrElm[][\star]}](\lst{\pthElm}) = \wghFun(\lst{\pthElm})$,
        due to Proposition~\ref{prp:qdresc}.
        Now, by Proposition~\ref{prp:qdrpth}, we have that
        $\mfElm[\qdrElm](\fst{\pthElm}) \leq \mfElm[\qdrElm](\lst{\pthElm}) +
        \wghFun(\pthElm_{< \ell - 1}) = \wghFun(\lst{\pthElm}) +
        \wghFun(\pthElm_{< \ell - 1}) = \wghFun(\pthElm)$, where $\ell \in
        \SetN$ is the length of $\pthElm$
        Moreover, $0 < \mfElm[\qdrElm](\posElm) \leq
        \mfElm[\qdrElm](\fst{\pthElm}) + \posElm =
        \mfElm[\qdrElm](\fst{\pthElm}) + \wghFun(\posElm)$, thanks to the
        previously proved Conditions~\ref{def:qsidomrepspc(set:plr)}
        and~\ref{def:qsidomrepspc(set:opp)}.
        Hence, $0 < \mfElm[\qdrElm](\posElm) \leq \mfElm[\qdrElm](\fst{\pthElm})
        + \wghFun(\posElm) \leq \wghFun(\posElm) + \wghFun(\pthElm) =
        \wghFun(\posElm \pthElm)$, as required by the definition of quasi
        $\PlrSym$-dominion.
        Finally, if $\pthElm$ is infinite and does meet $\posElm$, it can be
        decomposed as $(\posElm \pthElm')^{\omega}$, where $\pthElm$ is a
        non-empty finite path that does not meet $\posElm$.
        Then, by exploiting the same reasoning done above for the case where
        $\pthElm$ is finite, we have that $\wghFun(\posElm \pthElm') > 0$,
        which implies $\wghFun(\pthElm) = \wghFun((\posElm \pthElm')^{\omega}) =
        \infty$.
      \item
        \textbf{[$\prgFun[+]$].}
        Let $\qdrElm[][\star] \defeq \prgFun[+](\qdrElm)$ and consider the two
        infinite monotone sequences $\QSet[0] \supseteq \QSet[1] \supseteq
        \ldots$ and $\qdrElm[0] \sqsubseteq \qdrElm[1] \sqsubseteq \ldots$
        defined as follows: $\QSet[0] \defeq \dmnFun(\qdrElm)$ and $\qdrElm[0]
        \defeq \qdrElm$; $\QSet[i + 1] \defeq \QSet[i] \setminus \ESet[i]$ and
        $\qdrElm[i + 1] = \liftFun(\qdrElm[i], \ESet[i], \dual{\QSet[i]})$,
        where $\ESet[i] \defeq \bepFun(\qdrElm[i], \QSet[i]) \subseteq
        \escFun(\qdrElm[i], \QSet[i])$, for all $i \in \SetN$.
        Since $\card{\QSet[0]} < \infty$, there necessarily exists an index $k
        \in \SetN$ such that $\QSet[k + 1] = \QSet[k]$, $\qdrElm[k + 1] =
        \qdrElm[k]$.
        Moreover, observe that $\qdrElm[][\star] = \winFun(\qdrElm[k],
        \QSet[k])$.
        We first prove, by induction on the index $i \in \SetN$ of the
        sequences, that every $\qdrElm[i]$ satisfies
        Conditions~\ref{def:qsidomrepspc(set:qsi)}
        and~\ref{def:qsidomrepspc(set:plr)} of
        Definition~\ref{def:qsidomrepspc}.
        Finally, we show that $\qdrElm[][\star]$ is a \qdr.

        The base case $i = 0$ is trivial, since $\qdrElm[i] = \qdrElm$ is a
        \qdr.
        Now, let us consider the inductive case $i > 0$.
        Since the lift operator only modifies the measure of positions belonging
        to $\ESet[i - 1] \subseteq \QSet[i - 1] \subseteq \dmnFun(\qdrElm)
        \subseteq \qsiFun(\qdrElm)$, it immediately follows that
        $\qsiFun(\qdrElm[i]) = \qsiFun(\qdrElm[i - 1]) = \qsiFun(\qdrElm)$.
        Moreover, if $\strElm[{\qdrElm[i]}](\posElm) \neq \strElm[{\qdrElm[i -
        1]}](\posElm)$, we have that $\mfElm[{\qdrElm[i - 1]}](\posElm) <
        \mfElm[{\qdrElm[i]}](\posElm) =
        \mfElm[{\qdrElm[i]}](\strElm[{\qdrElm[i]}](\posElm)) = \mfElm[{\qdrElm[i
        - 1]}](\strElm[{\qdrElm[i]}](\posElm))$, for all positions $\posElm \in
        \qsiFun(\qdrElm[i]) \cap \PosSet[\PlrSym]$, where the latter equality is
        due to the fact that $\strElm[{\qdrElm[i]}](\posElm) \not\in \ESet[i -
        1]$.
        Thus, by Lemma~\ref{lmm:qdrstr}, it holds that $\strElm[{\qdrElm[i]}]$
        is a $\PlrSym$-witness for $\qsiFun(\qdrElm[i])$, \ie,
        Condition~\ref{def:qsidomrepspc(set:qsi)} is verified.
        Also, Condition~\ref{def:qsidomrepspc(set:plr)} directly follows from
        the definition of the $\PlrSym$-strategy inside the lift operator.

        At this point, we can conclude the proof by showing that
        $\qdrElm[][\star]$ is a \qdr.
        Indeed, by Lemma~\ref{lmm:qdrstr}, $\strElm[{\qdrElm[][\star]}]$ is a
        $\PlrSym$-witness for $\qsiFun(\qdrElm[k]) = \qsiFun(\qdrElm)$, so,
        Condition~\ref{def:qsidomrepspc(set:qsi)} is satisfied.
        Similarly to the inductive analysis developed above,
        Condition~\ref{def:qsidomrepspc(set:plr)} directly follows from the
        definition of the $\PlrSym$-strategy inside the $\winFun$ function.
        Moreover, the set $\QSet[k]$ is a closed subset of
        $\qsiFun(\qdrElm[k])$, since $\ESet[k] = \emptyset$ and, so,
        $\escFun(\qdrElm[k], \QSet[k]) = \emptyset$.
        Therefore, $\QSet[k] \subseteq \WinSet[\PlrSym]$, by
        Proposition~\ref{prp:win}.
        In addition, all positions in $\denot{\mfElm[{\qdrElm[k]}]}[\PlrSym]
        \setminus (\denot{\mfElm}[\PlrSym] \cup \QSet[k])$ necessarily reach
        $(\denot{\mfElm}[\PlrSym] \cup \QSet[k]) \subseteq \WinSet[\PlrSym]$.
        As a consequence, Condition~\ref{def:qsidomrepspc(set:dom)} is verified
        as well.

        It remains to prove Condition~\ref{def:qsidomrepspc(set:opp)}.
        To do so, let $\fElm[i] \defeq \min[\posElm \in {\escFun(\qdrElm[i],
        \QSet[i])}] \befFun(\mfElm[{\qdrElm[i]}], \QSet[i], \posElm)$.
        We now first show that the sequence of natural numbers $\fElm[0],
        \fElm[1], \ldots$ is monotone, \ie, $\fElm[i] \leq \fElm[i + 1]$.
        Suppose by contradiction that $\fElm[i] > \fElm[i + 1]$, for some index
        $i \in \SetN$.
        Then, there necessarily exists a position $\posElm \in \escFun(\qdrElm[i
        + 1], \QSet[i + 1]) \setminus \escFun(\qdrElm[i], \QSet[i])$ with
        $\posElm \in \ESet[i + 1]$ such that $\fElm[i + 1] =
        \befFun(\mfElm[{\qdrElm[i + 1]}], \QSet[i + 1],\posElm) < \fElm[i]$.
        We proceeds by a case analysis on the owner of the position $\posElm$.
        \begin{itemize}
          \item
            \textbf{[$\posElm \in \PosSet[\PlrSym]$].}
            By definition of the best-escape forfeit function, we have that
            $\fElm[i + 1] = \max \set{ \mfElm[{\qdrElm[i + 1]}](\uposElm) +
            \posElm - \mfElm[{\qdrElm[i + 1]}](\posElm) }{ \uposElm \in
            \MovRel(\posElm) \setminus \QSet[i + 1] } \geq
            \mfElm[{\qdrElm[i + 1]}](\strElm[{\qdrElm[i]}](\posElm)) + \posElm -
            \mfElm[{\qdrElm[i + 1]}](\posElm)$, since
            $\strElm[{\qdrElm[i]}](\posElm) \in \ESet[i]$ and, so,
            $\strElm[{\qdrElm[i]}](\posElm) \not\in \QSet[i + 1]$.
            Therefore, the following equalities and inequalities hold, which
            lead to the contradiction $\fElm[i] \leq \fElm[i + 1] < \fElm[i]$:
            \begin{linenomath}
            \begin{align*}
              \fElm[i + 1]
              & \geq \mfElm[{\qdrElm[i + 1]}](\strElm[{\qdrElm[i]}](\posElm)) +
              \posElm - \mfElm[{\qdrElm[i + 1]}](\posElm) \\
              & = \mfElm[{\qdrElm[i + 1]}](\strElm[{\qdrElm[i]}](\posElm)) +
              \wghFun(\posElm) - \mfElm[{\qdrElm[i + 1]}](\posElm) \\
              & = \mfElm[{\qdrElm[i]}](\strElm[{\qdrElm[i]}](\posElm)) +
              \fElm[i] + \wghFun(\posElm) - \mfElm[{\qdrElm[i + 1]}](\posElm) \\
              & = \mfElm[{\qdrElm[i]}](\strElm[{\qdrElm[i]}](\posElm)) +
              \fElm[i] + \wghFun(\posElm) - \mfElm[{\qdrElm[i]}](\posElm) \\
              & = \mfElm[{\qdrElm[i]}](\strElm[{\qdrElm[i]}](\posElm)) + \posElm
              - \mfElm[{\qdrElm[i]}](\posElm) + \fElm[i] \\
              & \geq \fElm[i].
            \end{align*}
            \end{linenomath}
            Notice that the first and last equality are due to the definition
            of the measure stretch operator.
            The second one is derived from the fact that
            $\strElm[{\qdrElm[i]}](\posElm) \in \ESet[i]$, while the third one
            from $\posElm \in \ESet[i + 1]$, which implies $\mfElm[{\qdrElm[i +
            1]}](\posElm) = \mfElm[{\qdrElm[i]}](\posElm)$.
            Finally, the last inequality follows from
            Condition~\ref{def:qsidomrepspc(set:plr)} applied to $\qdrElm[i]$,
            \ie, $\mfElm[{\qdrElm[i]}](\posElm) \leq
            \mfElm[{\qdrElm[i]}](\strElm[{\qdrElm[i]}](\posElm)) + \posElm$.
          \item
            \textbf{[$\posElm \in \PosSet[\OppSym]$].}
            Again by definition of the best-escape forfeit function, we have
            that
            $\fElm[i + 1] = \min \set{ \mfElm[{\qdrElm[i + 1]}](\uposElm) +
            \posElm - \mfElm[{\qdrElm[i + 1]}](\posElm) }{ \uposElm \in
            \MovRel(\posElm) \setminus \QSet[i + 1] }$.
            In addition, $\MovRel(\posElm) \setminus \QSet[i + 1] \subseteq
            \ESet[i]$
            Therefore, the following equalities hold:
            \begin{linenomath}
            \begin{align*}
              \fElm[i + 1]
              & = \min {\set{ \mfElm[{\qdrElm[i + 1]}](\uposElm) + \posElm -
              \mfElm[{\qdrElm[i + 1]}](\posElm) }{ \uposElm \in \MovRel(\posElm)
              \setminus \QSet[i + 1] }} \\
              & = \min {\set{ \mfElm[{\qdrElm[i + 1]}](\uposElm) +
              \wghFun(\posElm) - \mfElm[{\qdrElm[i + 1]}](\posElm) }{ \uposElm
              \in \MovRel(\posElm) \setminus \QSet[i + 1] }} \\
              & = \min {\set{ \mfElm[{\qdrElm[i]}](\uposElm) + \fElm[i] +
              \wghFun(\posElm) - \mfElm[{\qdrElm[i + 1]}](\posElm) }{ \uposElm
              \in \MovRel(\posElm) \setminus \QSet[i + 1] }} \\
              & = \min {\set{ \mfElm[{\qdrElm[i]}](\uposElm) + \fElm[i] +
              \wghFun(\posElm) - \mfElm[{\qdrElm[i]}](\posElm) }{ \uposElm
              \in \MovRel(\posElm) \setminus \QSet[i + 1] }} \\
              & = \min {\set{ \mfElm[{\qdrElm[i]}](\uposElm) + \posElm -
              \mfElm[{\qdrElm[i]}](\posElm) + \fElm[i] }{ \uposElm \in
              \MovRel(\posElm) \setminus \QSet[i + 1] }} \\
              & \geq \fElm[i].
            \end{align*}
            \end{linenomath}
            Notice that the second and last equality are due to the definition
            of the measure stretch operator.
            The third one is derived from the fact that $\uposElm \in
            \MovRel(\posElm) \setminus \QSet[i + 1] \subseteq \ESet[i]$, while
            the fourth one from $\posElm \in \ESet[i + 1]$, which implies
            $\mfElm[{\qdrElm[i + 1]}](\posElm) = \mfElm[{\qdrElm[i]}](\posElm)$.
            Finally, the last inequality follows from
            Condition~\ref{def:qsidomrepspc(set:opp)} applied to $\qdrElm$, \ie,
            $\mfElm[{\qdrElm[i]}](\posElm) = \mfElm[\qdrElm](\posElm) \leq
            \mfElm[\qdrElm](\uposElm) + \posElm \leq
            \mfElm[{\qdrElm[i]}](\uposElm) + \posElm$, for all adjacents
            $\uposElm \in \MovRel(\posElm)$.
        \end{itemize}
        Now suppose by contradiction that
        Condition~\ref{def:qsidomrepspc(set:opp)} does not hold for
        $\qdrElm[][\star]$.
        Then, there exist a $\OppSym$-position $\posElm \in
        \qsiFun(\qdrElm[][\star]) \cap \PosSet[\OppSym]$ and one of its
        adjacents $\uposElm \in \MovRel(\posElm)$ such that
        $\mfElm[{\qdrElm[][\star]}](\uposElm) + \posElm <
        \mfElm[{\qdrElm[][\star]}](\posElm)$.
        Due to the process used to compute $\qdrElm[][\star]$, there are indexes
        $i, j \in \numcc{0}{k}$ such that $\mfElm[{\qdrElm[][\star]}](\uposElm)
        = \mfElm[{\qdrElm[i + 1]}](\uposElm) = \mfElm[\qdrElm](\uposElm) +
        \fElm[i]$ and $\mfElm[{\qdrElm[][\star]}](\posElm) = \mfElm[{\qdrElm[j +
        1]}](\posElm) = \mfElm[\qdrElm](\posElm) + \fElm[j]$.
        Now, by Condition~\ref{def:qsidomrepspc(set:opp)} applied to $\qdrElm$,
        we have $\mfElm[\qdrElm](\posElm) \leq \mfElm[\qdrElm](\uposElm) +
        \posElm$, which implies that $0 \leq \mfElm[\qdrElm](\uposElm) +
        \posElm - \mfElm[\qdrElm](\posElm) < \fElm[j] - \fElm[i]$ and,
        consequently, both $i < j$ and $\uposElm \not\in \QSet[j]$.
        However,
        \begin{linenomath}
        \begin{align*}
          \fElm[j] - \fElm[i]
          & = \min {\set{ \mfElm[{\qdrElm[j]}](\zposElm) + \posElm -
          \mfElm[{\qdrElm[j]}](\posElm) }{ \zposElm \in \MovRel(\posElm)
          \setminus \QSet[j] }} - \fElm[i] \\
          & \leq \mfElm[{\qdrElm[j]}](\uposElm) + \posElm -
          \mfElm[{\qdrElm[j]}](\posElm) - \fElm[i] \\
          & = \mfElm[{\qdrElm[j]}](\uposElm) + \posElm -
          \mfElm[\qdrElm](\posElm) - \fElm[i] \\
          & = \mfElm[{\qdrElm[i + 1]}](\uposElm) + \posElm -
          \mfElm[\qdrElm](\posElm) - \fElm[i] \\
          & = (\mfElm[\qdrElm](\uposElm) + \fElm[i]) + \posElm -
          \mfElm[\qdrElm](\posElm) - \fElm[i] \\
          & = \mfElm[\qdrElm](\uposElm) + \posElm -
          \mfElm[\qdrElm](\posElm),
        \end{align*}
        \end{linenomath}
        leading to the contradiction $\mfElm[\qdrElm](\uposElm) + \posElm -
        \mfElm[\qdrElm](\posElm) < \fElm[j] - \fElm[i] \leq
        \mfElm[\qdrElm](\uposElm) + \posElm - \mfElm[\qdrElm](\posElm)$.
        Notice that the first equality is due to the definition of the
        best-escape forfeit function.
        The second and third ones, instead, follows from the fact that $\posElm$
        and $\uposElm$ changed their values at iterations $j + 1$ and $i + 1$,
        respectively.
        Finally, the fourth equality derives from the operation of lift and
        best-escape forfeit computed on $\uposElm$.
    \end{itemize}
  \end{proof}

  \begin{lemma}
    \label{lmm:qdrbnd}
    Let $\qdrElm[][\star] \defeq \prgFun[+](\qdrElm)$, for some $\qdrElm \in
    \QDRSet$, and $\SElm \defeq \sum \set{ \wghFun(\posElm) \in \SetN }{ \posElm
    \in \PosSet \land \wghFun(\posElm) > 0 }$.
    Then, for all positions $\posElm \in \qsiFun(\qdrElm[][\star])$ with
    $\qdrElm[][\star](\posElm) \neq \infty$, it holds that
    $\qdrElm[][\star](\posElm) \leq \SElm$.
  \end{lemma}
  \begin{proof}
    Suppose by contradiction that there exists a position $\posElm \in
    \qsiFun(\qdrElm[][\star])$ with $\qdrElm[][\star](\posElm) \neq \infty$, but
    $\qdrElm[][\star](\posElm) > \SElm$.
    If $\posElm \in \qsiFun(\qdrElm) \setminus \dmnFun(\qdrElm)$, then
    $\qdrElm[][\star](\posElm) = \qdrElm(\posElm)$.
    Moreover, there exists a finite path $\pthElm$ compatible with the
    $\PlrSym$-strategy $\strElm[\qdrElm]$ and entirely contained in
    $\qsiFun(\qdrElm) \setminus \dmnFun(\qdrElm)$, which starts in $\posElm$ and
    ends in $\escFun(\qdrElm, \qsiFun(\qdrElm))$, \ie, $\fst{\pthElm} = \posElm$
    and $\lst{\pthElm} \in \escFun(\qdrElm, \qsiFun(\qdrElm))$.
    By Propositions~\ref{prp:qdrpth} and~\ref{prp:qdresc}, we have that
    $\SElm < \mfElm[\qdrElm](\posElm) \leq \wghFun(\pthElm) \leq
    \SElm[][\star]$, with $\SElm[][\star] \defeq \sum \set{ \wghFun(\posElm)
    \in \SetN }{ \posElm \in \qsiFun(\qdrElm) \setminus \dmnFun(\qdrElm) \land
    \wghFun(\posElm) > 0 }$, where the last inequality is obviously due to the
    fact that there are no repeated positions in $\pthElm$, being it finite.
    However, $\SElm[][\star] \leq \SElm$, which means that a contradiction has
    been reached with $\posElm \in \qsiFun(\qdrElm) \setminus \dmnFun(\qdrElm)$.
    Thus, assume $\posElm \in \dmnFun(\qdrElm)$ and consider the two infinite
    monotone sequences $\QSet[0] \supseteq \QSet[1] \supseteq \ldots$ and
    $\qdrElm[0] \sqsubseteq \qdrElm[1] \sqsubseteq \ldots$ defined as in the
    proof of Theorem~\ref{thm:tot}: $\QSet[0] \defeq \dmnFun(\qdrElm)$ and
    $\qdrElm[0] \defeq \qdrElm$; $\QSet[i + 1] \defeq \QSet[i] \setminus
    \ESet[i]$ and $\qdrElm[i + 1] = \liftFun(\qdrElm[i], \ESet[i],
    \dual{\QSet[i]})$, where $\ESet[i] \defeq \bepFun(\qdrElm[i], \QSet[i])
    \subseteq \escFun(\qdrElm[i], \QSet[i])$, for all $i \in \SetN$.
    Also, let $\SElm[0] \leq \SElm[1] \leq \ldots < \top$ be the sequence of
    natural numbers defined as $\SElm[i] \defeq \sum \set{ \wghFun(\posElm) \in
    \SetN }{ \posElm \in \qsiFun(\qdrElm) \setminus \QSet[i] \land
    \wghFun(\posElm) > 0 }$.
    Since $\qdrElm[][\star](\posElm) \neq \infty$, there exists an index $k$
    such that $\posElm \in \ESet[k]$ with $\befFun(\qdrElm[k], \QSet[k],
    \posElm) < \infty$.
    Therefore, to prove the thesis, it suffices to show that
    $\mfElm[{\qdrElm[i]}](\zposElm) \leq \SElm[i]$, for all positions $\zposElm
    \in \QSet[i]$ and index $i \in \numcc{0}{k}$.
    The base case $i = 0$ follows by applying the same reasoning previously done
    for the case $\posElm \in \qsiFun(\qdrElm) \setminus \dmnFun(\qdrElm)$ and
    by noticing that $\SElm[0] = \SElm[][\star]$.
    Now, let $i > 0$.
    By definition of the lift operator, there exists at least one adjacent
    $\xposElm$ of $\zposElm$ such that $\mfElm[{\qdrElm[i + 1]}](\zposElm) =
    \mfElm[{\qdrElm[i]}](\xposElm) + \zposElm$.
    By the inductive hypothesis, $\mfElm[{\qdrElm[i]}](\xposElm) \leq \SElm[i]$.
    Thus, $\mfElm[{\qdrElm[i + 1]}](\zposElm) \leq \SElm[i] + \wghFun(\zposElm)
    \leq \SElm[i + 1]$, since $\wghFun(\zposElm) \not\in \SElm[i]$.
  \end{proof}

  \renewcounter{theorem}{thm:ter}
  \begin{theorem}[Totality]
    The solver operator $\solFun \defeq \ifpFun\, \qdrElm \,.\,
    \prgFun[+](\prgFun[0](\qdrElm))$ is a well-defined total function.
    Moreover, for every $\qdrElm \in \QDRSet$ it holds that $\solFun(\qdrElm) =
    (\ifpFun[k]\, \qdrElm[][\star] \,.\,
    \prgFun[+](\prgFun[0](\qdrElm[][\star])))(\qdrElm)$, for some index $k \leq
    n \cdot (\SElm + 1)$, where $n$ is the number of positions in the \MPG and
    $\SElm \defeq \sum \set{ \wghFun(\posElm) \in \SetN }{ \posElm \in \PosSet
    \land \wghFun(\posElm) > 0 }$ the total sum of its positive weights.
  \end{theorem}
  \begin{proof}
    Consider the sequence $\qdrElm[0], \qdrElm[1], \ldots$ recursively defined
    as follows: $\qdrElm[0] \defeq (\ifpFun[0]\, \qdrElm[][\star] \,.\,
    \prgFun[+](\prgFun[0](\qdrElm[][\star])))(\qdrElm) = \qdrElm$ and $\qdrElm[i
    + 1] \defeq (\ifpFun[i + 1]\, \qdrElm[][\star] \,.\,
    \prgFun[+](\prgFun[0](\qdrElm[][\star])))(\qdrElm) =
    \prgFun[+](\prgFun[0](\qdrElm[i]))$, for all $i \in \SetN$.
    By induction on the index $i$, thanks to the totality and inflationary
    properties of the progress operators $\prgFun[0]$ and $\prgFun[+]$
    previously proved in Theorem~\ref{thm:tot}, one can easily show that every
    $\qdrElm[i]$ is a \qdr satisfying $\qdrElm[i] \sqsubseteq \qdrElm[i + 1]$.
    Moreover, by Lemma~\ref{lmm:qdrbnd}, we have that $\qdrElm[i](\posElm) \leq
    \SElm$, for all positions $\posElm \in \qsiFun(\qdrElm[i])$ with
    $\qdrElm[i](\posElm) \neq \infty$ and index $i > 0$.
    Now, there are at most $n \cdot (\SElm + 1)$ such {\qdr}s, thus, there
    necessarily exists an index $k \leq n \cdot (\SElm + 1)$ such that
    $\qdrElm[k + 1] = \qdrElm[k]$, which implies $\solFun(\qdrElm) = (\ifpFun\,
    \qdrElm[][\star] \,.\, \liftFun(\qdrElm[][\star]))(\qdrElm) = \qdrElm[k]$.
    Hence, the thesis immediately follows.
  \end{proof}

  \begin{lemma}
    \label{lmm:fixpnt}
    Let $\qdrElm[][\star] \defeq \solFun(\qdrElm)$ be the result of the solver
    operator applied to an arbitrary $\qdrElm \in \QDRSet$.
    Then, $\qdrElm[][\star]$ is a fixpoint of the progress operators, \ie,
    $\qdrElm[][\star] = \prgFun[0](\qdrElm[][\star]) =
    \prgFun[+](\qdrElm[][\star])$.
  \end{lemma}
  \begin{proof}
    By definition of inflationary fixpoint, $\qdrElm[][\star]$ is a fixpoint of
    the composition of the two progress operators, \ie, $\qdrElm[][\star] =
    \prgFun[+](\prgFun[0](\qdrElm[][\star]))$, which are inflationary functions,
    due to Theorem~\ref{thm:tot}.
    As a consequence, we have that $\qdrElm[][\star] =
    \prgFun[+](\prgFun[0](\qdrElm[][\star])) \sqsupseteq
    \prgFun[0](\qdrElm[][\star]) \sqsupseteq \qdrElm[][\star]$.
    Thus, $\prgFun[0](\qdrElm[][\star]) = \qdrElm[][\star]$ and, so,
    $\prgFun[+](\qdrElm[][\star]) = \qdrElm[][\star]$.
  \end{proof}

  \renewcounter{lemma}{lmm:prgzer}
  \begin{lemma}
    Let $\qdrElm \in \QDRSet$ be a fixpoint of $\prgFun[0]$.
    Then, $\mfElm[\qdrElm]$ is a progress measure over
    $\dual{\qsiFun(\qdrElm)}$.
  \end{lemma}
  \begin{proof}
    By definition of the progress operator $\prgFun[0]$, we have that $\qdrElm =
    \prgFun[0](\qdrElm) = \sup \{ \qdrElm, \liftFun(\qdrElm,
    \dual{\qsiFun(\qdrElm)}, \PosSet) \}$, from which we derive
    $\qdrElm[][\star] \defeq \liftFun(\qdrElm, \dual{\qsiFun(\qdrElm)}, \PosSet)
    \sqsubseteq \qdrElm$.
    Now, consider an arbitrary position $\posElm \in \dual{\qsiFun(\qdrElm)}$
    and observe that $\mfElm[{\qdrElm[][\star]}](\posElm) \leq
    \mfElm[\qdrElm](\posElm)$, due to Item~\ref{def:qsidomrepspc(ord)} of
    Definition~\ref{def:qsidomrepspc}.
    At this point, the proof proceeds by a case analysis on the owner of the
    position $\posElm$ itself.
    \begin{itemize}
      \item
        \textbf{[$\posElm \in \PosSet[\PlrSym]$].}
        By definition of the lift operator, we have that
        $\mfElm[\qdrElm](\uposElm) + \posElm \leq \max \set{
        \mfElm[\qdrElm](\uposElm) + \posElm }{ \uposElm \in \MovRel(\posElm) } =
        \mfElm[{\qdrElm[][\star]}](\posElm)$, for all adjacents $\uposElm \in
        \MovRel(\posElm)$ of $\posElm$.
        Thus, $\mfElm[\qdrElm](\uposElm) + \posElm \leq
        \mfElm[{\qdrElm[][\star]}](\posElm) \leq \mfElm[\qdrElm](\posElm)$,
        thanks to the above observation.
        Consequently, Condition~\ref{def:prgmsr(plr)} of
        Definition~\ref{def:prgmsr} is satisfied on $\dual{\qsiFun(\qdrElm)}$.
      \item
        \textbf{[$\posElm \in \PosSet[\OppSym]$].}
        Again by definition of the lift operator, we have that
        $\mfElm[\qdrElm](\uposElm) + \posElm \leq \min \set{
        \mfElm[\qdrElm](\uposElm) + \posElm }{ \uposElm \in \MovRel(\posElm) } =
        \mfElm[{\qdrElm[][\star]}](\posElm)$, for some adjacent $\uposElm \in
        \MovRel(\posElm)$ of $\posElm$.
        Due to the above observation, it holds that $\mfElm[\qdrElm](\uposElm) +
        \posElm \leq \mfElm[{\qdrElm[][\star]}](\posElm) \leq
        \mfElm[\qdrElm](\posElm)$.
        Hence, Condition~\ref{def:prgmsr(opp)} of Definition~\ref{def:prgmsr} is
        satisfied on $\dual{\qsiFun(\qdrElm)}$ as well.
    \end{itemize}
  \end{proof}

  \renewcounter{lemma}{lmm:prgpls}
  \begin{lemma}
    Let $\qdrElm \in \QDRSet$ be a fixpoint of $\prgFun[+]$.
    Then, $\mfElm[\qdrElm]$ is a progress measure over $\qsiFun(\qdrElm)$.
  \end{lemma}
  \begin{proof}
    Let us consider the infinite monotone sequence of position sets $\QSet[0]
    \supseteq \QSet[1] \supseteq \ldots$ defined as follows: $\QSet[0] \defeq
    \dmnFun(\qdrElm)$; $\QSet[i + 1] \defeq \QSet[i] \setminus \ESet[i]$, where
    $\ESet[i] \defeq \bepFun(\qdrElm, \QSet[i])$, for all $i \in \SetN$.
    Since $\card{\QSet[0]} < \infty$, there necessarily exists an index $k \in
    \SetN$ such that $\QSet[k + 1] = \QSet[k]$.
    By definition of the progress operator $\prgFun[+]$ and the equality
    $\qdrElm = \prgFun[+](\qdrElm)$, we have that $\qdrElm = \liftFun(\qdrElm,
    \ESet[i], \dual{\QSet[i]})$, for all $i \in \numco{0}{k}$, and $\qdrElm =
    \winFun(\qdrElm, \QSet[k])$.
    Now, consider an arbitrary position $\posElm \in \qsiFun(\qdrElm)$.
    If $\posElm \not\in \dmnFun(\qdrElm)$, due to the definition of the set
    $\dmnFun(\qdrElm)$, the position $\posElm$ satisfies by definition of the
    appropriate condition of Definition~\ref{def:prgmsr} on $\qsiFun(\qdrElm)$.
    Therefore, let us assume $\posElm \in \dmnFun(\qdrElm)$.
    Then, it is obvious that either $\posElm \in \QSet[k]$ or there is a unique
    index $i \in \numco{0}{k}$ such that $\posElm \in \QSet[i] \setminus \QSet[i
    + 1]$, \ie, $\posElm \in \ESet[i]$.
    In the first case, we have $\mfElm[\qdrElm](\posElm) = \infty$, due to the
    definition of the function $\winFun$.
    Therefore, $\posElm$ is a progress position.
    In the other case, the proof proceeds by a case analysis on the owner of the
    position $\posElm$ itself.
    \begin{itemize}
      \item
        \textbf{[$\posElm \in \PosSet[\PlrSym]$].}
        First observe that $\bepFun(\qdrElm, \QSet[i]) \subseteq
        \escFun(\qdrElm, \QSet[i])$.
        Thus, due to the definition of the function $\escFun$, we have that
        $\mfElm[\qdrElm](\uposElm) + \posElm \leq \mfElm[\qdrElm](\posElm)$, for
        all positions $\uposElm \in \MovRel(\posElm) \cap \QSet[i]$.
        Now, by the definition of the lift operator, we have that
        $\mfElm[\qdrElm](\uposElm) + \posElm \leq \max \set{
        \mfElm[\qdrElm](\uposElm) + \posElm }{ \uposElm \in \MovRel(\posElm)
        \cap \dual{\QSet[i]} } = \mfElm[\qdrElm](\posElm)$, for all adjacents
        $\uposElm \in \MovRel(\posElm)\cap \dual{\QSet[i]}$ of $\posElm$.
        Consequently, $\mfElm[\qdrElm](\uposElm) + \posElm \leq
        \mfElm[\qdrElm](\posElm)$, for all positions $\uposElm \in
        \MovRel(\posElm)$, as required by Condition~\ref{def:prgmsr(plr)} of
        Definition~\ref{def:prgmsr} on $\qsiFun(\qdrElm)$.
      \item
        \textbf{[$\posElm \in \PosSet[\OppSym]$].}
        Again by definition of the lift operator, we have that
        $\mfElm[\qdrElm](\uposElm) + \posElm \leq \min \set{
        \mfElm[\qdrElm](\uposElm) + \posElm }{ \uposElm \in \MovRel(\posElm)
        \cap \dual{\QSet[i]} } = \mfElm[\qdrElm](\posElm)$, for some adjacent
        $\uposElm \in \MovRel(\posElm) \cap \dual{\QSet[i]} \subseteq
        \MovRel(\posElm)$ of $\posElm$.
        Hence, Condition~\ref{def:prgmsr(opp)} of Definition~\ref{def:prgmsr} is
        satisfied on $\qsiFun(\qdrElm)$ as well.
    \end{itemize}
  \end{proof}

  \renewcounter{theorem}{thm:snd}
  \begin{theorem}[Soundness]
    $\denot{\solFun(\qdrElm)}[\OppSym] \subseteq \WinSet[\OppSym]$, for every
    $\qdrElm \in \QDRSet$.
  \end{theorem}
  \begin{proof}
    Let $\qdrElm[][\star] \defeq \solFun(\qdrElm)$ be the result of the solver
    operator applied to $\qdrElm \in \QDRSet$.
    By Lemma~\ref{lmm:fixpnt}, it holds that $\qdrElm[][\star] =
    \prgFun[0](\qdrElm[][\star]) = \prgFun[+](\qdrElm[][\star])$.
    As a consequence, $\mfElm[{\qdrElm[][\star]}]$ is a progress measure, due to
    Lemmas~\ref{lmm:prgzer} and~\ref{lmm:prgpls}.
    At this point, by recalling that $\denot{\qdrElm[][\star]}[\OppSym] =
    \denot{\mfElm[{\qdrElm[][\star]}]}[\OppSym]$, as reported in
    Definition~\ref{def:qsidomrepspc}, the thesis is immediately derived by
    applying Theorem~\ref{thm:prgmsr} to $\mfElm[{\qdrElm[][\star]}]$.
  \end{proof}

  \renewcounter{theorem}{thm:com}
  \begin{theorem}[Completeness]
    $\denot{\solFun(\qdrElm)}[\PlrSym] \subseteq \WinSet[\PlrSym]$, for every
    $\qdrElm \in \QDRSet$.
  \end{theorem}
  \begin{proof}
    The thesis immediately follows by considering Theorem~\ref{thm:ter} and
    Condition~\ref{def:qsidomrepspc(set:dom)} of
    Definition~\ref{def:qsidomrepspc}.
    Indeed, by the statement of the recalled theorem, $\solFun(\qdrElm)$ is a
    \qdr, independently of the element $\qdrElm \in \QDRSet$ given as input to
    the solver operator.
    Thus, thanks to the above condition, it holds the
    $\denot{\solFun(\qdrElm)}[\PlrSym] \subseteq \WinSet[\PlrSym]$.
  \end{proof}

  \renewcounter{lemma}{lmm:qdrchg}
  \begin{lemma}
    Let $\qdrElm[][\star] \defeq \prgFun[+](\qdrElm)$, for some $\qdrElm \in
    \QDRSet$.
    Then, $\mfElm[{\qdrElm[][\star]}](\posElm) > \mfElm[\qdrElm](\posElm)$, for
    all positions $\posElm \in \dmnFun(\qdrElm)$.
  \end{lemma}
  \begin{proof}
    Consider the set $\ESet \defeq \bepFun(\qdrElm, \dmnFun(\qdrElm)) \subseteq
    \escFun(\qdrElm, \dmnFun(\qdrElm))$ and let $\der{\qdrElm} \defeq
    \liftFun(\qdrElm, \ESet, \dual{\dmnFun(\qdrElm)})$.
    First observe that $\mfElm[\der{\qdrElm}](\posElm) =
    \mfElm[{\qdrElm[][\star]}](\posElm)$, for all escape positions $\posElm \in
    \ESet$.
    We now show that $\mfElm[{\qdrElm[][\star]}](\posElm) >
    \mfElm[\qdrElm](\posElm)$, via a case analysis on the owner of the position
    $\posElm$ itself.
    \begin{itemize}
      \item
        \textbf{[$\posElm \in \PosSet[\PlrSym]$].}
        By definition of the function $\escFun$, it holds that
        $\strElm[\qdrElm](\posElm) \not\in \dmnFun(\qdrElm)$ and
        $\mfElm[\qdrElm](\posElm) \geq \mfElm[\qdrElm](\uposElm) +
        \posElm$, for all adjacents $\uposElm \in \MovRel(\posElm) \cap
        \dmnFun(\qdrElm)$.
        Since $\posElm \in \dmnFun(\qdrElm)$, due to the way this specific weak
        quasi dominion is constructed, $\posElm \in \nppFun(\qdrElm)$.
        Thus, there exists a successor $\uposElm[][\star] \in \MovRel(\posElm)$
        with $\mfElm[\qdrElm](\posElm) < \mfElm[\qdrElm](\uposElm[][\star]) +
        \posElm$, from which it follows that $\uposElm[][\star] \not\in
        \dmnFun(\qdrElm)$, \ie, $\uposElm[][\star] \in \dual{\dmnFun(\qdrElm)}$.
        As a consequence, we obtain that $\mfElm[\der{\qdrElm}](\posElm) \geq
        \mfElm[\qdrElm](\uposElm[][\star]) + \posElm >
        \mfElm[\qdrElm](\posElm)$.
        Hence, $\mfElm[{\qdrElm[][\star]}](\posElm) > \mfElm[\qdrElm](\posElm)$.
      \item
        \textbf{[$\posElm \in \PosSet[\OppSym]$].}
        Since $\posElm \in \dmnFun(\qdrElm)$, we have that
        $\mfElm[\qdrElm](\posElm) < \mfElm[\qdrElm](\uposElm) + \posElm$, for
        all adjacents $\uposElm \in \MovRel(\posElm) \setminus
        \dmnFun(\qdrElm)$.
        Thus, $\mfElm[\der{\qdrElm}](\posElm) = \min \set{
        \mfElm[\qdrElm](\uposElm) + \posElm }{ \uposElm \in \MovRel(\posElm)
        \setminus \dmnFun(\qdrElm) } > \mfElm[\qdrElm](\posElm)$.
        Hence, $\mfElm[{\qdrElm[][\star]}](\posElm) > \mfElm[\qdrElm](\posElm)$
        in this case as well.
    \end{itemize}
    Now, consider a position $\posElm \in \dmnFun(\qdrElm) \setminus \ESet$.
    Obviously, $\mfElm[\qdrElm](\posElm) < \infty$.
    If $\mfElm[{\qdrElm[][\star]}](\posElm) = \infty$, the thesis immediately
    follows.
    Otherwise, it will be considered as an escape of some weak quasi dominion
    $\QSet \subset \dmnFun(\qdrElm)$, after the removal of the first escape
    positions in $\ESet$.
    Due to the non-decreasing property of the sequence of best-escape forfeit
    shown in the proof of Theorem~\ref{thm:tot}, $\posElm$ exits from $\QSet$
    with a forfeit $\fElm[][\star]$ at least as great as the one $\fElm$ of
    $\ESet$ that we just proved to be strictly positive.
    Indeed, $\fElm = \mfElm[{\qdrElm[][\star]}](\zposElm) -
    \mfElm[\qdrElm](\zposElm) > 0$, for all $\zposElm \in \ESet$.
    Therefore, $\mfElm[{\qdrElm[][\star]}](\posElm) - \mfElm[\qdrElm](\posElm) =
    \fElm[][\star] \geq \fElm > 0$, which implies
    $\mfElm[{\qdrElm[][\star]}](\posElm) > \mfElm[\qdrElm](\posElm)$.
  \end{proof}

  \renewcounter{theorem}{thm:cmp}
  \begin{theorem}[Complexity]
    \QDPM requires time $\AOmicron{n \cdot m \cdot \WElm \cdot \log(n \cdot
    \WElm)}$ to solve an \MPG with $n$ positions, $m$ moves, and maximal
    positive weight $\WElm$.
  \end{theorem}
  \begin{proof}
    To compute $\solFun$ efficiently, we now provide an imperative reformulation
    of the functional fixpoint algorithm $\solFun \defeq \ifpFun\, \qdrElm \,.\,
    \prgFun[+](\prgFun[0](\qdrElm))$ with the desired complexity.
    Recall that, by Lemma~\ref{lmm:qdrbnd}, each position can only be lifted
    $\SElm + 1$ times, where $\SElm \defeq \sum \set{ \wghFun(\posElm) \in \SetN
    }{ \posElm \in \PosSet \land \wghFun(\posElm) > 0 } = \AOmicron{n \cdot
    \WElm }$.
    Therefore, to obtain the claimed complexity, we have to guarantee that the
    cost of all the computational steps be linear in the number of measure
    increases.
    To do so, it suffices to ensure that the algorithm explores the incoming and
    outgoing moves only of those positions whose measures are actually lifted.
    This is clearly the case for the lift operator itself, since it only
    explores the outgoing moves of each position in its source set.
    The only remaining problem is to be able to identify the positions that need
    to be lifted in the next iteration, by only exploring the incoming moves of
    the positions just lifted.
    Solving this problem requires some technical tricks.
    Specifically, inspired by~\cite{BCDGR11}, will employ vectors of counters,
    namely $\cFun$, $\dFun$ and $\gFun$, that associates with
    $\PlrSym$-positions the number of moves that do not satisfy the progress
    condition, and with $\OppSym$-positions the number of moves that satisfy
    it.
    In addition, we will also use a priority queue $\TSet$ to allow an efficient
    identification of the \emph{best-escape positions} during the computation of
    the operator $\prgFun[+]$.

    \algsol

      Algorithm~\ref{alg:sol} reports the procedural implementation of
      $\solFun(\qdrElm[0])$, where $\qdrElm[0]$ is the smallest possible \qdr,
      as defined at Line~1.
      At the beginning of each iteration $i \in \SetN$ of the while-loop at
      Line~4, the variable $\qdrElm$ maintains the \qdr $\qdrElm[i]$ computed
      by applying to $\qdrElm[0]$ the composition $\prgFun[+] \cmp \prgFun[0]$
      $i$ times.
      Moreover, the sets $\NSet[0]$ and $\NSet[+]$ contain, respectively, the
      positions that need to be lifted by $\prgFun[0]$ and the non-progress
      positions in $\qdrElm[i]$.
      The formal invariants at Line~4 are: $\NSet[0] = \set{ \posElm \in \PosSet
      }{ \mfElm[{\qdrElm[i]}](\posElm) = 0 \neq \mfElm[{\qdrElm[i +
      1]}](\posElm) }$ and $\NSet[+] = \nppFun(\qdrElm[i])$.
      Observe that these invariants are trivially satisfied for $i = 0$, thanks
      to Line~3.
      After the execution of the progress procedure $\prgFun[0]$ at Line~5,
      whose code is reported in Algorithm~\ref{alg:prgzer}, we have that
      $\NSet[0] \subseteq \set{ \posElm \in \PosSet }{ \mfElm[{\qdrElm[i +
      1]}](\posElm) = 0 \neq \mfElm[{\qdrElm[i + 2]}](\posElm) }$ and $\NSet[+]
      \cup \ASet = \nppFun(\qdrElm[i][\star])$, where $\qdrElm[i][\star] \defeq
      \prgFun[0](\qdrElm[i])$.
      Thus, Line~6 ensures that $\NSet[+] = \nppFun(\qdrElm[i][\star])$.
      Line~7 calls the progress procedure $\prgFun[+]$, which is reported in
      Algorithm~\ref{alg:prgpls}, and forces the lift of the measures of all
      the positions in $\dmnFun(\qdrElm[i][\star])$, as stated by
      Lemma~\ref{lmm:qdrchg}.
      In addition, the verified invariants are $\NSet[0] \cup \ASet = \set{
      \posElm \in \PosSet }{ \mfElm[{\qdrElm[i + 1]}](\posElm) = 0 \neq
      \mfElm[{\qdrElm[i + 2]}](\posElm) }$ and $\NSet[+] = \nppFun(\qdrElm[i +
      1])$.
      Finally, after Line~8, it holds that $\NSet[0] = \set{ \posElm \in \PosSet
      }{ \mfElm[{\qdrElm[i + 1]}](\posElm) = 0 \neq \mfElm[{\qdrElm[i +
      2]}](\posElm) }$, as required by the previously discussed invariants for
      the next iteration $i + 1$.
      Observe that Line~2 is used to initialize, for each $\OppSym$-position
      $\posElm \in \PosSet[\OppSym]$, the counter $\cFun(\posElm)$ to the number
      of adjacents $\uposElm \in \MovRel(\posElm)$ of $\posElm$ that satisfy the
      progress inequality $\mfElm[{\qdrElm[0]}](\posElm) \geq
      \mfElm[{\qdrElm[0]}](\uposElm) + \posElm$.

    The subsequent analysis of Algorithms~\ref{alg:prgzer} and~\ref{alg:prgpls}
    shows that the procedures $\prgFun[0](\qdrElm, \cFun, \NSet[0])$ and
    $\prgFun[+](\qdrElm, \cFun, \NSet[+])$ require time
    \[
      \AOmicron{\sum_{\posElm \in \NSet[0]} (\card{\MovRel(\posElm)} +
      \card{\MovRel[][-1](\posElm)}) \cdot \log \SElm}
      \text{ and }
      \AOmicron{\sum_{\posElm \in \dmnFun(\qdrElm)} (\card{\MovRel(\posElm)} +
      \card{\MovRel[][-1](\posElm)}) \cdot \log \SElm},
    \]
    respectively, where $\nppFun(\qdrElm) = \NSet[+]$.
    In particular, the factor $\log \SElm$ is due to all the arithmetic
    operations required to compute the stretch of the measures.
    Since during the entire execution of the algorithm each position $\posElm
    \in \PosSet$ can appear at most once in some $\NSet[0]$ and at most $\SElm$
    times in some $\dmnFun(\qdrElm)$, it follows that the total cost of
    Algorithm~\ref{alg:sol} is
    \[
      \AOmicron{n + (\SElm + 1) \cdot\! \sum_{\posElm \in \PosSet}
      (\card{\MovRel(\posElm)} \!+\! \card{\MovRel[][-1](\posElm)}) \cdot \log
      \SElm} = \AOmicron{n + \SElm \cdot m \cdot \log \SElm} = \AOmicron{n \cdot
      m \cdot \WElm \cdot \log(n \cdot \WElm)},
    \]
    where the term $n$ is due to the initialization operations at Lines~1-3.

    Observe that, the two procedures $\prgFun[0]$ and $\prgFun[+]$, together
    with the auxiliary one reported in Algorithm~\ref{alg:dmn}, share with
    Algorithm~\ref{alg:sol} both the current \qdr $\qdrElm$ and the counter
    $\cFun$ as global variables.

    \vspace{0.5em}
    \algprgzer
    \vspace{0.5em}

    Algorithm~\ref{alg:prgzer} simply computes the lift of all the positions
    contained in its input set $\NSet$ (Line~3) and then identifies the new
    positions that will be lifted by either the next application of
    $\prgFun[0]$, namely $\ZSet \cap \mfElm[\qdrElm][-1](0)$, or the subsequent
    application of $\prgFun[+]$, namely $\ZSet \setminus
    \mfElm[\qdrElm][-1](0)$.
    To do so, it first reinitializes the counter for the positions just lifted
    (Line~4) and, then, for each of their incoming moves (Line~5), verifies if
    there exists a new position whose measure needs to be increased.
    The case of an incoming $\PlrSym$-move is trivial (Lines~6-7).
    Therefore, let us consider the opponent player.
    A position $\posElm \in \PosSet[\OppSym]$ needs to be lifted only if
    $\mfElm[\qdrElm](\posElm) < \mfElm[\qdrElm](\uposElm) + \posElm$, for all
    adjacents $\uposElm \in \MovRel(\posElm)$.
    Therefore, we decrement the associated counter (Line~8) every time a
    non-progress move, that previously satisfied the progress condition \wrt the
    unlifted \qdr, is identified.
    The counter reaching zero means that the above condition is satisfied, thus,
    the considered position need to be lifted in the next iteration (Line~9).

    \vspace{0.5em}
    \algdmn
    \vspace{0.5em}

    Algorithm~\ref{alg:dmn} computes the weak quasi dominion $\dmnFun(\qdrElm)$,
    starting from its trigger set $\NSet = \nppFun(\qdrElm)$ that contains all
    the non-progress positions in $\qsiFun(\qdrElm)$.
    The implementation almost precisely follows the functional definition of the
    two operators $\dmnFun$ and $\preFun$, by caring only about keeping the
    whole computation cost linear in the number of incoming moves in each
    position contained in the resulting set.
    To do so, we exploit the same tricks used in the previous procedure
    employing a counter $\dFun$ for the $\OppSym$-positions.
    Note that, $\dFun$ contains a copy of the values in $\cFun$, in order to
    preserve the values in $\cFun$ for the other procedure.

    Finally, Algorithm~\ref{alg:prgpls} implements the procedure described in
    Algorithm~\ref{alg:prg}.
    It first computes the weak quasi dominion $\dmnFun(\qdrElm)$, by calling
    Algorithm~\ref{alg:dmn} (Line~2).
    After that, it identifies its escape positions and the associated forfeit,
    in order to identify the set of best-escape positions that need to be
    lifted (Line~4).
    To do so, we employ a priority queue $\TSet$ based on a min-heap, which will
    contain at most $\SElm$ different forfeit values during the entire
    execution of the algorithm (positions associated with the same forfeit are
    clustered together).
    Obviously, each insert, decrease-key, and remove-min operation on $\TSet$
    will require time $\AOmicron{\log \SElm}$.
    The while-loop at Line~6 simulates the while-loop at Line~2 of
    Algorithm~\ref{alg:prg}, where instructions at Lines~7-9 precisely
    correspond to those at Lines~3-5.
    After the measure update of the best-escape positions in $\ESet$, the
    associated counters in $\cFun$ are reinitialized (line~10).
    At this point, an analysis on the incoming moves of $\ESet$ takes place
    (Line~11).
    For all moves $(\posElm, \uposElm) \in \MovRel$ with $\uposElm \in \ESet$
    and $\posElm \not\in \QSet$, the algorithm performs, at Lines~17-22, almost
    exactly the same operations done by Algorithm~\ref{alg:prgzer} at Lines~6-9.
    The only difference here is that $\OppSym$-positions can only be forced to
    lift their measure if they are not yet contained in the quasi dominion
    $\qsiFun(\qdrElm)$.
    The case $\posElm \in \QSet$, instead, identifies a possible discovering of
    a new escape of the remaining weak quasi dominion (Line~12).
    If $\posElm \in \PosSet[\OppSym]$, this is obviously an escape from $\QSet$,
    thus, it needs to be added to the priority queue $\TSet$ paired with the
    associated best-escape forfeit computed along the move $(\posElm,
    \uposElm)$ (Lines~13-14).
    If $\posElm$ is already contained in $\TSet$, the associated valued is
    decreased, if necessary.
    The case $\posElm \in \PosSet[\PlrSym]$ is more complicated, since a
    $\PlrSym$-position is an escape \iff its current strategy exits from $\QSet$
    and it has no move within $\QSet$ that allows an increase of its measure.
    To do this check, once again, we employ the counter trick, where this time
    we associate with a $\PlrSym$-position in $\dmnFun(\qdrElm)$ the number of
    moves that satisfy the above property (Line~5).
    If the move $(\posElm, \uposElm)$ satisfies the property \wrt the unlifted
    \qdr (\ie, before the lifted of $\uposElm$ occurred), then the corresponding
    counter $\gFun(\posElm)$ is decreased (Line~15).
    When the counter reaches value $0$, the position is necessarily an escape,
    so, it is added to the queue paired with its best possible forfeit
    (Line~16).
    Line~23 calls the $\winFun$ function in order to identify a possible new
    $\PlrSym$-dominion.
    Finally, Lines~24-29 update both the set of positions $\ZSet$ to be lifted
    in the next iteration and the counter $\cFun$, by executing exactly the same
    instructions as those at Lines~18-22 on the moves that reach the dominion
    $\QSet$.
  \end{proof}

  \algprgpls

\end{section}





\end{document}
